\documentclass[aps,10pt,twocolumn,showpacs,pra,citeautoscript,floatfix,superscriptaddress,nofootinbib,longbib]{revtex4-1}

\usepackage{amsmath}
\usepackage{amssymb}
\usepackage{epsfig}
\usepackage{color}
\usepackage{amsthm}
\usepackage[top=50pt,bottom=50pt,left=60pt,right=60pt]{geometry}
\usepackage{mathtools}
\usepackage[hidelinks]{hyperref}
\usepackage{soul}
\usepackage{verbatim}
\usepackage{wasysym}
\usepackage{subcaption} 

\usepackage{tikz}

\usepackage{thm-restate}
\usepackage{thmtools,nameref,cleveref}

\usepackage{dsfont}
\usepackage{bbm} 
\usepackage[normalem]{ulem} 
\usepackage{color} 
\usepackage{braket} 
\usepackage{comment}

\newcommand{\ketbra}[2]{\ket{#1}\!\!\bra{#2}}
\newcommand{\tr}{\textup{tr}}
\newcommand{\be}{\begin{equation}}
\newcommand{\ee}{\end{equation}}
\newcommand{\nn}{{\mathbbm{N}}}
\newcommand{\nnp}{{\mathbbm{N}_{> 0}}}
\newcommand{\nno}{{\mathbbm{N}_{\geq 0}}}
\newcommand{\rr}{{\mathbbm{R}}}

\newcommand{\rro}{{\mathbbm{R}_{\geq 0}}}
\newcommand{\cc}{{\mathbbm{C}}}
\newcommand{\ccr}{{\mathbbm{R}}}
\newcommand{\zz}{{\mathbbm{Z}}}
\newcommand{\me}{\mathrm{e}}
\newcommand{\mi}{\mathrm{i}}
\newcommand{\id}{{\mathbbm{1}}} 
\newcommand{\bo}{\mathcal{O}}

\theoremstyle{definition}

\newtheorem{proposition}{Proposition}
\newtheorem{lemma}{Lemma}

\newtheorem{definition}{Definition}

\newcommand{\Mspace}{\vspace{0.2cm}} 

\def\ba#1\ea{\begin{align}#1\end{align}} 

\newcommand{\app}{\text{appendix}}
\newcommand{\Apps}{\text{Appendices}}
\newcommand{\Meth}{\text{Supplementary}}
\newcommand{\meth}{\text{supplementary}}

\newcommand{\ifree}{counterfactual}
\newcommand{\Ifree}{Counterfactual}

\newcommand{\ifreeoutcome}{counterfactual outcome}
\newcommand{\ifreeoutcomes}{counterfactual outcomes}

\newcommand{\mpw}[1]{\textcolor{blue}{MW: #1}}
\newcommand{\mpwf}[1]{\textcolor{blue}{MW\footnote{\textcolor{blue}{MW: #1}}}}
\newcommand{\hiddenremark}[1]{}

\newcommand\mpwS[1]{MW:{\let\helpcmd\sout\parhelp#1\par\relax\relax} }
\long\def\parhelp#1\par#2\relax{%
	\helpcmd{#1}\ifx\relax#2\else\par\parhelp#2\relax\fi%
}

\usepackage[toc,page]{appendix}

\newcommand{\nocontentsline}[3]{}
\newcommand{\tocless}[2]{\bgroup\let\addcontentsline=\nocontentsline#1{#2}\egroup}

\newcommand{\proj}[1]{|#1\rangle\!\langle #1|}

\newcommand{\Sys}{\textup{S}}
\newcommand{\cl}{\textup{Cl}}
\newcommand{\RL}{\textup{2}}
\newcommand{\SW}{\textup{S}}
\newcommand{\R}{\textup{R}}
\newcommand{\A}{\textup{C}}
\newcommand{\on}{\textup{on}}
\newcommand{\off}{\textup{off}}

\newcommand{\clockw}{{clockwork}}

\newcommand{\rg}{{register}}
\newcommand{\An}{\textup{A}}
\newcommand{\NT}{{N_T}}
\newcommand{\measure}{\textup{m}}

\newcommand{\effect}{\mathtt E}
\newcommand{\sete}{\mathtt e}
\newcommand{\sets}{\mathtt s}
\newcommand{\cf}{\textup{cf}}
\newcommand{\mainr}{\mathcal R}

\newcommand{\squig}{$\scriptsize$\sim$\normalsize$\!}

\newcommand{\rsquigend}{$\scriptsize\rule{.1ex}{0ex}$\rhd$\normalsize$}
\newcommand{\myarrow}[1]{$\stackrel{#1}{\text{$\squig\squig\squig\squig\rsquigend$}}$}

\DeclareMathOperator\erf{erf}
\DeclareMathOperator\sinc{sinc}
\DeclareMathOperator\sign{sign}
\makeatletter
\renewcommand*\env@matrix[1][*\c@MaxMatrixCols c]{%
	\hskip -\arraycolsep
	\let\@ifnextchar\new@ifnextchar
	\array{#1}}
\makeatother

\begin{document}
\title{
Measuring time with stationary quantum clocks\\
}
\begin{abstract}
	Time plays a fundamental role in our ability to make sense of the physical laws in the world around us.  The nature of time has puzzled people | from the ancient Greeks to the present day | resulting in a long running debate between philosophers and physicists alike to whether time needs change to exist (the so-called relatival theory), or whether time flows regardless of change (the so-called substantival theory). One way to decide between the two is to attempt to measure the flow of time with a stationary clock, since if time were substantival, the flow of time would manifest itself in the experiment. Alas, conventional wisdom suggests that in order for a clock to function, it cannot be a static object, thus rendering this experiment seemingly impossible. Here we show, with the aid of \ifree{} measurements, the surprising result that a quantum clock can measure the passage of time even while being switched off, thus lending constructive support for the substantival theory of time.
\end{abstract}	

\author{Sergii Strelchuk}
\affiliation{DAMTP, University of Cambridge, UK}
\author{Mischa P. Woods}
\affiliation{Institute for Theoretical Physics, ETH Zurich, Zurich, Switzerland}


\maketitle 

\section{Introduction}
Time is an essential ingredient in the world we inhabit. Not so surprisingly, it has played a central role in all of our physical theories. 
Yet this role was rather mundane until the advent of modern physics. In particular, out of the three pillars of modern physics | quantum mechanics, special and general relativity | it was only the latter two which forced us to change our preconceptions about the nature of time. Quantum mechanics, on the other hand, while very strange and mysterious in many ways, did not bring any novel insights as far as time is concerned: time is just a parameter which increases in line with any mundane classical clock | just like in Newton's laws. Even more recently, it has been proven, in the context of quantum mechanics, that the time-analogue of Bell's inequality always has a perfectly sound classical explanation~\cite{Temporal_Bell_inequality_Purves,Alistair_causal}.\Mspace

Conversely, in the philosophy of physics domain, the notion of time occupied a prominent role in debates dating back millennia which are still very active today. One of the ongoing debates is whether time necessitates change to exist. The substantival theory of time says that time exists and provides an invisible container in which matter lives, regardless of whether the matter is moving. On the other hand, the relational theory of time says that time is a set of relationships among the events of physical material in space | that is to say, time ceases to exist if matter ceases to change.\Mspace

Views have continuously shifted through the ages: Many believe that Greek atomists such as Democritus thought time was substantival. The first written account dates back to Aristotle, who advocated for a relational theory: ``\emph{But neither does time exist without change;\,\dots}''~\cite{aristotle2006physics}. Newton and Leibniz had opposed views, with Newton on the substantive camp~\cite{sep-newton-stm}, and Leibniz on the relational camp~\cite{sep-leibniz-physics}. Ernst Mach attacked Newton's arguments 
in favour of a more relative theory~\cite{sep-ernst-mach}. Einstein credited Mach's views as being highly influential in his guiding principles when developing his theory of general relativity; although later shifted his stance to a more substantival interpretation of his theory~\cite{Hoefer1994EinsteinsSF}. 

Arguments for and against either theories are still ongoing~\cite{Maudlin2007,rovelli2018order,Smolin2013,barbour2001end}. One of the biggest problems in this debate is the apparent inability to distinguish experimentally between the two scenarios | indeed, it is widely accepted that any clock (or physical matter used as a rudimentary clock), would need to change its state in order to measure the passing of time, hence rendering a clock useless for detecting the substantival's hypothesised passage of time without change.\Mspace

Here we show that nonrelativistic quantum mechanics renders this widely held belief wrong. This is to say, we demonstrate by designing special clocks and adhering to well established interpretations of quantum theory, that it is possible to measure the passage of time between two events even when said clock has been always off, in other words, never evolving. Our result thus allows one to precisely detect the presence of the flow of time without evolution, which we argue provides strong and experimentally testable evidence for the substantival theory of time.

\section{Results}\label{sec:Clock model}
Overview: We will start by considering the simplest systems naturally occurring in nature which can serve as elementary clocks when dynamically evolving. We then show how they can be used to tell the time while being switched off via a more advanced protocol and measurement scheme. This is followed by an analysis of the interpretations used and a demonstration of how this strange phenomena cannot be explained if one assumes an underlying  ``real'' description of nature, using classical variables. Finally, in \cref{sec:discussion}, we explore the implications of our results in relation to the long standing historical debates on the nature of time.\Mspace

\noindent\textbf{Elementary timekeeping systems}:
While we tend to think of timekeeping devices as engineered systems, we can use the dynamics of naturally occurring processes as elementary clocks. The simplest of which can be thought of as an arbitrary state $\ket{\psi}$ at time $t=0$ whose evolution takes it through a sequence of mutually orthogonal states $\ket{\tau_0}\!,\, \ket{\tau_1}\!, \ldots,\, \ket{\tau_\NT}$ at times $\tau_0, \tau_1,\ldots, \tau_\NT$ consecutively. 
If we know that the system is initially in the state $\ket{\psi}$, and we let it evolve for an unknown time $t$ given the promise that $t$ belongs to the set $\{\tau_0, \tau_1, \tau_2,\ldots, \tau_\NT\}$, then we can determine precisely what the elapsed time is by measuring in the basis $\ket{\tau_0}\!,\,\ket{\tau_1}\!,\, \ket{\tau_2}\!, \ldots,\, \ket{\tau_\NT}$, since at said times, the measurement will deterministically allow one to distinguish between said states. To use this timekeeping system in practice, consider two external events, which we will call \emph{1st event} and \emph{2nd event}, with some unknown elapsed time $t\in\{\tau_0, \tau_1, \tau_2,\ldots, \tau_\NT\}$ between them. One initialises the system in the state $\ket{\psi}$ when the 1st event occurs and then waits until the 2nd event occurs, at which point the measurement would be performed; revealing the elapsed time. When a clock is revealed to have been dynamically evolving upon measurement as in this example, we say it is operating in \emph{standard fashion}.\Mspace

This is arguably one of the most elementary types of timekeeping devices one can conceive of. It has no inherently quantum features and should be quite ubiquitous in nature since the states of systems tend to become completely distinguishable if one waits long enough. As a simple illustration of how it could be used, imagine being on a boat at sea heading in to port. There is a lighthouse whose pulsing flashes of light reach you every 3 seconds, but thick fog rolls in obscuring the flashes for some time before clearing. 
By setting the times $\{\tau_n\}$ to coincide with the flashes and choosing the 1st and 2nd events to be flashes before and after the presence of the fog, the clock can deduce the duration of the foggy period.\Mspace

Suppose that, in addition, there is a state denoted $\ket{E}$ which is invariant under time evolution, in other words, it remains in the state $\ket{E}$ at all times. Furthermore, suppose it is orthogonal to $\ket{\psi}$ and its time evolved state at all times. We refer to this state as an \emph{off} state, in contrast to any state which evolves in time, which we refer to as an \emph{on} state. The state $\ket{E}$ could be, for example, an energy eigenstate of the system. If one starts in the state $\ket{E}$ and then measures the state later, clearly there is no possible measurement one can perform which would reveal any information about the elapsed time 
 | it would be analogous to taking the batteries out of a wall clock, and then trying to use it to tell the time.\Mspace
 
 \noindent\textbf{Telling the time when the clock is off:} We will now show, with the aid of some additional stationary ancilla qubits, how quantum mechanics allows one to determine precisely which time $\tau_n$ it is, even though the timekeeping device was never on (in the runs of the experiment in which we determined $\tau_n$). For the purpose of illustration, consider the most basic example of $\NT\!=\!1$, in which the system in question merely has two times $\tau_0, \tau_1$. We relegate the $\NT \geq 2$ version to \app~\ref{sec app: clocks in nature}. In this case we will only need one ancillary state, which we denote $\ket{A}$. It could be, for example, chosen to be another energy eigenstate which is  orthogonal to the other states, or if such a state is not available, we can use an ancilla qubit with a trivial Hamiltonian so as to ensure its states do not evolve. In the latter case, denoting a basis for the ancilla by $\ket{\uparrow}\!,\,\ket{\downarrow}$, the states would then be associated with $\ket{\psi(t)}\equiv \ket{\psi(t)}\ket{\uparrow}$,  $\ket{E}\equiv \ket{E}\ket{\uparrow}$, and the new ancilla state with $\ket{A}= \ket{E}\ket{\downarrow}$ (Here $\ket{\psi(t)}$ denotes $\ket{\psi}$ after evolving for a time $t$).\\


We will use the same retrodictive arguments as in other famous experiments, such as the Elitzur-Vaidman bomb tester or Hardy's paradox. However, let us first describe the protocol, before discussing the interpretation: Initially the system is set to the off state $\ket{E}$. Then, when the 1st event occurs, we apply a unitary $U_0$ such that the system is now in a superposition of on and off, namely, of $\ket{E}$ and $\ket{\psi}$. We use branching notation to indicate the orthogonal branches of the superposition associated with static and dynamical terms (Upper branch is static, lower branch is dynamic). The state reads:

\begin{tikzpicture}
\tikzstyle{level 1}=[level distance=1cm, sibling distance=1.5cm]
\tikzstyle{end} = [circle, minimum width=3pt,fill, inner sep=0pt]
\node (root) {$\ket{E}$}[grow'=right]
child {
	node [end, label=right:{$c\ket{E}$}] {}}
child {
	node [end, label=right:{$s\ket{\psi}$}] {}};
\end{tikzpicture}\label{eq:before time old}
\\where $c=\cos(\theta),$ $s=\sin(\theta)$. One then waits until the 2nd event occurs, at either time $\tau_0$ or $\tau_1$, at which point one applies a judiciously chosen unitary $U_\measure$ followed by immediately measuring in the $\ket{E}$, $\ket{\tau_0}$, $\ket{\tau_1}$, $\ket{A}$ basis, which we will call the measurement basis for short. This completes the protocol | modulo specification of the required constraints on $U_0$ and $U_\measure$. Diagrammatically, up to the point of measurement, this protocol and unitary $U_\measure$ have the form:
\begin{widetext}
\begin{align}
	\begin{tikzpicture}
	\tikzstyle{level 1}=[level distance=1cm, sibling distance=1.5cm]
	\tikzstyle{end} = [circle, minimum width=3pt,fill, inner sep=0pt]
	\node (root) {$\ket{E}$}[grow'=right]
	child {
		node [end, label=right:{$c\ket{E}$ \myarrow{\tau_0} $c\ket{E}\overset{U_\measure\,}{\xrightarrow{\hspace*{0.7cm}}}\,\, cA_1^0\ket{E}+cA_2^0\ket{A}+cA_3^0\ket{\tau_0}+cA_4^0\ket{\tau_1}$}] {}}
	child {	node [end, label=right:{$s\ket{\psi}$ \myarrow{\tau_0} $s\ket{\tau_0}\overset{U_\measure\,}{\xrightarrow{\hspace*{0.7cm}}}\,\, \qquad\quad\,\,-cA_2^0\ket{A}+sA_3^1\ket{\tau_0}+sA_4^1\ket{\tau_1}$} ] {}};
	\end{tikzpicture}\label{eq:before time}
\end{align}
at time $\tau_0$ and
\begin{align}
	\begin{tikzpicture}
	\tikzstyle{level 1}=[level distance=1cm, sibling distance=1.5cm]
	\tikzstyle{end} = [circle, minimum width=3pt,fill, inner sep=0pt]
	\node (root) {$\ket{E}$}[grow'=right]
	child {
		node [end, label=right:{$c\ket{E}$ \myarrow{\tau_1} $c\ket{E}\overset{U_\measure\,}{\xrightarrow{\hspace*{0.7cm}}}\,\,cA_1^0\ket{E}+cA_2^0\ket{A}+cA_3^0\ket{\tau_0}+cA_4^0\ket{\tau_1}$}] {}}
	child {	node [end, label=right:{$s\ket{\psi}$ \myarrow{\tau_1} $s\ket{\tau_1}\overset{U_\measure\,}{\xrightarrow{\hspace*{0.7cm}}}\,\, -cA_1^0\ket{E}\qquad\quad\,\,+s{A_3^1}'\ket{\tau_0}+s{A_4^1}'\ket{\tau_1}$} ] {}};
	\end{tikzpicture}\label{eq:after time}
\end{align}
\end{widetext}

\noindent at time $\tau_1$, where $|c A_1^0|>0$, $|c A_2^0|>0$ and other amplitudes are arbitrary. Squiggly arrows represent the passing of an amount of time $\tau_0$ or $\tau_1$, while straight arrows the application of $U_\measure$.\Mspace

Suppose the $E$ outcome is obtained.\footnote{We use the convention in which ``outcome $x$'' means that the post-measurement state is $\ket{x}$.} At time $\tau_0$, a non zero amplitude associated with $\ket{E}$ exists only for the off branch, while for time $\tau_1$, this amplitude from the off branch cancels with an on branch amplitude due to destructive interference. As such, when $E$ is obtained, we can deduce that the clock has collapsed to a branch of the wave function that was always off \emph{and} that the unknown time $t$ must be $\tau_0$. Similarly, when outcome $A$ is obtained, we deduce that the clock was collapsed to a branch of the wave function that was always off \emph{and} that the time $t$ must be $\tau_1$. However, if outcome $\tau_0$ or $\tau_1$ is obtained, we cannot conclude anything useful: the clock may have been on and we cannot deduce whether the time is $\tau_0$ or $\tau_1$.\Mspace

Hence whenever outcome $E$ or $A$ is obtained, we can deduce what the time is, even though the system was always off, that is to say, always in a stationary state (since $\ket{E}$ and $\ket{A}$ are orthogonal to all dynamical branches whenever $E$, $A$ have a non zero probability of being obtained).\Mspace

 \noindent\textbf{{\bf Interpretation:}} One may object that the system has certainly been in a superposition of on and off, and in this sense, has actually been evolving in time. However, the outcomes $E$ and $A$ can only arise via an always off branch of the wave function so if it is seen then we have been confined to a part of the total quantum state in which the system is always stationary. Furthermore, this interpretation is independent of the basis used to represent the state during the different stages of the protocol | we represented it in the measurement basis purely for convenience.\Mspace


Our explanation of the results has relied on a certain kind of retrodictive interpretation of quantum mechanics. This interpretation can be most readily captured by Schr\"odinger's proverbial cat experiment: one starts with an alive cat, denoted $\ket{\text{alive}}$, which when put in the closed box, takes on the form $\frac{1}{\sqrt{2}} \left( \ket{\text{alive}}+\ket{\text{dead}} \right)$. Suppose that upon opening the box, we make a measurement of alive versus dead, and obtain the outcome ``alive''. Then in the conventional collapse of the wave function formalism, the state of the cat changes discontinuously into $\ket{\text{alive}}$ and the $\ket{\text{dead}}$ component ceases to have any further physical existence or further consequence | see \cref{fig:clock and cat}. This is a form of retrodiction, since it says that upon observing an alive cat we collapse to a state in which the cat was \emph{always} alive, since the dead and alive branches were orthogonal at all times. In our setup, this retrodiction takes on the form of a so-called \emph{\ifree{} measurement},\footnote{Also known as ``interaction-free measurement'' or ``quantum interrogation''.} since when outcome $E$ or $A$ is obtained, it allows one to deduce properties of a dynamical clock by reasoning counterfactually, while collapsing the wave function to a state where those eventualities never actually took place.  \Ifree{} measurements were first discovered by Elitzur and Vaidman in their famous bomb tester experiment~\cite{Elitzur1993,Vaidman2003,penrose1994shadows}, and later applied to other scenarios including counterfactual computation~\cite{Jozsa1999,Mitchison2001,2010.00623,PhysRevLett.74.4763}, and experimentally verified~\cite{Weinfurter95,PhysRevLett.83.4725,White:98}. Prior examples have played out in space, while ours is the first protocol to do so in time. Another example of retrodiction and \ifree{} measurement in quantum mechanics is the celebrated Hardy paradox~\cite{hardy92}. Weakly measuring the distinct branches can provide theoretical and experimental evidence for the validity of such retrodictive interpretations~\cite{Aharonov2002,PhysRevLett.102.020404,Yokota_2009}.\Mspace

We will call our clock when operated in this retrodictive fashion a \emph{counterfactual clock}. Likewise, the measurement outcomes of interest in a \ifree{} measurement ($E$ and $A$), are referred to as \ifreeoutcomes. If the above protocol is implemented and a \ifreeoutcome{} is obtained, we say that the clock was \emph{always off}. In~\cref{fig:clock and bomb} we compare one instance of the Elitzur and Vaidman bomb test with one instance of the counterfactual clock.\Mspace

Our analysis thus far has been presented \emph{forward in time}, namely starting from the initial state $\ket{E}$ and analysing the different steps of the protocol until post-selecting on a \ifreeoutcome. We can also examine the \emph{backwards in time} process where we start with one of the post-selected states, and examine how it evolves backwards in time to the initial state $\ket{E}$. Whenever analysing paradoxes involving post-selection, it is prudent to check that the paradox still exits when the problem is analysed backwards in time, since this can reveal any ``hidden'' assumptions inadvertently made due to post-selecting. It is important to do so in the case of counterfactual experiments, as shown in \cite{Vaidman_PostSelection07}. We prove in \meth~\cref{sec:backwards in time analysis} that when the \ifreeoutcomes~are analysed backwards in time, the pre-selected wave function is collapsed onto an always off branch. This guarantees that our \ifree{} measurements have been properly implemented.

Of course, one can apply alternative interpretations of quantum mechanics. In the many-worlds interpretation of Schr\"odinger's cat, upon measurement, reality splits into two parallel worlds | one in which the component $\ket{\text{alive}}$ prevails and the cat is alive, the other, in which the component $\ket{\text{dead}}$ prevails and the cat is dead. Conversely, the explanation of the counterfactual clock also changes if one uses this latter interpretation. Specifically, in the language of many-worlds, whenever the outcome $E$ or $A$ is obtained, we will be living in a world in which the system never dynamically evolved (i.e. never in an on state) yet in this world we learn information about the dynamical states (i.e. the on branch), namely what the time is. The fact that the system was evolving in another world is of no consequence to \emph{us}.\Mspace
\captionsetup[figure]{justification=justified}
\captionsetup[subfigure]{labelfont=bf,justification=raggedright}
\renewcommand{\thesubfigure}{\Alph{subfigure}}
\begin{figure}[h!]
		\includegraphics[width=1\linewidth]{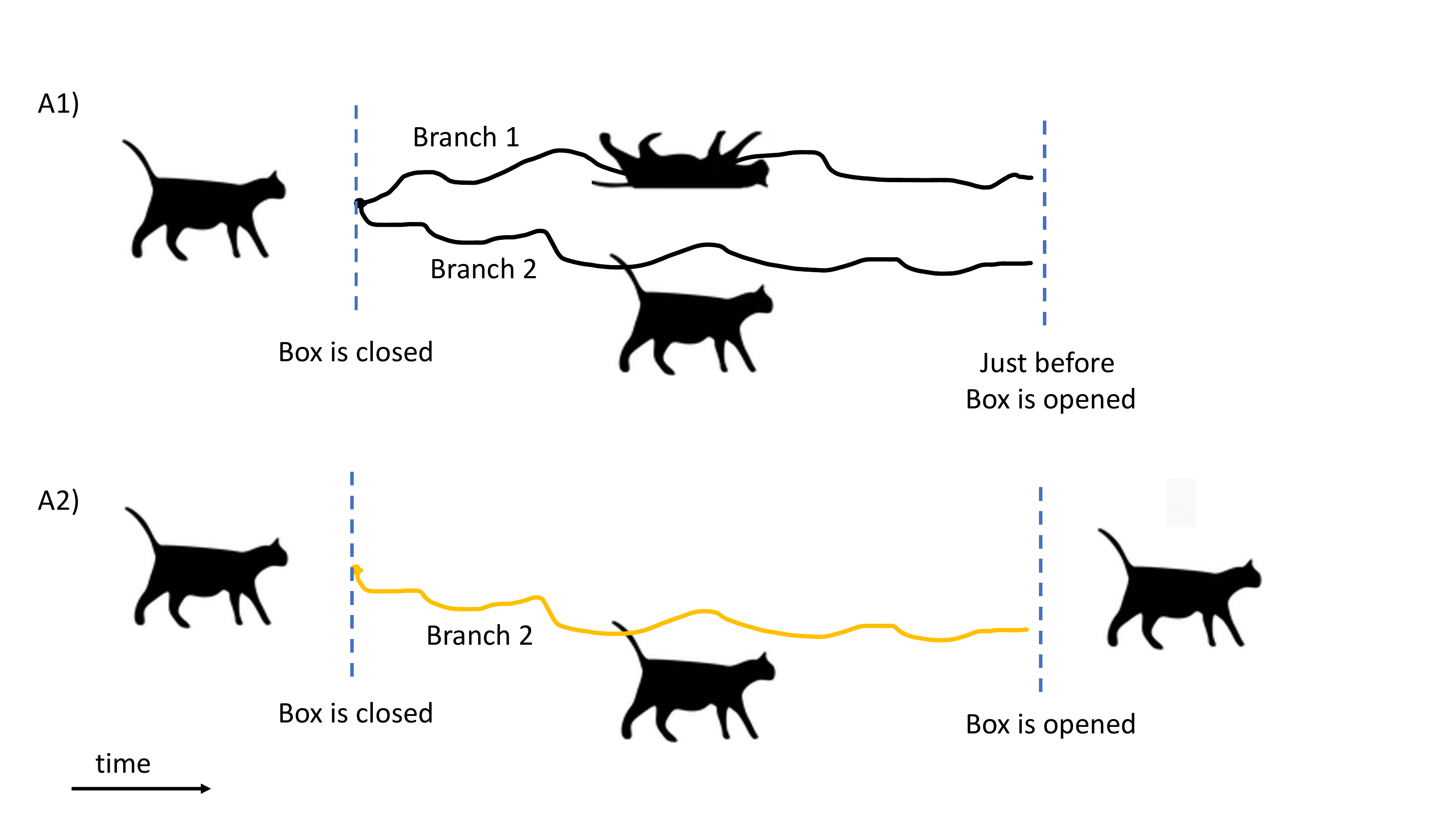}
		\caption{\small{{\bf Schr\"odinger's cat.}
	{\bf A1)} Schr\"odiger's cat starts out alive. At the 1st dotted line, the box is closed and the cat is in a superposition of dead and alive until just before the box is opened (2nd dotted line). {\bf A2)} The box has now been opened, causing the cat to be measured. Here we show the case in which we see an alive cat. The history of the cat from before the box was closed to present is now set in stone, with the dead cat branch of the wave function ceasing to have ever existed.
	}}\label{fig:clock and cat}
\end{figure}

\begin{figure}[htb!]
	
	\begin{subfigure}[b]{0.517\textwidth}
		\includegraphics[width=1\linewidth]{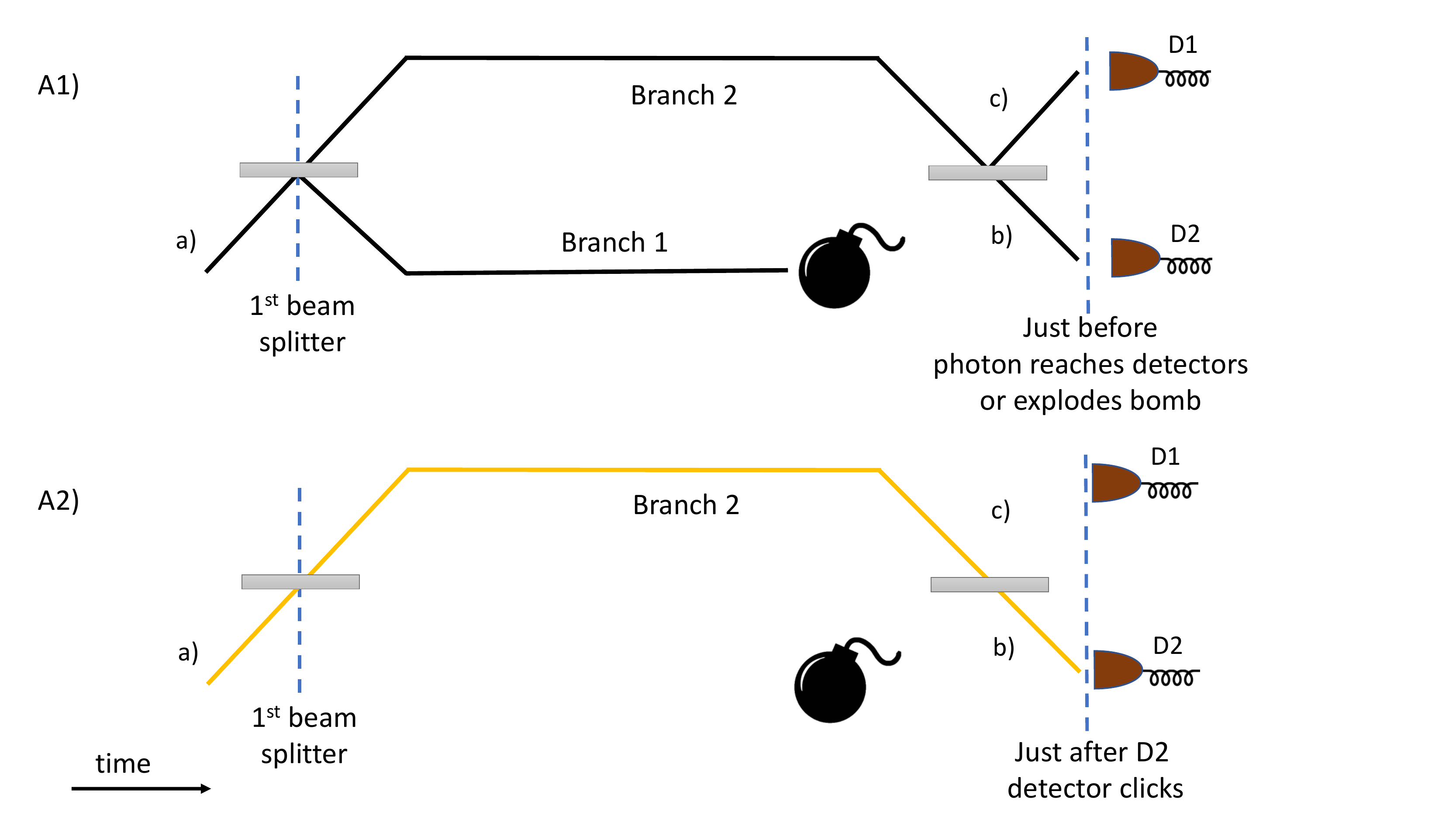}
		\caption{Elizur-Vaidman bomb test} 	
	\end{subfigure}
	
	\begin{subfigure}[b]{0.516\textwidth}
		\includegraphics[width=1\linewidth]{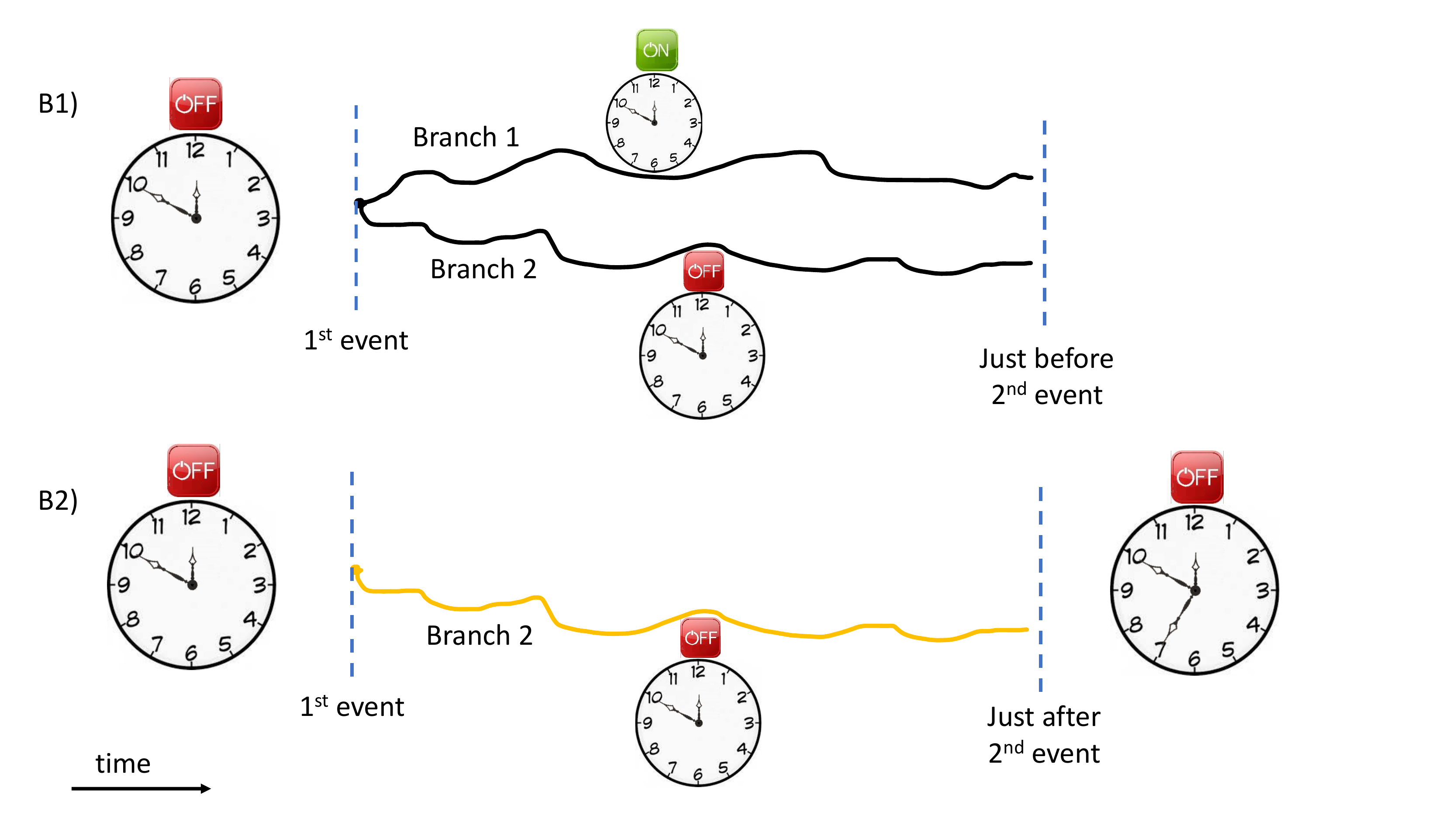}
		\caption{The counterfactual clock} 	
	\end{subfigure} 
	\caption[justification=raggedright]{{\bf Comparison between the Elizur-Vaidman bomb test and the counterfactual clock}\\
\small{		{\bf A1)} 
	Mach-Zehnder interferometer with 50/50 beam splitters and a photosensitive bomb placed in the bottom arm. Photon trajectories are black lines. A photon enters at a) and splits into a superposition travelling down the upper and lower arms upon interacting with the 1st beam splitter. The photon in the bottom arm is coherently absorbed by the bomb, hence exploding. The photon in the~upper~arm~encounters~the~2nd~beam~splitter~and~coherently~splits~into~trajectories~b),~c)~without~bomb~exploding. 
	
	{\bf A2)}	 Continuing on from A1), a measurement is now performed via the photodetectors D1 and D2. There are 3 possible outcomes: If photon took path c), D1 clicks. If photon took path b), D2 clicks. If bomb exploded, photon is lost and neither photodetector clicks. Fig A2 depicts the scenario that D2 clicks (the \ifreeoutcome). Upon clicking, the wave function collapses into an eventuality where the bomb never exploded.
	{\bf B1)} The counterfactual clock starts out off. At the 1st dotted line, the 1st event occurs and the clock is put into a superposition of on and off until just before the 2nd event occurs (2nd dotted line). {\bf B2)} The 2nd event now occurs, causing the clock to be measured. Here we show the case in which we see a \ifreeoutcome. The history of the clock from before the 1st event occurred to present is not set in stone, with the on clock branch of the wave function ceasing to have ever existed.
	{\bf Comparison between A and B}: Both use a \ifree{} measurement to avoid the eventualities of branch 1 of the wave function occurring when obtaining certain outcomes.} 
		
	}\label{fig:clock and bomb}
\end{figure}

The role of the ancilla is more subtle than may appear at first sight: indeed, one might envisage a situation with no ancilla and where we interpret outcomes $E$ and $\tau_1$ as \ifreeoutcomes{} by choosing the unitary $U_\measure$ differently to as presented.
Notwithstanding, we prove in \app~\ref{sec:Optimality and achievability proofs}, that there does not exist a counterfactual clock for which there is no ancilla and the probabilities of the two \ifreeoutcomes{} are both non zero. This implies a more fundamental role for the ancilla state: by including it, we can reach a larger class of unitary transformations which are necessary for the counterfactual clock to function. This is important for understanding what is the minimal model for making the appropriate comparison with classical descriptions, which we will do shortly.\Mspace 

In this set-up, the optimal protocol is the one which maximises the probabilities of obtaining either measurement outcomes $E$ or $A$ over unitaries $U_\measure$ and branching amplitudes $c,s$. In the case in which both outcomes are equally likely, $|c A_1^0|^2= |c A_2^0|^2$; we prove in \app~\ref{sec:Optimality and achievability proofs} that there exists a protocol which assigns probability $1/6$ to obtaining the outcome $E$ at time $\tau_0$ and probability $1/6$ to outcome $A$ at time $\tau_1$. For later reference, we will refer to the sum of the probabilities corresponding to the \ifreeoutcomes, the \emph{total \ifree{} probability}. In this case, it is~$1/3$.\Mspace

One may wonder whether a way out of this conundrum is via the existence of an underlying description respecting realism, that is, a theory where the relevant physical properties have well defined values at all times prior to the measurement in our protocol. For example, such a theory could have hidden variables which evolve dynamically when the clock is supposedly off. We prove that there cannot exist a non-contextual ontic model for the relevant degrees of freedom for the above counterfactual clock with probability $1/6$ for each \ifreeoutcome. This rules out all classical models under natural assumptions; see \meth~\ref{sec:non contextual proof} for details.\Mspace

\noindent{\bf Engineered clocks}: in this section we will discuss a more advanced type of engineered clock. This is more of an aside, and one can skip this section heading straight to the discussion in \cref{sec:discussion}, if only interested in interpretation and consequences for the nature of time.\Mspace

While the protocol thus far discussed allows one to use the most elementary of systems as a timekeeping device, it is limited in that it can only distinguish between times $\tau_0,$ $\tau_1, \ldots, \tau_\NT$ given the guarantee that the elapsed time between the two events coincides with one of these times. This limitation is due to an elementary clock design | indeed, this feature was present even when the clock was run in standard fashion. This limitation means that while it can be used for some applications (recall, e.g., the lighthouse example), it cannot be used for others, such as determining whether the winner of a race set a new record, because there in no reason to believe that the interval between starting the race and finishing it, will belong to any prior chosen set $\{\tau_0,$ $\tau_1, \ldots, \tau_\NT\}$. To overcome these restrictions, we demonstrate the existence of a counterfactual clock which has \ifreeoutcomes{} with a protocol identical to the one described prior, modulo a few distinctions: 1) The elapsed time between the two events can be arbitrary. 2) The unitaries $U_0$ and $U_\measure$ are chosen differently. 3) The projective measurement is in a different basis.\Mspace

The physically significant difference in the protocol is clearly 1). Furthermore, it enjoys the analogous physical interpretation when the \ifreeoutcomes{} are obtained to that described in the previous case above. Moreover, the Hamiltonian we use in the construction is time independent. This is important since otherwise the protocol would likely require an external timekeeping device for its implementation | rendering the entire experiment pointless. Another physically relevant feature of the engineered clock, is that when the clock is on, it cannot change instantaneously between distinct states representing the distinct ticks of the clock. This introduces a small error | regardless of whether it is run counterfactually or in standard fashion. The error is analogous to the situation we face with a wall clock: the second hand (i.e. the hand the records the passing of seconds) cannot instantaneously transition between one clock face marking and the next; therefore, if you happen to read the clock around the milliseconds interval in which the second hand is transitioning, you will likely be off by a second.\Mspace

In particular, when the clock is run counterfactually, and one of the \ifreeoutcomes{} is obtained, this error means that there is a small probability of a false positive: the quantum destructive interference is not perfect, leading to the possibility that the clock was actually on after all. Luckily, this error can be made arbitrarily small  | although at the expense of a decrease in the probability of obtaining a \ifreeoutcome. Notwithstanding, reasonable probabilities can be obtained. For example, in the case of one tick, and a total \ifree{} probability of $1/12$, the probability that the clock was on when one of the \ifreeoutcomes{} is obtained is of order $10^{-14}$. See \Meth~\ref{sec:Engineered clocks} for details.\Mspace

Physically speaking, this tiny error is due to imperfect destructive interference between on and off breaches. Such events are key to all quantum experiments using \ifree{} measurements and due to engineering imperfections, will always be present. For example, in the quantum counterfactual experiments involving photons~\cite{Weinfurter95,White:98,Hosten2006,PhysRevLett.83.4725,PhysRevA.59.2322}, small deviations in the reflectivity of the beam splitters lead to ports with photodetectors which were not completely dark.

\section{Discussion}\label{sec:discussion}

With the aid of \ifree{} measurements, we have shown that quantum mechanics allows one to predict the time passed between two events, even though no dynamics has taken place between the two events. The most elementary form of this clock presented in the main text should be obtainable from ubiquitous systems in nature (a four dimensional system with two stationary and two non-stationary states), while the more advanced clocks providing a larger number of distinguishable times (large $\NT$) require more ancillas and should be realisable with state of the art current technologies such as those developed for quantum computation.\Mspace

While it is readily clear that a classical clock cannot tell the time without dynamical evolution, it could have been that our quantum clock had hidden variables which were changing when the clock was supposedly off | we ruled out such a possibility under reasonable assumptions.\Mspace

We will now discuss the implications for the relational and substantival theories of time which were presented in the introduction. As explained previously, one of the reasons why the debate has been ongoing for more that two thousand years, is due to the lack of the possibility to experimentally distinguish between the two theories due to an apparent inability to measure the flow of time without evolution. Our results can remedy this dilemma by the following  experiment:\Mspace

Alice has a random number generator which outputs heads or tails. If she finds heads (tails), she sends two light pulses to Bob separated by an elapsed time $\tau_0$ ($\tau_1$). Bob possesses a counterfactual clock in the energy eigenstate $\ket{E}$, which can distinguish between two elapsed times $\tau_0$ and $\tau_1$, but does not know the outcome of the random number generator. Alice tunes her two light pulses so that when they come into contact with Bob's clock, the 1st pulse implements $U_0$ while the 2nd implements $U_\measure$. Bob measures his clock in the basis $\ket{E}$, $\ket{\tau_0}$, $\ket{\tau_1}$, $\ket{A}$  upon seeing the second light pulse. Alice repeats this experiment until Bob reports to her that he found a \ifreeoutcome. There is a second system near Bob which remains in an energy eigenstate during the entire experiment and does not interact with the light pulses; bob will use it for reference purposes only. Consider the last run of the experiment: the clock did not change between the two light pulses (relative to the second system) but was nevertheless responsible for the detection and quantification of an elapsed time which is identical to what a clock which was dynamically evolving would have predicted (had one been present). We argue that in a relative theory of time, from the perspective of the clock relative to the second system, any value for the elapsed time would have to be arbitrary, since time would not exist from it's perspective. Contrarily, the experiment is perfectly consistent with a substantival theory of time, since in this case, from the clock's perspective, there is a well-defined elapsed time between the events, as time passes independently of the clock's state of motion.\Mspace
\newpage
We can formalise this observation via the following assumptions and theorem: 
\begin{itemize}
	\item [(A)] \Ifree{} measurements exist. \vspace{-0.3cm}
	\item [(B)] If time can be measured via a clock which is always off, then time is substantival.
\end{itemize}

\begin{restatable}{theorem}{thrmSubstantival}\label{thm:time is substantival}
Time is substantival if assumptions (A) and (B) hold.
\end{restatable}
\!\!We prove the theorem in \meth{}~\ref{Proof of thm:time is substantiva}. If (A) is false, then the other phenomena which uses \ifree{} measurements, such as counterfactual computing or the bomb test, would require serious re-evaluation. Similarly, if (A) holds yet (B) is false, then arguably, the very meaning of a substantival theory of time requires substantial review.\Mspace

Finally, we conclude with a discussion how one might circumvent or modify our conclusions. Our results relied on the validity of the \ifree{} measurement interpretation in quantum mechanics, and we discussed how in the alternative interpretation of many-worlds, when one of the \ifreeoutcomes{} is obtained, while \emph{we} would be in a world where the clock was always off, there would be another world where it would be on. 
Regarding the theory of time debate, this interpretation would still lead to difficulties for a relatival theory of time, since in our thought experiment the second system remains in Bob's world, so when he obtains the counterfactual outcome, there would have been no relative change between the clock and the second system in Bob's world, yet Bob would know how much time had passed in his world from the measurement outcome. 
 Another well-known alternative interpretation of quantum theory is Bohmian mechanics. Here there is always a real state associated with the wave function even when not observed. It is a contextual interpretation of quantum mechanics, and thus is not ruled out via our no-go results regarding a classical interpretation.

As alluded to in the introduction, our beliefs about the nature of time have been highly influential in the development of our physical theories. Going forward, perhaps one of the key ingredients to finding a correct theory of quantum gravity is contingent of asserting an appropriate belief about the nature of time. We hope that our work will bring much needed attention to this debate. 

\onecolumngrid
\acknowledgements

We are grateful to Caslav Brukner, Rob Spekkens and Victor Gitton for discussions.
M.W. acknowledges funding from the Swiss National Science Foundation (AMBIZIONE Fellowship, No. PZ00P2 179914) in addition to supported from the Swiss National Science Foundation via the National Centre for Competence in Research “QSIT”. S.S. acknowledges support from the QuantERA ERA-NET Cofund in Quantum
Technologies implemented within the European Union's
Horizon 2020 Programme (QuantAlgo project), and administered through EPSRC Grant No. EP/R043957/1,
and support from the Royal Society University Research Fellowship scheme.
M.W. and S.S. both acknowledge funding from the Royal Society International Exchanges program (grant No. IES\textbackslash R1\textbackslash1801060).
\Mspace
\begin{center}
	\textbf{{\small AUTHOR CONTRIBUTIONS}}
\end{center}
M.W. derived the protocols and wrote the paper. S.S. implemented the numerical proof for the non existence of a noncontextual ontic model. Both authors contributed to ideas and discussions.

\bibliography{References}

\begin{thebibliography}{39}%
\makeatletter
\providecommand \@ifxundefined [1]{%
 \@ifx{#1\undefined}
}%
\providecommand \@ifnum [1]{%
 \ifnum #1\expandafter \@firstoftwo
 \else \expandafter \@secondoftwo
 \fi
}%
\providecommand \@ifx [1]{%
 \ifx #1\expandafter \@firstoftwo
 \else \expandafter \@secondoftwo
 \fi
}%
\providecommand \natexlab [1]{#1}%
\providecommand \enquote  [1]{``#1''}%
\providecommand \bibnamefont  [1]{#1}%
\providecommand \bibfnamefont [1]{#1}%
\providecommand \citenamefont [1]{#1}%
\providecommand \href@noop [0]{\@secondoftwo}%
\providecommand \href [0]{\begingroup \@sanitize@url \@href}%
\providecommand \@href[1]{\@@startlink{#1}\@@href}%
\providecommand \@@href[1]{\endgroup#1\@@endlink}%
\providecommand \@sanitize@url [0]{\catcode `\\12\catcode `\$12\catcode
  `\&12\catcode `\#12\catcode `\^12\catcode `\_12\catcode `\%12\relax}%
\providecommand \@@startlink[1]{}%
\providecommand \@@endlink[0]{}%
\providecommand \url  [0]{\begingroup\@sanitize@url \@url }%
\providecommand \@url [1]{\endgroup\@href {#1}{\urlprefix }}%
\providecommand \urlprefix  [0]{URL }%
\providecommand \Eprint [0]{\href }%
\providecommand \doibase [0]{http://dx.doi.org/}%
\providecommand \selectlanguage [0]{\@gobble}%
\providecommand \bibinfo  [0]{\@secondoftwo}%
\providecommand \bibfield  [0]{\@secondoftwo}%
\providecommand \translation [1]{[#1]}%
\providecommand \BibitemOpen [0]{}%
\providecommand \bibitemStop [0]{}%
\providecommand \bibitemNoStop [0]{.\EOS\space}%
\providecommand \EOS [0]{\spacefactor3000\relax}%
\providecommand \BibitemShut  [1]{\csname bibitem#1\endcsname}%
\let\auto@bib@innerbib\@empty
\bibitem [{\citenamefont {Purves}\ and\ \citenamefont
  {Short}(2021)}]{Temporal_Bell_inequality_Purves}%
  \BibitemOpen
  \bibfield  {author} {\bibinfo {author} {\bibfnamefont {T.}~\bibnamefont
  {Purves}}\ and\ \bibinfo {author} {\bibfnamefont {A.~J.}\ \bibnamefont
  {Short}},\ }\href@noop {} {\enquote {\bibinfo {title} {Quantum theory cannot
  violate a causal inequality},}\ } (\bibinfo {year} {2021}),\ \Eprint
  {http://arxiv.org/abs/arXiv:2101.09107} {arXiv:2101.09107} \BibitemShut
  {NoStop}%
\bibitem [{\citenamefont {Wechs}\ \emph {et~al.}(2021)\citenamefont {Wechs},
  \citenamefont {Dourdent}, \citenamefont {Abbott},\ and\ \citenamefont
  {Branciard}}]{Alistair_causal}%
  \BibitemOpen
  \bibfield  {author} {\bibinfo {author} {\bibfnamefont {J.}~\bibnamefont
  {Wechs}}, \bibinfo {author} {\bibfnamefont {H.}~\bibnamefont {Dourdent}},
  \bibinfo {author} {\bibfnamefont {A.~A.}\ \bibnamefont {Abbott}}, \ and\
  \bibinfo {author} {\bibfnamefont {C.}~\bibnamefont {Branciard}},\ }\href@noop
  {} {\enquote {\bibinfo {title} {Quantum circuits with classical versus
  quantum control of causal order},}\ } (\bibinfo {year} {2021}),\ \Eprint
  {http://arxiv.org/abs/arXiv:2101.08796} {arXiv:2101.08796} \BibitemShut
  {NoStop}%
\bibitem [{\citenamefont {{Aristotle}}(2006)}]{aristotle2006physics}%
  \BibitemOpen
  \bibfield  {author} {\bibinfo {author} {\bibnamefont {{Aristotle}}},\ }\href
  {http://classics.mit.edu/Aristotle/physics.4.iv.html} {\emph {\bibinfo
  {title} {Physics, \textup{Book IV, Chapters 10 and 11.}}}}\ (\bibinfo
  {publisher} {Neeland Media},\ \bibinfo {year} {2006})\ \bibinfo {note}
  {{Translated by R.P. Hardie and R.K. Gaye}}\BibitemShut {NoStop}%
\bibitem [{\citenamefont {Rynasiewicz}(2014)}]{sep-newton-stm}%
  \BibitemOpen
  \bibfield  {author} {\bibinfo {author} {\bibfnamefont {R.}~\bibnamefont
  {Rynasiewicz}},\ }in\ \href
  {https://plato.stanford.edu/archives/sum2014/entries/newton-stm/} {\emph
  {\bibinfo {booktitle} {The {Stanford} Encyclopedia of Philosophy}}},\
  \bibinfo {editor} {edited by\ \bibinfo {editor} {\bibfnamefont {E.~N.}\
  \bibnamefont {Zalta}}}\ (\bibinfo  {publisher} {Metaphysics Research Lab,
  Stanford University},\ \bibinfo {year} {2014})\ \bibinfo {edition} {summer
  2014}\ ed.\BibitemShut {Stop}%
\bibitem [{\citenamefont {McDonough}(2019)}]{sep-leibniz-physics}%
  \BibitemOpen
  \bibfield  {author} {\bibinfo {author} {\bibfnamefont {J.~K.}\ \bibnamefont
  {McDonough}},\ }in\ \href
  {https://plato.stanford.edu/archives/fall2019/entries/leibniz-physics/}
  {\emph {\bibinfo {booktitle} {The {Stanford} Encyclopedia of Philosophy}}},\
  \bibinfo {editor} {edited by\ \bibinfo {editor} {\bibfnamefont {E.~N.}\
  \bibnamefont {Zalta}}}\ (\bibinfo  {publisher} {Metaphysics Research Lab,
  Stanford University},\ \bibinfo {year} {2019})\ \bibinfo {edition} {fall
  2019}\ ed.,\ \bibinfo {note} {{S}ection 5: Leibniz on the Space and
  Time}\BibitemShut {NoStop}%
\bibitem [{\citenamefont {Pojman}(2020)}]{sep-ernst-mach}%
  \BibitemOpen
  \bibfield  {author} {\bibinfo {author} {\bibfnamefont {P.}~\bibnamefont
  {Pojman}},\ }in\ \href
  {https://plato.stanford.edu/archives/win2020/entries/ernst-mach/} {\emph
  {\bibinfo {booktitle} {The {Stanford} Encyclopedia of Philosophy}}},\
  \bibinfo {editor} {edited by\ \bibinfo {editor} {\bibfnamefont {E.~N.}\
  \bibnamefont {Zalta}}}\ (\bibinfo  {publisher} {Metaphysics Research Lab,
  Stanford University},\ \bibinfo {year} {2020})\ \bibinfo {edition} {winter
  2020}\ ed.\BibitemShut {Stop}%
\bibitem [{\citenamefont {Hoefer}(1994)}]{Hoefer1994EinsteinsSF}%
  \BibitemOpen
  \bibfield  {author} {\bibinfo {author} {\bibfnamefont {C.}~\bibnamefont
  {Hoefer}},\ }\href {\doibase https://doi.org/10.1016/0039-3681(94)90056-6}
  {\bibfield  {journal} {\bibinfo  {journal} {Studies in History and Philosophy
  of Science Part A}\ }\textbf {\bibinfo {volume} {25}},\ \bibinfo {pages}
  {287} (\bibinfo {year} {1994})}\BibitemShut {NoStop}%
\bibitem [{\citenamefont {Maudlin}(2007)}]{Maudlin2007}%
  \BibitemOpen
  \bibfield  {author} {\bibinfo {author} {\bibfnamefont {T.}~\bibnamefont
  {Maudlin}},\ }in\ \href {\doibase 10.1093/acprof:oso/9780199218219.003.0005}
  {\emph {\bibinfo {booktitle} {The Metaphysics Within Physics}}}\ (\bibinfo
  {publisher} {Oxford University Press},\ \bibinfo {year} {2007})\ pp.\
  \bibinfo {pages} {104--142}\BibitemShut {NoStop}%
\bibitem [{\citenamefont {Rovelli}(2018)}]{rovelli2018order}%
  \BibitemOpen
  \bibfield  {author} {\bibinfo {author} {\bibfnamefont {C.}~\bibnamefont
  {Rovelli}},\ }\href
  {https://www.penguin.co.uk/books/301539/the-order-of-time/9780141984964.html}
  {\emph {\bibinfo {title} {The Order of Time}}}\ (\bibinfo  {publisher}
  {Penguin Books Limited},\ \bibinfo {year} {2018})\BibitemShut {NoStop}%
\bibitem [{\citenamefont {Smolin}(2013)}]{Smolin2013}%
  \BibitemOpen
  \bibfield  {author} {\bibinfo {author} {\bibfnamefont {L.}~\bibnamefont
  {Smolin}},\ }\href
  {https://www.hmhbooks.com/shop/books/time-reborn/9780544245594} {\emph
  {\bibinfo {title} {{Time Reborn: From the Crisis in Physics to the Future of
  the Universe}}}}\ (\bibinfo  {publisher} {Houghton Mifflin Harcourt},\
  \bibinfo {year} {2013})\BibitemShut {NoStop}%
\bibitem [{\citenamefont {Barbour}(2001)}]{barbour2001end}%
  \BibitemOpen
  \bibfield  {author} {\bibinfo {author} {\bibfnamefont {J.}~\bibnamefont
  {Barbour}},\ }\href@noop {} {\emph {\bibinfo {title} {The End of Time: The
  Next Revolution in Physics}}}\ (\bibinfo  {publisher} {Oxford University
  Press, USA},\ \bibinfo {year} {2001})\BibitemShut {NoStop}%
\bibitem [{\citenamefont {Elitzur}\ and\ \citenamefont
  {Vaidman}(1993)}]{Elitzur1993}%
  \BibitemOpen
  \bibfield  {author} {\bibinfo {author} {\bibfnamefont {A.~C.}\ \bibnamefont
  {Elitzur}}\ and\ \bibinfo {author} {\bibfnamefont {L.}~\bibnamefont
  {Vaidman}},\ }\href {\doibase 10.1007/BF00736012} {\bibfield  {journal}
  {\bibinfo  {journal} {Foundations of Physics}\ }\textbf {\bibinfo {volume}
  {23}},\ \bibinfo {pages} {987} (\bibinfo {year} {1993})}\BibitemShut
  {NoStop}%
\bibitem [{\citenamefont {Vaidman}(2003)}]{Vaidman2003}%
  \BibitemOpen
  \bibfield  {author} {\bibinfo {author} {\bibfnamefont {L.}~\bibnamefont
  {Vaidman}},\ }\href {\doibase 10.1023/a:1023767716236} {\bibfield  {journal}
  {\bibinfo  {journal} {Foundations of Physics}\ }\textbf {\bibinfo {volume}
  {33}},\ \bibinfo {pages} {491} (\bibinfo {year} {2003})}\BibitemShut
  {NoStop}%
\bibitem [{\citenamefont {Penrose}(1994)}]{penrose1994shadows}%
  \BibitemOpen
  \bibfield  {author} {\bibinfo {author} {\bibfnamefont {R.}~\bibnamefont
  {Penrose}},\ }\href {https://books.google.st/books?id=gDbOAK89tmcC} {\emph
  {\bibinfo {title} {Shadows of the Mind: A Search for the Missing Science of
  Consciousness}}}\ (\bibinfo  {publisher} {Oxford University Press},\ \bibinfo
  {year} {1994})\BibitemShut {NoStop}%
\bibitem [{\citenamefont {Jozsa}(1998)}]{Jozsa1999}%
  \BibitemOpen
  \bibfield  {author} {\bibinfo {author} {\bibfnamefont {R.}~\bibnamefont
  {Jozsa}},\ }in\ \href@noop {} {\emph {\bibinfo {booktitle} {{1st NASA
  Conference on Quantum Computing and Quantum Communications}}}}\ (\bibinfo
  {year} {1998})\ \Eprint {http://arxiv.org/abs/quant-ph/9805086}
  {arXiv:quant-ph/9805086} \BibitemShut {NoStop}%
\bibitem [{\citenamefont {Mitchison}\ and\ \citenamefont
  {Jozsa}(2001)}]{Mitchison2001}%
  \BibitemOpen
  \bibfield  {author} {\bibinfo {author} {\bibfnamefont {G.}~\bibnamefont
  {Mitchison}}\ and\ \bibinfo {author} {\bibfnamefont {R.}~\bibnamefont
  {Jozsa}},\ }\href {\doibase 10.1098/rspa.2000.0714} {\bibfield  {journal}
  {\bibinfo  {journal} {Proceedings of the Royal Society of London. Series A:
  Mathematical, Physical and Engineering Sciences}\ }\textbf {\bibinfo {volume}
  {457}},\ \bibinfo {pages} {1175} (\bibinfo {year} {2001})}\BibitemShut
  {NoStop}%
\bibitem [{\citenamefont {Hasenöhrl}\ and\ \citenamefont
  {Wolf}(2020)}]{2010.00623}%
  \BibitemOpen
  \bibfield  {author} {\bibinfo {author} {\bibfnamefont {M.}~\bibnamefont
  {Hasenöhrl}}\ and\ \bibinfo {author} {\bibfnamefont {M.~M.}\ \bibnamefont
  {Wolf}},\ }\href@noop {} {\enquote {\bibinfo {title} {{'Interaction-Free'
  Channel Discrimination}},}\ } (\bibinfo {year} {2020}),\ \Eprint
  {http://arxiv.org/abs/arXiv:2010.00623} {arXiv:2010.00623} \BibitemShut
  {NoStop}%
\bibitem [{\citenamefont {Kwiat}\ \emph {et~al.}(1995)\citenamefont {Kwiat},
  \citenamefont {Weinfurter}, \citenamefont {Herzog}, \citenamefont
  {Zeilinger},\ and\ \citenamefont {Kasevich}}]{PhysRevLett.74.4763}%
  \BibitemOpen
  \bibfield  {author} {\bibinfo {author} {\bibfnamefont {P.}~\bibnamefont
  {Kwiat}}, \bibinfo {author} {\bibfnamefont {H.}~\bibnamefont {Weinfurter}},
  \bibinfo {author} {\bibfnamefont {T.}~\bibnamefont {Herzog}}, \bibinfo
  {author} {\bibfnamefont {A.}~\bibnamefont {Zeilinger}}, \ and\ \bibinfo
  {author} {\bibfnamefont {M.~A.}\ \bibnamefont {Kasevich}},\ }\href {\doibase
  10.1103/PhysRevLett.74.4763} {\bibfield  {journal} {\bibinfo  {journal}
  {Phys. Rev. Lett.}\ }\textbf {\bibinfo {volume} {74}},\ \bibinfo {pages}
  {4763} (\bibinfo {year} {1995})}\BibitemShut {NoStop}%
\bibitem [{\citenamefont {Weinfurter}\ \emph {et~al.}(1995)\citenamefont
  {Weinfurter}, \citenamefont {Kwiat}, \citenamefont {Herzog}, \citenamefont
  {Zeilinger},\ and\ \citenamefont {Kasevich}}]{Weinfurter95}%
  \BibitemOpen
  \bibfield  {author} {\bibinfo {author} {\bibfnamefont {H.}~\bibnamefont
  {Weinfurter}}, \bibinfo {author} {\bibfnamefont {P.}~\bibnamefont {Kwiat}},
  \bibinfo {author} {\bibfnamefont {T.}~\bibnamefont {Herzog}}, \bibinfo
  {author} {\bibfnamefont {A.}~\bibnamefont {Zeilinger}}, \ and\ \bibinfo
  {author} {\bibfnamefont {M.}~\bibnamefont {Kasevich}},\ }in\ \href
  {http://www.osapublishing.org/abstract.cfm?URI=QELS-1995-QTuF5} {\emph
  {\bibinfo {booktitle} {Quantum Electronics and Laser Science Conference}}}\
  (\bibinfo  {publisher} {Optical Society of America},\ \bibinfo {year}
  {1995})\ p.\ \bibinfo {pages} {QTuF5}\BibitemShut {NoStop}%
\bibitem [{\citenamefont {Kwiat}\ \emph {et~al.}(1999)\citenamefont {Kwiat},
  \citenamefont {White}, \citenamefont {Mitchell}, \citenamefont {Nairz},
  \citenamefont {Weihs}, \citenamefont {Weinfurter},\ and\ \citenamefont
  {Zeilinger}}]{PhysRevLett.83.4725}%
  \BibitemOpen
  \bibfield  {author} {\bibinfo {author} {\bibfnamefont {P.~G.}\ \bibnamefont
  {Kwiat}}, \bibinfo {author} {\bibfnamefont {A.~G.}\ \bibnamefont {White}},
  \bibinfo {author} {\bibfnamefont {J.~R.}\ \bibnamefont {Mitchell}}, \bibinfo
  {author} {\bibfnamefont {O.}~\bibnamefont {Nairz}}, \bibinfo {author}
  {\bibfnamefont {G.}~\bibnamefont {Weihs}}, \bibinfo {author} {\bibfnamefont
  {H.}~\bibnamefont {Weinfurter}}, \ and\ \bibinfo {author} {\bibfnamefont
  {A.}~\bibnamefont {Zeilinger}},\ }\href {\doibase
  10.1103/PhysRevLett.83.4725} {\bibfield  {journal} {\bibinfo  {journal}
  {Phys. Rev. Lett.}\ }\textbf {\bibinfo {volume} {83}},\ \bibinfo {pages}
  {4725} (\bibinfo {year} {1999})}\BibitemShut {NoStop}%
\bibitem [{\citenamefont {White}\ \emph {et~al.}(1998)\citenamefont {White},
  \citenamefont {Kwiat}, \citenamefont {Nairz}, \citenamefont {Weihs},\ and\
  \citenamefont {Zeilinger}}]{White:98}%
  \BibitemOpen
  \bibfield  {author} {\bibinfo {author} {\bibfnamefont {A.~G.}\ \bibnamefont
  {White}}, \bibinfo {author} {\bibfnamefont {P.~G.}\ \bibnamefont {Kwiat}},
  \bibinfo {author} {\bibfnamefont {O.}~\bibnamefont {Nairz}}, \bibinfo
  {author} {\bibfnamefont {G.}~\bibnamefont {Weihs}}, \ and\ \bibinfo {author}
  {\bibfnamefont {A.}~\bibnamefont {Zeilinger}},\ }in\ \href
  {http://www.osapublishing.org/abstract.cfm?URI=IQEC-1998-QMF1} {\emph
  {\bibinfo {booktitle} {International Quantum Electronics Conference}}}\
  (\bibinfo  {publisher} {Optical Society of America},\ \bibinfo {year}
  {1998})\ p.\ \bibinfo {pages} {QMF1}\BibitemShut {NoStop}%
\bibitem [{\citenamefont {Hardy}(1992)}]{hardy92}%
  \BibitemOpen
  \bibfield  {author} {\bibinfo {author} {\bibfnamefont {L.}~\bibnamefont
  {Hardy}},\ }\href {\doibase 10.1103/PhysRevLett.68.2981} {\bibfield
  {journal} {\bibinfo  {journal} {Phys. Rev. Lett.}\ }\textbf {\bibinfo
  {volume} {68}},\ \bibinfo {pages} {2981} (\bibinfo {year}
  {1992})}\BibitemShut {NoStop}%
\bibitem [{\citenamefont {Aharonov}\ \emph {et~al.}(2002)\citenamefont
  {Aharonov}, \citenamefont {Botero}, \citenamefont {Popescu}, \citenamefont
  {Reznik},\ and\ \citenamefont {Tollaksen}}]{Aharonov2002}%
  \BibitemOpen
  \bibfield  {author} {\bibinfo {author} {\bibfnamefont {Y.}~\bibnamefont
  {Aharonov}}, \bibinfo {author} {\bibfnamefont {A.}~\bibnamefont {Botero}},
  \bibinfo {author} {\bibfnamefont {S.}~\bibnamefont {Popescu}}, \bibinfo
  {author} {\bibfnamefont {B.}~\bibnamefont {Reznik}}, \ and\ \bibinfo {author}
  {\bibfnamefont {J.}~\bibnamefont {Tollaksen}},\ }\href {\doibase
  10.1016/S0375-9601(02)00986-6} {\bibfield  {journal} {\bibinfo  {journal}
  {Physics Letters A}\ }\textbf {\bibinfo {volume} {301}},\ \bibinfo {pages}
  {130} (\bibinfo {year} {2002})}\BibitemShut {NoStop}%
\bibitem [{\citenamefont {Lundeen}\ and\ \citenamefont
  {Steinberg}(2009)}]{PhysRevLett.102.020404}%
  \BibitemOpen
  \bibfield  {author} {\bibinfo {author} {\bibfnamefont {J.~S.}\ \bibnamefont
  {Lundeen}}\ and\ \bibinfo {author} {\bibfnamefont {A.~M.}\ \bibnamefont
  {Steinberg}},\ }\href {\doibase 10.1103/PhysRevLett.102.020404} {\bibfield
  {journal} {\bibinfo  {journal} {Phys. Rev. Lett.}\ }\textbf {\bibinfo
  {volume} {102}},\ \bibinfo {pages} {020404} (\bibinfo {year}
  {2009})}\BibitemShut {NoStop}%
\bibitem [{\citenamefont {Yokota}\ \emph {et~al.}(2009)\citenamefont {Yokota},
  \citenamefont {Yamamoto}, \citenamefont {Koashi},\ and\ \citenamefont
  {Imoto}}]{Yokota_2009}%
  \BibitemOpen
  \bibfield  {author} {\bibinfo {author} {\bibfnamefont {K.}~\bibnamefont
  {Yokota}}, \bibinfo {author} {\bibfnamefont {T.}~\bibnamefont {Yamamoto}},
  \bibinfo {author} {\bibfnamefont {M.}~\bibnamefont {Koashi}}, \ and\ \bibinfo
  {author} {\bibfnamefont {N.}~\bibnamefont {Imoto}},\ }\href {\doibase
  10.1088/1367-2630/11/3/033011} {\bibfield  {journal} {\bibinfo  {journal}
  {New Journal of Physics}\ }\textbf {\bibinfo {volume} {11}},\ \bibinfo
  {pages} {033011} (\bibinfo {year} {2009})}\BibitemShut {NoStop}%
\bibitem [{\citenamefont {Vaidman}(2007)}]{Vaidman_PostSelection07}%
  \BibitemOpen
  \bibfield  {author} {\bibinfo {author} {\bibfnamefont {L.}~\bibnamefont
  {Vaidman}},\ }\href {\doibase 10.1103/PhysRevLett.98.160403} {\bibfield
  {journal} {\bibinfo  {journal} {Phys. Rev. Lett.}\ }\textbf {\bibinfo
  {volume} {98}},\ \bibinfo {pages} {160403} (\bibinfo {year}
  {2007})}\BibitemShut {NoStop}%
\bibitem [{\citenamefont {Hosten}\ \emph {et~al.}(2006)\citenamefont {Hosten},
  \citenamefont {Rakher}, \citenamefont {Barreiro}, \citenamefont {Peters},\
  and\ \citenamefont {Kwiat}}]{Hosten2006}%
  \BibitemOpen
  \bibfield  {author} {\bibinfo {author} {\bibfnamefont {O.}~\bibnamefont
  {Hosten}}, \bibinfo {author} {\bibfnamefont {M.~T.}\ \bibnamefont {Rakher}},
  \bibinfo {author} {\bibfnamefont {J.~T.}\ \bibnamefont {Barreiro}}, \bibinfo
  {author} {\bibfnamefont {N.~A.}\ \bibnamefont {Peters}}, \ and\ \bibinfo
  {author} {\bibfnamefont {P.~G.}\ \bibnamefont {Kwiat}},\ }\href {\doibase
  10.1038/nature04523} {\bibfield  {journal} {\bibinfo  {journal} {Nature}\
  }\textbf {\bibinfo {volume} {439}},\ \bibinfo {pages} {949} (\bibinfo {year}
  {2006})}\BibitemShut {NoStop}%
\bibitem [{\citenamefont {Jang}(1999)}]{PhysRevA.59.2322}%
  \BibitemOpen
  \bibfield  {author} {\bibinfo {author} {\bibfnamefont {J.-S.}\ \bibnamefont
  {Jang}},\ }\href {\doibase 10.1103/PhysRevA.59.2322} {\bibfield  {journal}
  {\bibinfo  {journal} {Phys. Rev. A}\ }\textbf {\bibinfo {volume} {59}},\
  \bibinfo {pages} {2322} (\bibinfo {year} {1999})}\BibitemShut {NoStop}%
\bibitem [{\citenamefont {Aharonov}\ and\ \citenamefont
  {Vaidman}(2008)}]{TwoStateFormalisumAharonov2008}%
  \BibitemOpen
  \bibfield  {author} {\bibinfo {author} {\bibfnamefont {Y.}~\bibnamefont
  {Aharonov}}\ and\ \bibinfo {author} {\bibfnamefont {L.}~\bibnamefont
  {Vaidman}},\ }\enquote {\bibinfo {title} {The two-state vector formalism: An
  updated review},}\ in\ \href {\doibase 10.1007/978-3-540-73473-4_13} {\emph
  {\bibinfo {booktitle} {Time in Quantum Mechanics}}},\ \bibinfo {editor}
  {edited by\ \bibinfo {editor} {\bibfnamefont {J.}~\bibnamefont {Muga}},
  \bibinfo {editor} {\bibfnamefont {R.~S.}\ \bibnamefont {Mayato}}, \ and\
  \bibinfo {editor} {\bibfnamefont {{\'I}.}~\bibnamefont {Egusquiza}}}\
  (\bibinfo  {publisher} {Springer Berlin Heidelberg},\ \bibinfo {address}
  {Berlin, Heidelberg},\ \bibinfo {year} {2008})\ pp.\ \bibinfo {pages}
  {399--447}\BibitemShut {NoStop}%
\bibitem [{\citenamefont {Spekkens}(2005)}]{Spekkens05}%
  \BibitemOpen
  \bibfield  {author} {\bibinfo {author} {\bibfnamefont {R.~W.}\ \bibnamefont
  {Spekkens}},\ }\href {\doibase 10.1103/PhysRevA.71.052108} {\bibfield
  {journal} {\bibinfo  {journal} {Phys. Rev. A}\ }\textbf {\bibinfo {volume}
  {71}},\ \bibinfo {pages} {052108} (\bibinfo {year} {2005})}\BibitemShut
  {NoStop}%
\bibitem [{\citenamefont {Bell}(1964)}]{Bell}%
  \BibitemOpen
  \bibfield  {author} {\bibinfo {author} {\bibfnamefont {J.~S.}\ \bibnamefont
  {Bell}},\ }\href {\doibase 10.1103/PhysicsPhysiqueFizika.1.195} {\bibfield
  {journal} {\bibinfo  {journal} {Physics Physique Fizika}\ }\textbf {\bibinfo
  {volume} {1}},\ \bibinfo {pages} {195} (\bibinfo {year} {1964})}\BibitemShut
  {NoStop}%
\bibitem [{\citenamefont {Kochen}\ and\ \citenamefont {Specker}(1967)}]{KS}%
  \BibitemOpen
  \bibfield  {author} {\bibinfo {author} {\bibfnamefont {S.}~\bibnamefont
  {Kochen}}\ and\ \bibinfo {author} {\bibfnamefont {E.~P.}\ \bibnamefont
  {Specker}},\ }\href {http://www.jstor.org/stable/24902153} {\bibfield
  {journal} {\bibinfo  {journal} {Journal of Mathematics and Mechanics}\
  }\textbf {\bibinfo {volume} {17}},\ \bibinfo {pages} {59} (\bibinfo {year}
  {1967})}\BibitemShut {NoStop}%
\bibitem [{\citenamefont {Gitton}\ and\ \citenamefont
  {Woods}(2020)}]{Gitton2020}%
  \BibitemOpen
  \bibfield  {author} {\bibinfo {author} {\bibfnamefont {V.}~\bibnamefont
  {Gitton}}\ and\ \bibinfo {author} {\bibfnamefont {M.~P.}\ \bibnamefont
  {Woods}},\ }\href@noop {} {\enquote {\bibinfo {title} {Solvable criterion for
  the contextuality of any prepare-and-measure scenario},}\ } (\bibinfo {year}
  {2020}),\ \Eprint {http://arxiv.org/abs/arXiv:2003.06426} {arXiv:2003.06426}
  \BibitemShut {NoStop}%
\bibitem [{\citenamefont {Spekkens}\ \emph {et~al.}()\citenamefont {Spekkens},
  \citenamefont {Elliot},\ and\ \citenamefont {Leifer}}]{classical_bomb}%
  \BibitemOpen
  \bibfield  {author} {\bibinfo {author} {\bibfnamefont {R.}~\bibnamefont
  {Spekkens}}, \bibinfo {author} {\bibfnamefont {M.}~\bibnamefont {Elliot}}, \
  and\ \bibinfo {author} {\bibfnamefont {M.}~\bibnamefont {Leifer}},\
  }\href@noop {} {\ }\bibinfo {note} {In preparation; see talk here
  \url{http://pirsa.org/16060102/}}\BibitemShut {NoStop}%
\bibitem [{\citenamefont {Barnett}\ and\ \citenamefont
  {Croke}(2009)}]{Barnett09}%
  \BibitemOpen
  \bibfield  {author} {\bibinfo {author} {\bibfnamefont {S.~M.}\ \bibnamefont
  {Barnett}}\ and\ \bibinfo {author} {\bibfnamefont {S.}~\bibnamefont
  {Croke}},\ }\href {\doibase 10.1364/AOP.1.000238} {\bibfield  {journal}
  {\bibinfo  {journal} {Adv. Opt. Photon.}\ }\textbf {\bibinfo {volume} {1}},\
  \bibinfo {pages} {238} (\bibinfo {year} {2009})}\BibitemShut {NoStop}%
\bibitem [{\citenamefont {Havlicek}\ and\ \citenamefont
  {Svozil}(2018)}]{DimLift}%
  \BibitemOpen
  \bibfield  {author} {\bibinfo {author} {\bibfnamefont {H.}~\bibnamefont
  {Havlicek}}\ and\ \bibinfo {author} {\bibfnamefont {K.}~\bibnamefont
  {Svozil}},\ }\href {\doibase 10.3390/e20040284} {\bibfield  {journal}
  {\bibinfo  {journal} {Entropy}\ }\textbf {\bibinfo {volume} {20}} (\bibinfo
  {year} {2018}),\ 10.3390/e20040284}\BibitemShut {NoStop}%
\bibitem [{\citenamefont {Avis}(2000)}]{avis2000revised}%
  \BibitemOpen
  \bibfield  {author} {\bibinfo {author} {\bibfnamefont {D.}~\bibnamefont
  {Avis}},\ }in\ \href {\doibase 10.1007/978-3-0348-8438-9_9} {\emph {\bibinfo
  {booktitle} {{Polytopes-Combinatorics and Computation}}}}\ (\bibinfo
  {organization} {Springer},\ \bibinfo {year} {2000})\ pp.\ \bibinfo {pages}
  {177--198}\BibitemShut {NoStop}%
\bibitem [{\citenamefont {Boyd}(2006)}]{Tophat}%
  \BibitemOpen
  \bibfield  {author} {\bibinfo {author} {\bibfnamefont {J.~P.}\ \bibnamefont
  {Boyd}},\ }\href {\doibase 10.1007/s10915-005-9010-7} {\bibfield  {journal}
  {\bibinfo  {journal} {Journal of Scientific Computing}\ }\textbf {\bibinfo
  {volume} {29}} (\bibinfo {year} {2006}),\
  10.1007/s10915-005-9010-7}\BibitemShut {NoStop}%
\bibitem [{\citenamefont {Greitzer}(1986)}]{GeometriSeries}%
  \BibitemOpen
  \bibfield  {author} {\bibinfo {author} {\bibfnamefont {S.}~\bibnamefont
  {Greitzer}},\ }\href
  {https://1lib.ch/book/1175413/7082f7?id=1175413&secret=7082f7} {\bibfield
  {journal} {\bibinfo  {journal} {Arbelos}\ }\textbf {\bibinfo {volume} {4}},\
  \bibinfo {pages} {14} (\bibinfo {year} {1986})},\ \bibinfo {note} {{
  Alternatively, see
  \href{https://evergreen.loyola.edu/mpknapp/www/papers/knapp-sv.pdf}{M. P.
  Knapp, ``Sines and Cosines of Angles in Arithmetic Progression''}
  }}\BibitemShut {NoStop}%
\end{thebibliography}%
\bibliographystyle{apsrev4-1}

\section{\Meth}

Here we provide more details concerning the claims made in the main text. There are three subsections. The first, \cref{sec:backwards in time analysis}, describes the two state formalism analysis in which the states evolve both forwards and backwards in time. The second, \cref{sec:non contextual proof}, describes the setup and how we performed our calculations which led us to conclude that our counterfactual clock has no noncontextual ontic description. The third, \cref{sec:Engineered clocks}, provides more details about the engineered clocks. The fourth, \cref{Proof of thm:time is substantiva}, contains a proof of \cref{thm:time is substantival}.

\subsection{Backwards in time analysis}\label{sec:backwards in time analysis}
Here we analyse our protocol backwards in time for the post-selected outcomes of interest. The objective is to show that the clock was off when viewed from this perspective. For simplicity, we analyse the case $\NT=1$ presented in the main text where the probability of obtaining either \ifreeoutcome{} is $1/6$. The general case is relegated to \app~\ref{sec:Backwards in time analysis in the multiple tick scenario}. The physical interpretation is the same in all cases.

The amplitudes associated with obtaining the \ifreeoutcomes{} are given by
\begin{align}
\braket{E| U_\measure U(\tau_0) U_0 |E}\qquad \text{ and } \qquad \braket{A| U_\measure U(\tau_1) U_0 |E},
\end{align}
for times $t=\tau_0$ and $t=\tau_1$ respectively, where $U(\tau_j)$ is the unitary dynamics of the clock (when it is on) for a time $\tau_j$. In the main text we interpreted these processes going forwards in time, namely starting in state $\ket{E}$ and analysing the state after applying $U_0$, $U(\tau_j)$, $U_\measure$ and finally post-selecting on measurement outcomes $\ket{E}$ and $\ket{A}$. Here we will analyse the reverse process: starting with the post-selected state $\bra{E}$ or $\bra{A}$, evolving backwards in time by analysing the affect of the sequence of unitaries $U_\measure$,  $U(\tau_j)$, $U_0$, followed by pre-selecting onto the initial state $\ket{E}$. This process is known as the two-state vector formalism. This formalism, and many other variants on this idea, permit time-symmetric interpretations of quantum mechanics (see review~\cite{TwoStateFormalisumAharonov2008} for more details).

For the example under consideration, $U_\measure$ is provided at the end of \app~\ref{sec:Measurement basis characterisation}. Using it we find that the backwards evolved post-selected \ifreeoutcome{} state just prior to dynamics take on the form:
\begin{align}\label{eq:bawards 1}
\bra{E} U_\measure = \frac{1}{\sqrt{3}} \big( \bra{E}-\bra{\tau_1} - \bra{A}  \big)  \qquad \text{ and } \qquad \bra{A} U_\measure = \frac{1}{\sqrt{3}} \big( \bra{E}-\bra{\tau_0} - \bra{A}  \big).
\end{align}
We now have to evolve backwards in time for an amount $\tau_0$. Recall that $U(\tau_0)$ induces the following transitions: $U(\tau_0)\ket{E}=\ket{E}$, $U(\tau_0)\ket{A}=\ket{A}$, $U(\tau_0)\ket{\psi}=\ket{\tau_0}$. Consider the ansatz $\bra{\tau_1}U(\tau_0)= c_0 \bra{\psi_\perp}+c_2\bra{E}+c_3\bra{\psi}$, where $\ket{\psi_\perp}$ is some state which is perpendicular to $\ket{E}$ and $\ket{\psi}$. We have $0=\braket{\tau_1|E}=\braket{\tau_1|U(\tau_0) U(\tau_0)^\dag | E}= c_2$. Similarly we find that $c_3=0$ and thus that $\bra{\tau_1}U(\tau_0)= \bra{\psi_\perp}$. Furthermore, since $U(\tau_1)\ket{E}=\ket{E}$, $U(\tau_1)\ket{A}=\ket{A}$,  $U(\tau_1)\ket{\psi}=\ket{\tau_1}$, it follows in an analogous fashion that $\bra{\tau_0}U(\tau_1)= \bra{\psi_\perp'}$, where $\bra{\psi_\perp'}$ is some state perpendicular to $\ket{E}$ and $\ket{\psi}$. Therefore, from \cref{eq:bawards 1},  the backwards in time dynamics yields the states 
\begin{align}\label{eq:bawards 2}
\bra{E} U_\measure U(\tau_0) = \frac{1}{\sqrt{3}} \big( \bra{E}-\bra{\psi_\perp} - \bra{A}  \big)  \qquad \text{ and } \qquad \bra{A} U_\measure U(\tau_1) = \frac{1}{\sqrt{3}} \big( \bra{E}-\bra{\psi_\perp'} - \bra{A}  \big).
\end{align}
The unitary  $U_0$ only acts non trivially in the subspace spanned by $\ket{E}$ and $\ket{\psi}$:  $U_0\ket{E}=c\ket{E}+s\ket{\psi}$, $U_0\ket{\psi}=-s\ket{E}+c\ket{\psi}$, and thus it acts trivially on the dynamical branches $\bra{\psi_\perp}$ and $\bra{\psi_\perp'}$. We thus have
\begin{align}\label{eq:bawards 3}
\bra{E} U_\measure U(\tau_0) U_0 = \frac{1}{\sqrt{3}} \big( c\bra{E}-s\bra{\psi}-\bra{\psi_\perp} - \bra{A}  \big)  \quad \text{ and } \quad \bra{A} U_\measure U(\tau_1)U_0 = \frac{1}{\sqrt{3}} \big( c\bra{E}-s\bra{\psi}-\bra{\psi_\perp'} - \bra{A}  \big).
\end{align}
Finally, by post-selecting on $\ket{E}$, and recalling that $c=1/\sqrt{2}$, we find that the amplitudes associated with the \ifreeoutcomes{} is $1/\sqrt{6}$ in both cases, hence leading to a probability of $1/6$ for each \ifreeoutcome{} as stated in the main text. Importantly, during the entire backwards evolution, we see that the dynamical branches of the wave function has not contributed to the pre-selected state $\ket{E}$.  Another way to see this is to note that had we measured whether the clock was on or off during the evolution in either of the above pre and post selected cases, we would have found that the probability of finding the clock on was zero. Namely, we can define the binary outcome projective measurement $P_{\on}:=\proj{\psi}+\proj{\tau_0}+\proj{\tau_1}$ and $P_{{\scriptscriptstyle \backslash}\on}:=\proj{E}+\proj{A}$, where $P_{\on}$ corresponds to a projection onto the on subspace. From~\cite{TwoStateFormalisumAharonov2008}, it follows that the probability of the clock being on at any time during the dynamical phase, that is any $t\in(0,\tau_0)$, for the run of the experiment corresponding to post-selecting on $\ket{E}$ is
\begin{align}
\textup{Prob}\big(\on, E\big)&:= \frac{\big{|}\braket{E|U_\measure U(\tau_0-t) P_\on U(t) U_0 |E}\big{|}^2 }{\big{|}\braket{E|U_\measure U(\tau_0-t) P_\on U(t) U_0 |E}\big{|}^2+ \big{|}\braket{E|U_\measure U(\tau_0-t) P_{{\scriptscriptstyle \backslash}\on} U(t) U_0 |E}\big{|}^2}.
\end{align}
Note that since the ancilla and energy eigenstate are stationary, we have that $U(t)=U_{\perp E, A}(t) + \id_{E,A}$, where $\id_{E,A}$ is the identity operator on the subspace spanned by $\ket{E}$ and $\ket{A}$, and that $U_{\perp, A}(t)$ is orthogonal to said subspace. Therefore, if we assume that the dynamics is Markovian, namely $U(t_1)U(t_2)=U(t_1+t_2)$ for all $t_1, t_2\in\rr$, it follows that $U(\tau_0-t) P_\on U(t)=U_{\perp E, A}(\tau_0)$ and hence that 
\begin{align}
\textup{Prob}\big(\on, E\big)=0
\end{align}
for all $t\in(0,\tau_0)$.
Similarity, it follows that the probability of the clock being on at any time during the dynamical phase, that is any $t\in(0,\tau_1)$, for the run of the experiment corresponding to post-selecting on $\ket{A}$ is
\begin{align}
\textup{Prob}\big(\on, A\big)&:= \frac{\big{|}\braket{A|U_\measure U(\tau_1-t) P_\on U(t) U_0 |E}\big{|}^2 }{\big{|}\braket{A|U_\measure U(\tau_1-t) P_\on U(t) U_0 |E}\big{|}^2+ \big{|}\braket{A|U_\measure U(\tau_1-t) P_{{\scriptscriptstyle \backslash}\on} U(t) U_0 |E}\big{|}^2}=0.
\end{align}

\subsection{No classical analogue}\label{sec:non contextual proof}

The interpretation of measuring time when the clock was off using  counterfactual reasoning, relied on a basic concept in quantum mechanics, namely the superposition principle, in which if a system is in a superposition of two states, it cannot be regarded as being in either until measured. As we discussed, Schr\"odinger famously popularised this point with his thought experiment concerning a cat. However, what if underlying our counterfactual clock protocol, there was a ``real'' description; that is to say if one could associate the states in our clock protocols with probability distributions $\lambda$ over some underlying states | a so-called \emph{ontic state apace} $\Lambda$ | and the measurements with update rules for said distributions, in a meaningful way? While such a result would not invalidate the current interpretation provided, it would, at least in-principle, provide for an alternative interpretation in which the clock might have been dynamically evolving even when a measurement collapsing the clock to the off branch is obtained. We will now consider such a possibility for the simplest elementary timekeeping systems. 

To start with, lets consider the case where we only run the clock when it is on; that is, the clock is never in a superposition of on and off. Then the question is, can we find an ontic state space on which there exists appropriate probability distributions that represent our states $\{ \proj{E}, \proj{\tau_0}, \proj{\tau_1}  \}$ and measurements (with PVM elements $\{ \proj{E}$, $ \proj{\tau_0}$, $\proj{\tau_1}$, $\proj{A}\}$). Moreover, if an underlying ontic state space does exist, one should expect to be able to take probabilistic mixtures of the states and measurements in it: if your lab assistant walked in to your lab at some unknown stage of the protocol's implementation, the assistant would attribute such a description. Thus denoting the convex hull by $\textup{conv}$, we required that when a state $\rho\in \sets_\on= \textup{conv}\big(  \{ \proj{E}, \proj{\tau_0}, \proj{\tau_1}  \} \big)$ is prepared in our protocol, the probability that the ontic state is in state $\lambda$, is $P\big[\lambda\big|\rho\big]$. Likewise, when we make the canonical measurement with corresponding POVM element $\effect\in \sete_\on= \textup{conv}\big(\{ \proj{E}, \proj{\tau_0}, \proj{\tau_1}, \proj{A}  \} \big)$, then the probability that the outcome associated with $\effect$ occurred, given that the ontic state is in state $\lambda$, is $P\big[\effect\big|\lambda\big]$.
Furthermore, we require the usual convexity relationships of mixtures of quantum states and measurements hold at the ontic level: for all $\lambda\in\Lambda$, for all $p\in[0,1]$, for all $\rho_1,\rho_2\in \sets_\on$: $P\big[\lambda \big|p \rho_1+(1-p) \rho_2\big]= p P\big[\lambda\big| \rho_1\big]+ (1-p) P\big[\lambda\big|\rho_2\big]$ and similarly, that for all $\lambda\in\Lambda$, for all $p\in[0,1]$, for all $\effect_1,\effect_2\in \sete_\on$: $P\big[p \effect_1+(1-p) \effect_2\big|\lambda\big]= p P\big[ \effect_1\big|\lambda\big]+ (1-p) P\big[\effect_2\big|\lambda\big]$.  Finally, we require that the ontic states can reproduce the measurement statistics of our protocol: for all $\rho\in\sets_\on$, for all $\effect\in\sete_\on$,
\begin{align}
\tr\big[\rho\, \effect\big] =\int_\Lambda\!\! d\lambda\, P\big[ \lambda \big| \rho  \big] P\big[\effect \big| \lambda \big], 
\end{align}
where $\tr[\cdot]$ denotes the trace, and the r.h.s. is the probability of obtaining outcome associated with $\effect$ for quantum state $\rho$, according to quantum mechanics.\Mspace

Notice how we have assumed that the probability distributions $P\big[\lambda\big|\rho\big]$ do not depend on how the quantum state $\rho$ was prepared nor do $P\big[\effect\big|\lambda\big]$ depend on how the measurements were implemented. Different preparation procedures which lead to the same quantum state and different measurement procedures which lead to the same POVM elements, are called different \emph{contexts}. We are thus looking for a non-contextual ontic variable model. This notion of contextuality was pioneered by Spekkens \cite{Spekkens05}. Prior notions of classicality based on local realism~\cite{Bell} or its generalisation to non-contextual realism \cite{KS}, cannot be tested in our case since we do not dispose of a local structure, nor relevant communing measurements, in our clock protocols. However, in keeping with this philosophy of preparation and measurement non-contextuality, our current formulation of an ontic variable model requires some refinement: since the set $\sete_\on$ is not tomographically complete, nor the set $\sets_\on$ span all quantum states in the Hilbert space, neither all the d.o.f. of the density matrices nor those in the POVM elements, contribute to the measurement statistics. Indeed, $\tr\big[\rho \effect\big]=\tr\big[\mathcal{P}_\mainr(\rho) \mathcal{P}_\mainr (\effect)\big] $ for all $\rho\in\sets_\on$, $\effect\in\sete_\on$, where $\mathcal{P}_\mainr$ is a projection onto the vector space $\mainr$, generated by projecting the span of $\sets_\on$ onto the span of $\sete_\on$. We should thus make the replacements $P\big[ \lambda \big| \cdot  \big] \mapsto P\big[ \lambda \big| \mathcal{P}_\mainr(\cdot ) \big]$ and $P\big[\cdot \big| \lambda \big] \mapsto P\big[\mathcal{P}_\mainr(\cdot) \big| \lambda \big]$ in the above theory. This completes the description of our would-be non-contextual ontic variable model. Finally, if a classical description exists, one would also need a stochastic map describing the dynamics of the ontic variables throughout the protocols. This imposes an additional constraint on our would-be ontic description, which we will not need to consider.\Mspace

In this case, it can readily be seen that such a model \emph{does} exist; for example, one may simply choose the vectors $\{ \ket{E}\!, \ket{\tau_0}\!, \ket{\tau_1}\!, \ket{A}  \}$ as a basis for the ontic state space $\Lambda$ with each step and measurement allocated to deterministic distributions $\lambda$ on it.


We now turn to the case is which the clock can also be used to tell the time when off via \ifree{} measurements. We need to supplement the sets $\sets_\on,\sete_\on$ with the other relevant elements which are now required, namely, for states,
\begin{align}\label{eq:states}
\sets_\cf= \textup{conv}\big(\sets_\on\cup \{ \proj{\cf_0}, \proj{\cf_1}, U_\measure\proj{\cf_0}U_\measure^\dag, U_\measure\proj{\cf_1}U_\measure^\dag \} \big),
\end{align}
where $\ket{\cf_0}:=c \ket{E}+s\ket{\tau_0}$, $\ket{\cf_1}:=c \ket{E}+s\ket{\tau_1}$, since these are the states which appear in our protocol. When it comes to the measurements, in addition to those required in the final measurement, namely $\sete_\on$, we want to include measurements corresponding to ontic degrees of freedom which are able to describe the paradoxical aspects of our protocol. Before applying $U_\measure$, at times $\tau_0$, $\tau_1$, when we measure in the measurement basis, the measurement determined which branch (off or on) we collapsed to, but we can only deduce the time when we happen to collapse onto an on branch. This situation is analogous to a statistical mixture over on and off. What is surprising, is that after $U_\measure$ is applied, we can deduce not only which branch we were on, but also the time when we collapse to an off branch. So our would-be ontic model should have a variable which determines whether we are in the off branch, or the on branch at the times $\tau_0$, $\tau_1$. In other words, a variable which, after the application of $U_\measure$, plays the same role that the measurement basis played before the application of $U_\measure$. Including this in the set of things which are measurable in our would-be non-contextual theory, gives us 
\begin{align}\label{eq:effects}
\sete_\cf= \textup{conv}\big(\sete_\on\cup \{ U_\measure^\dag\proj{E}U_\measure, U_\measure^\dag\proj{\tau_0}U_\measure, U_\measure^\dag\proj{\tau_1}U_\measure, U_\measure^\dag\proj{A}U_\measure \} \big).
\end{align}
An alternative motivation for the inclusion of the additional ontic degrees of freedom, is that, quantum mechanically, the unitary channel generated by $U_\measure$ could have been applied to rotate the measurement basis, rather than being applied to the state. This alternative protocol, is equivalent to the one we study here, from the perspective of quantum mechanics. We can thus think of these two alternative implementations of our protocol as different contexts which are indistinguishable according to the laws of quantum mechanics, in an analogous way to how would-be ontic variable model is preparation and measurement non-contextual by design.\Mspace

In~\cite{Gitton2020} an algorithm was developed which, given a set $\sets$ of states and set $\sete$ of POVM elements, can determine whether an ontic variable model as described here exits. When said sets have finitely many extremal points, as in the case for $\sets_\cf$, $\sete_\cf$, it provably runs in finite time. For the simple case where the probability of each \ifreeoutcome{} is $1/6$ as described in the main text, the unitary $U_\measure$ is given by~\cref{table:special form} and $c, s$ by~\cref{eq:eq:c equals s equal 1}. For this case, we used the algorithm to provide a computer assisted proof that the sets $\sets_\cf, \sete_\cf$ | and hence our counterfactual clock | do \emph{not} admit a non-contextual ontic variable model as per the above description; see \app~\ref{sec:non contextual proof numerics} for details.\Mspace

A crucial aspect of our protocol for using the clock to tell the time when off, via \ifree{} measurements, was the existence of negative amplitudes allowing for destructive interference between the on and off branches. Since probabilities are nonnegative, one might have thought that this aspect of our protocol automatically rules out any non contextual realistic theory. However, this is definitely not the case since other experiments using \ifree{} measurements, such as the Elizur-Vaidman bomb test~\cite{Elitzur1993}, have been shown to admit a non contextual ontic model description analogous to the one ruled out here for our setup; see~\cite{classical_bomb}.

\subsection{Engineered clocks}\label{sec:Engineered clocks} Here we give a more detailed account of the engineered clocks outline in the main text. As before, we will discuss the simplest case here of just one tick while relegating the full details of the construction and multiple tick scenario, which is qualitatively the same, to the \app{} (\cref{sec:telling time when all is off}). As with the elementary clock, we have a stationary state, $\ket{\psi_\off}=\me^{-\mi t \hat H} \ket{\psi_\off}$, and a dynamical one, $\ket{\psi_\on(t)}=\me^{-\mi t \hat H} \ket{\psi_\on}$, which are mutually orthogonal at all times: $\braket{\psi_\off|\psi_\on(t)}=0$. We can run the clock in standard fashion by applying a unitary to $\ket{\psi_\off}$ which maps it to $\ket{\psi_\on}$ when the 1st event occurs, followed by measurement via an appropriate projection-valued measure when the 2nd event occurs. We use the convention that when the clock is on, it ticks at time $t_1>0$, and that the elapsed time between the 1st and 2nd events is at most $2 t_1$. Parameter $t_1$ can be chosen to be any value in our construction. Unlike with the elementary clock, we now have that the state $\ket{\psi_\on(t)}$ before the tick takes place at time $t_1$, will not be exactly orthogonal to the state after the tick occurs. Therefore, in order to unambiguously predict whether the clock has ticked, one must perform an unambiguous quantum discrimination measurement \cite{Barnett09}. This is a projective measurement with three possible outcomes: clock has not ticked yet, it has ticked, or I do not know. The last outcome can be thought of as an error, since when it is obtained we cannot say what time it is. This setting allows for more flexibility than in the prior clock setups. In the counterfactual case, this aspect will not be detrimental to its functioning, since, as with previous cases, the \ifreeoutcomes{} will only occur with some probability.

We will now explain how to run this clock in a counterfactual manner. The protocol proceeds similarly to in prior cases: we start the clock in $\ket{\psi_\off}$ and apply a unitary which maps it to a suitable superposition of off and on when the 1st event occurs, namely to $c\ket{\psi_\off}+s\ket{\psi_\on}$. When the 2nd event occurs at some unknown time $t\in[0, 2 t_1)$, we apply a unitary $U_\measure$ and measure using the projection-valued measure $\proj{\psi_\off}$, $\proj{A}$,  $\text{\large $\id$} -\proj{\psi_\off}-\proj{A}$, where $\ket{A}$ is stationary under the Hamiltonian evolution. After applying $U_\measure$ to the on and off clock states at time $t$, the states take on the form
\begin{subequations}
	\begin{align}
	U_\measure \ket{\psi_\off}= & \frac{A_1}{N} \ket{\psi_\off}+ \frac{A_1}{N} \ket{A} +A_2 \ket{A_\off},\label{eq:psi off t}\\
	U_\measure \ket{\psi_\on(t)}= &
	-\frac{c}{s} A_1 \ket{\bar \psi_\off(t)} -\frac{c}{s} A_1 \ket{\bar A(t)} +A_3 \ket{A_\on(t)} ,\label{eq:psi on t}
	\end{align}
\end{subequations}
where $A_2=\sqrt{1-2(A_1/N)^2}$, $A_3=\sqrt{1-2 (c/s)^2 A_1^2}$ are normalisation parameters and the other parameters will be discussed shortly. All kets in the superpositions are orthonormal except for the overlaps $\braket{\psi_\off| \bar \psi_\off(t)}$, $\braket{A| \bar A(t)}$ and  $\braket{A_\off|A_\on(t)}$. The kets $\ket{A_\off}$, $\ket{A_\on(t)}$ play the role of ancilla states that are chosen to guarantee $\braket{\psi_\off | \psi_\on (t)}=0$ holds at all times. The overlaps $\braket{\psi_\off | \bar \psi_\off(t)}$, $\braket{A| \bar A(t)}$ are chosen such that the counterfactual clock functions properly: consider the state of the clock just before the measurement 
\begin{align}\label{eq: Um applied to engineered clock}
U_\measure\me^{-\mi t\hat H} \left( c \ket{\psi_\off} +s \ket{\psi_\on}\right)= c\, U_\measure\ket{\psi_\off} +s\, U_\measure\ket{\psi_\on(t)},
\end{align}
there the latter quantities are provided by \cref{eq:psi off t,eq:psi on t}. If the 2nd event occurs in interval $t\in[0,t_1)$, we require
\begin{subequations}
	\begin{align}\label{eq:x 0 overlap}
	\braket{\psi_\off | \bar \psi_\off(t)}=
	\begin{cases}
	0 &\mbox{ for } t\in[0,t_1)\\
	1/N &\mbox{ for } t\in[t_1,2t_1)\,,
	\end{cases}
	\end{align}
	in order to be sure that the clock was off and  $t\in[0,t_1)$, when the measurement outcome $\psi_\off$ is obtained (This can readily be seen from \cref{eq:psi on t,eq:psi off t,eq: Um applied to engineered clock,eq:x 0 overlap} and the same \ifree{} measurement reasoning presented in the analysis of the elementary timekeeping systems). Similarly, if the 2nd event occurs in time interval $t\in[t_1,2t_1)$, we require
	\begin{align}\label{eq:x 1 overlap}
	\braket{A| \bar A(t)}=
	\begin{cases}
	1/N &\mbox{ for } t\in[0,t_1)\\
	0 &\mbox{ for } t\in[t_1,2t_1)\,,
	\end{cases}
	\end{align}
\end{subequations}
in order to be sure that the clock was off and $t\in[t_1,2t_1)$, when the outcome $A$ is obtained, according to a \ifree{} measurement. The role of $N$ now becomes apparent: it quantifies the overlap between the projectors associated with the \ifreeoutcomes{} and the dynamical kets $\ket{\bar \psi_\off(t)}$ and $\ket{\bar A(t)}$.\Mspace

\begin{figure}
	\includegraphics[scale=0.31]{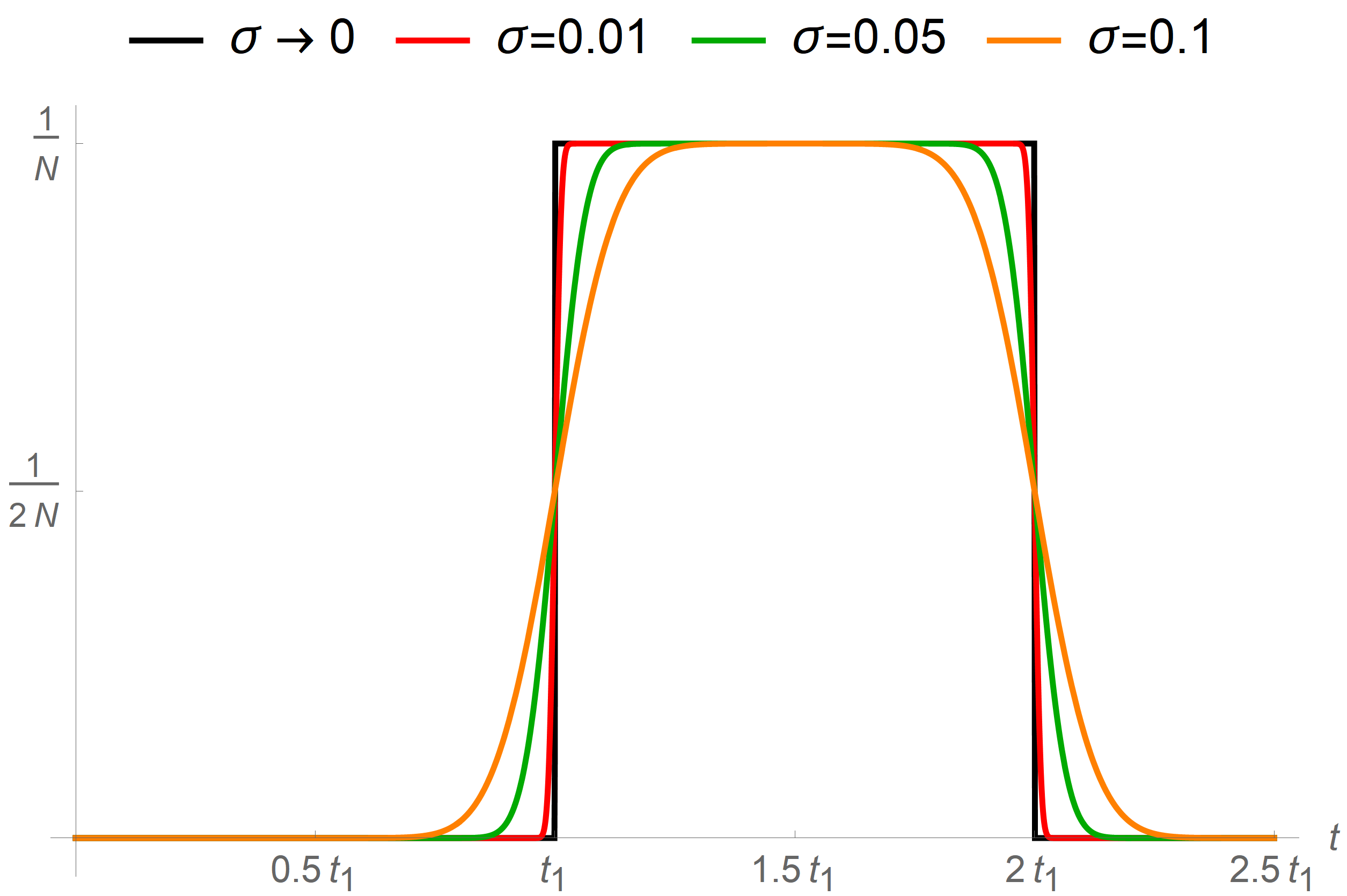}
	\centering
	\caption{Plot of the overlap $\braket{\tilde x_0|\bar x_0(t)}$ in \cref{eq:overlaps as a fucion of sigma} for different values of $\sigma$ \big(note that $\braket{\tilde x_1|\bar x_1(t)}$ is identical upto a shift to the left by an amount $t_1$\big). Observe how the difference between the ideal case (black) and that of $\sigma=0.01$ (red) are practically indistinguishable.  The small deviation between the black and coloured plots induces a small error. Note how the error is centred around $t_1$, which is the time at which the clock would tick if it were on, and $2t_1$, which is the time interval corresponding to the 1st tick would ends if the clock were on.}\label{fig:erfV1}
\end{figure}

It readily follows that the probability of knowing whether the 2nd event occurred in time interval $[0,t_1)$ or $[t_1, T_0)$ when the clock was off, that is to say, the probability of obtaining measurement outcomes $\psi_\off$ or $A$ in our protocol, is $P_\textup{cf}= 2 c^2 A_1^2/N^2$. One would want to choose $A_1, s$ and $N$ such as to maximise this probability while satisfying all the constraints. It happens that one can find states $\ket{\psi_\off}$, $\ket{\psi_\on}$, and a Hamiltonian $\hat H$, such that all the constraints can be satisfied other than \cref{eq:x 0 overlap,eq:x 1 overlap}, which appears to lead to unnormalisable states $\ket{\psi_\off}$, $\ket{A}$. From a physical perspective, the problem appears to be related to requiring the overlaps to transition from zero to a finite constant instantaneously while having a time evolution governed by a time independent Hamiltonian. However, a minor modification resolves the conflict: One can replace the r.h.s. of \cref{eq:x 0 overlap,eq:x 1 overlap} with approximate version using the \emph{error function} $\erf$. Specifically, $\braket{\psi_\off| \bar \psi_\off(t)}=f_0(t)$, $\braket{A| \bar A(t)}=f_1(t)$, with
\begin{align}
f_l(t)=  \frac{1}{2N}\! \left[ \erf\!\left(\frac{t/t_1\!+l\!-\!1}{\sqrt{2} \sigma}\right)\!- \erf\!\left(\frac{t/t_1\!+l\!-\!2}{\sqrt{2} \sigma}\right)\!\right],\label{eq:overlaps as a fucion of sigma}
\end{align}
and where $\sigma>0$ controls the approximation. In the limit $\sigma \to 0$ the above overlaps are equal to \cref{eq:x 0 overlap,eq:x 1 overlap}. We can now find solutions for all $\sigma>0$. We plot \cref{eq:overlaps as a fucion of sigma} in \cref{fig:erfV1} for varying approximation values $\sigma>0$.
We can quantify the error by the difference in probabilities associated with the \ifreeoutcomes, between the actual clock state which is measured, and an ``idealised'' clock state: $\textup{Dif}_p(\sigma,t):=\Big|\, \big|  \bra{x_p} \big(c \ket{\psi_\off}+s \ket{\psi_\on(t)}\!\big) \big|^2 -\big|  \bra{x_p} \big(c \ket{\psi_\off}+s \ket{\psi_\on I(t)}\!\big) \big|^2\, \Big|$, where we have used the short-hand ${x_0}={\psi_\off}$ for $p=0$, ${x_1}={A}$ for $p=1$. Here the ``idealised on state'', namely $ \ket{\psi_\on I(t)}$, is identical to $ \ket{\psi_\on (t)}$ given by \cref{eq:psi on t}, except that the kets $\ket{\bar \psi_\off(t)},\ket{\bar A(t)}$, satisfy \cref{eq:x 0 overlap,eq:x 1 overlap} rather than \cref{eq:overlaps as a fucion of sigma}. Importantly, the quantification of the error is meaningful, since if one were to use this idealised state in our protocol, it would result in zero false predictions: it would always predict the correct time and the clock would have always been off when a \ifreeoutcome{} had been obtained. Additionally, since the time $t$ at which one measures the clock could be any time in $[0,2 t_1)$, we time-average $\textup{Dif}_p(\sigma,t)$ over the time interval $[p t_1,(p+1)t_1)$ in which the outcome $x_p$ should have occurred, resulting in what we call the type-1 error (This error is thus due to the clock predicting the correct time, but being on); and we time-average over the remaining time, in which the outcome $x_p$ should not have occurred, resulting in what we call the type-2 error (This error is thus due to the clock predicting the incorrect time). In the case considered here, in which the clock ticks once when turned on, one only needs to consider the type-1 and type-2 errors for the $p=0$ case, since the error types are identical for the $p=1$ case. Hence we will forgo the label $p$ in the following discussion.\Mspace

We show in the \app{}~\ref{Precision quantification} that the difference in fidelities can be made arbitrarily small for both error types, while always having a non zero probability of obtaining the \ifreeoutcome. In particular, if we choose $\sigma=0.019$, we can obtain a total \ifree{} probability of half what it was in the previous optimal case, namely to $1/6$, and only incur type one and two errors of $1.0\times 10^{-3}$ and $3.7\times 10^{-4}$ respectively. If we further reduce $\sigma$ to $\sigma=0.0012$, the total \ifree{} probability only drops by half again, to $1/12$, and the type one and two error probabilities are now merely $7.3\times 10^{-14}$ and $2.8 \times 10^{-4}$ respectively.

\subsection{Proof of \Cref{thm:time is substantival}}\label{Proof of thm:time is substantiva}

Here we prove the theorem from the main text:

\thrmSubstantival*

Note that while we will only need to consider the counterfactual clock explained in the main text (i.e. the $\NT=1$ case), the proof is analogous in the case one uses the clocks generalised for $\NT\geq 2$ found in \app~\ref{sec app: clocks in nature}. In the latter case, the definitions of a \emph{\ifree{} measurement} and of an \emph{always off} clock which are found in the main text and used in assumptions (A) and (B), would be substituted for the more general definitions found in the \app; i.e. \cref{def:IFM,def:IFM 2}.

\begin{proof}
	Assumption (A) together with the example of a counterfactual clock for $\NT=1$ provided in the main text (the values for $U_\measure$ and $c,s$ which fully determine the protocol are provided at the end of \ref{sec:Measurement basis characterisation} in the \app), prove that for any two elapsed times $\tau_0$, $\tau_1$ between two events, a counterfactual clock exists and can determine said elapsed time with the clock always off. Therefore, the theorem follows by invoking assumption (B). 
\end{proof}

\newpage
\appendix


\begin{center}
	{\huge \Apps}
\end{center}

\section{Elementary time keeping devices}\label{sec app: clocks in nature}
In subsections \ref{sec:Setup and protocol} and \ref{sec:Interaction-free measurement (Formal definition)}, we detail the protocol one implements when using the clock to tell the time counterfactually. In subsection \ref{sec:Measurement basis characterisation}, we characterise the measurement basis used at the end of the protocol from subsections \ref{sec:Setup and protocol} and \ref{sec:Interaction-free measurement (Formal definition)}. In subsection \ref{sec:Backwards in time analysis in the multiple tick scenario} we prove some consequences of the protocols and definitions made in the previous subsections. Finally in subsection \ref{sec:Optimality and achievability proofs}, we provide two structure theorems which prove that ancillas are necessary and that counterfactual clocks exist for arbitrary $\NT\in\nnp$.

\subsection{Setup and protocol}\label{sec:Setup and protocol}

In this subsection, we explain the full protocol which allows one to determine the time $t\in\{\tau_0, \tau_1, \tau_2,\ldots, \tau_\NT\}$ using \ifree{}  measurements. 
For the purpose of performing general projective measurements, we append an $m$-dimensional ancilla space $\mathcal{H_\An}$ to the Hilbert $\mathcal{H}_\Sys$ of the system, leading to a total Hilbert space $\mathcal{H}_\Sys \oplus  \mathcal{H}_\An$. We will require the number of ancillas to be grater or equal to the number of distinguishable time intervals, namely $m\geq \NT$. The ancilla states are always stationary, namely they do not change in time. When the clock is tuned on, it is initialised to another state denoted $\ket{\psi}$. Thus in total, we have the orthonormal basis states
$\big\{ \ket{\psi}, \ket{E}, \ket{\tau_0}, \ket{\tau_1}, \ldots, \ket{\tau_\NT}, \ket{A_1}, \ket{A_2}, \ldots, \ket{A_m} \big\}$, which we call the measurement basis (Since the clock never has support on the basis state $\ket{psi}$ when measured, it will often be omitted for simplicity). Observe that the $m$ ancillas could be other energy eigenstates of the system which are orthogonal to the states $\ket{\tau_0}, \ket{\tau_1},\ldots, \ket{\tau_\NT}, \ket{E}$, or they could be produced via a separate $m+1$ dimensional system via the identification $\ket{\tau_0}\equiv \ket{\tau_0}\otimes \ket{d_0}, \ket{\tau_1}\equiv \ket{\tau_1}\otimes \ket{d_0},\ldots, \ket{\tau_\NT}\equiv \ket{\tau_\NT}\otimes \ket{d_0}$, and $\ket{E}\equiv \ket{E}\otimes \ket{d_0}$, for original states of the system and $\ket{A_1}\equiv \ket{E}\otimes \ket{d_1}, \ket{A_2}\equiv\ket{E}\otimes \ket{d_2}, \ldots, \ket{A_m}\equiv \ket{E}\otimes \ket{d_m}$, for the ancillary states, where $\{ \ket{d_0}, \ket{d_1}, \ldots,  \ket{d_m} \}$, are orthonormal states of the separate $m+1$ dimensional system.

Initially we set the system to be in the energy eigenstate, $\ket{E}$. Then, when the 1st event occurs, a unitary is applied to take the state to a superposition of $\ket{E}$ and $\ket{\psi}$. Using the branching notation from the main text to distinguish between the two orthogonal states, this has the form: 
\begin{align}
\begin{tikzpicture}
\tikzstyle{level 1}=[level distance=1cm, sibling distance=1.5cm]
\tikzstyle{end} = [circle, minimum width=3pt,fill, inner sep=0pt]
\node (root) {$\ket{E}$}[grow'=right]
child {
	node [end, label=right:{$c\ket{E}$}] {}}
child {	node [end, label=right:{$s\ket{\psi}$} ] {}};
\end{tikzpicture}
\end{align}
where $c:=\cos \theta$ and $s:=\sin{\theta}$. We then wait until the 2nd event occurs at unknown time $\tau_l$. The system now is of the form
\begin{align}\label{eq:app:before meas}
\begin{tikzpicture}
\tikzstyle{level 1}=[level distance=1cm, sibling distance=1.5cm]
\tikzstyle{end} = [circle, minimum width=3pt,fill, inner sep=0pt]
\node (root) {$\ket{E}$}[grow'=right]
child {
	node [end, label=right:{$c\ket{E}$ \myarrow{\tau_l} $c\ket{E}$}] {}}
child {	node [end, label=right:{$s\ket{\psi}$ \myarrow{\tau_l} $s\ket{\tau_l}$} ] {}};
\end{tikzpicture}
\end{align}
where $\tau_l\in\{ \tau_0, \tau_1,\ldots, \tau_\NT\}$. We now perform a unitary $U_\measure$ over the system. The quantum states have support on the ancillas after the application of the unitary.  
The system now takes on the form:
\begin{align}\label{eq:app:genreal proto}
	\begin{tikzpicture}
	\tikzstyle{level 1}=[level distance=1cm, sibling distance=1.5cm]
	\tikzstyle{end} = [circle, minimum width=3pt,fill, inner sep=0pt]
	\node (root) {$\ket{E}$}[grow'=right]
	child {
		node [end, label=right:{$c\ket{E}$ \myarrow{\tau_l} $c\ket{E}\overset{U_\measure\,}{\xrightarrow{\hspace*{0.7cm}}}\,\, c A_0^0\ket{E}+\sum_{j=1}^{\NT} cA_{j}^0\ket{A_j}+\sum_{k=\NT+1}^m cB_{-\!1,k}\ket{A_{k}}+\sum_{k=1}^{\NT+1} cB_{-\!1,k}\ket{\tau_{k\!-\! 1}}$}] {}}
	child {	node [end, label=right:{$s\ket{\psi}$ \myarrow{\tau_l} $s\ket{\tau_l}\overset{U_\measure\,}{\xrightarrow{\hspace*{0.7cm}}}\,sA_0^1(\tau_l) \ket{E}+\sum_{j=1}^\NT s A_{j}^1(\tau_l)\ket{A_{j}}+\sum_{k=\NT+1}^m s B_{l,k}\ket{A_{k}}+\sum_{k=1}^{\NT+1} s B_{l,k}\ket{\tau_{k\!-\! 1}}$} ] {}};
	\end{tikzpicture}
\end{align}
where the amplitudes $A_q^0$, $A_q^1(\tau_l)$ are to be characterised in \app~\ref{sec:Measurement basis characterisation}. To finalise the protocol, we immediately measure in the measurement basis after $U_\measure$ is applied. Here we will use the convention that the \ifreeoutcomes{} are associated with the post-measurement states $\big\{\ket{E}, \ket{A_1},\ldots, \ket{A_{\NT}} \big\}$. 

\subsection{\Ifree{} measurement (Formal definition)}\label{sec:Interaction-free measurement (Formal definition)}
Given the presentation of the general protocol in \app~\ref{sec:Setup and protocol}, we can now formally define the \ifree{} measurement which is appropriate for our setting. These form a generalisation of the main text to the $\NT>1$ case. In the following, we label the energy $E$ by $A_0:=E$ for simplicity of notation (since the symbol $A_0$ has not been used previously, this should not lead to confusion).
\begin{definition}[\Ifree{} measurement for clocks for $\NT\in\nnp$]\label{def:IFM}
	We call a projective measurement a \emph{\ifree{} measurement} if it has outcomes $A_0$, $A_1$, $A_2$, \ldots, $A_\NT$ (e.g. post-measurement states $\ket{A_0}$, $\ket{A_1}$, $\ket{A_2}$, \ldots, $\ket{A_\NT}$ respectively) of which at least one of these outcomes is a \emph{\ifreeoutcome}. Furthermore, we say that the outcome $A_k$ is a \ifreeoutcome{} if the following two equations are both satisfied.
	\begin{align}\label{eq:cont 1 2}
	A_k^1(\tau_k)=0 , 
	\end{align}
	and 
	\begin{align}\label{eq:cont 2 2}
	A_k^1(\tau_l)=- \frac{c}{s}A^0_k 
	\end{align}
	for all $l= 0,1,2,\ldots, k-1,k+1,\ldots, \NT$ and where $k\in \{0,1,2,\ldots, \NT\}$.  
\end{definition}
\begin{definition}[Clock always off for $\NT\in\nnp$]\label{def:IFM 2}
	When a \ifreeoutcome{} $A_k$ is obtained, we say that the elapsed time between the two events was $\tau_k$, and that the clock was \emph{always off}.
\end{definition}
	The reasoning behind these definitions is based on the usual arguments of \ifree{} measurements: \cref{eq:cont 1 2} allows us to conclude that if one obtains measurement outcome $A_k$ \emph{and} $t=\tau_k$, the clock must have been always off, since the amplitude in the on branch associate with the ket $\ket{A_k}$ is zero during said time interval (see \cref{eq:app:genreal proto}). Meanwhile, \cref{eq:cont 2 2} guarantees that if $\tau_k$ is obtained, then we \emph{must} have $t=\tau_k$, since the amplitude  $A_k^1(\tau_l)$ cancels out with the amplitude $A_k^0$ for all times $t\neq \tau_k$, (see \cref{eq:app:genreal proto}). In \app~\ref{sec:Optimality and achievability proofs} we show that these definitions are not vacuous, specifically, that clocks exist for any $\NT\in\nnp$ with the probabilities of obtaining their \ifreeoutcomes{} all non-zero.  
	In \app~\ref{sec:Backwards in time analysis in the multiple tick scenario} we will see that these definitions have other interpretational implications.

\subsection{Characterisation of unitary $U_\measure$.}\label{sec:Measurement basis characterisation}

We consider the case in which we associate all of the measurement outcomes $E$, $\tau_1$, $\tau_2$, \ldots, $\tau_\NT$ with the elapsed time being $\tau_0$, $\tau_1$, $\tau_2$, \ldots, $\tau_\NT$ respectively, and the clock being off (this is to say, we associate $E$, $\tau_1$, $\tau_2$, \ldots, $\tau_\NT$ with \ifreeoutcomes). The constraints imposed by \cref{eq:cont 1 2,eq:cont 2 2} constrain the matrix coefficients of the unitary matrix $U_\measure$. Specifically, by writing them out explicitly, we have that $U_\measure$ is of the form $U_\measure=U_\textup{ex} U_\measure'$ where $U_\textup{ex}$ implements a simple change of basis, defined via the relations: $U_\textup{ex} \ket{E}=\ket{E}$, $U_\textup{ex} \ket{\tau_0}=\ket{A_1}$, $U_\textup{ex} \ket{\tau_1}=\ket{A_2}$, $\ldots$, $U_\textup{ex} \ket{\tau_{\NT-1}}=\ket{A_\NT}$, $U_\textup{ex} \ket{A_1}=\ket{\tau_0}$, $U_\textup{ex} \ket{A_2}=\ket{\tau_1}$,  $\ldots$, $U_\textup{ex} \ket{A_\NT}=\ket{\tau_{\NT-1}}$, and $U_\textup{ex} \ket{\tau_\NT}=\ket{\tau_\NT}$, $U_\textup{ex} \ket{A_{\NT+1}}=\ket{A_{\NT+1}} $, $U_\textup{ex} \ket{A_{\NT+2}}=\ket{A_{\NT+2}}$, \ldots, $U_\textup{ex} \ket{A_{m}}=\ket{A_{m}}$. The other matrix $U_\measure'$ has the following matrix representation:
\begin{align}\label{table:generix form}
&\resizebox{1 \textwidth}{!} 
{
	$\begin{matrix}[c|cccccccccc]
	{}  & \ket{E} & \hspace{0.1cm}\ket{\tau_0} & \hspace{0.1cm}\ket{\tau_1} & \hspace{0.1cm}\ket{\tau_2} & \quad  \ldots\quad & \hspace{-0.1cm}\ket{\tau_{\NT-1}}
	& \hspace{-0.1cm}\ket{\tau_\NT}   & \quad \hspace{-0.6cm}\ket{A_1}  & \quad \ldots
	& \hspace{0.1cm} \ket{A_m} \\
	\hline
	\ket{E}  & A_0^0 &  A_1^0  &  A_2^0 &  A_3^0 & \quad \ldots\quad & A_\NT^0 & \gamma_0 \quad   & \hspace{-1.88cm}\hspace{0.91cm} \large{\brokenvert}\hspace{0.25cm}\hspace{0.5cm} B_{-\!1,1} & \quad \ldots  & B_{-\!1,m} \\ 
	\ket{\tau_0}   & 0  &  - A_1^0\, r  &  - A_2^0\, r &  - A_3^0\, r & \quad \ldots \quad & - A_\NT^0\, r & \gamma_1 \quad   &  \hspace{-2cm}\hspace{0.9cm}\large{\brokenvert}\hspace{0.1cm}\,\,\,\hspace{0.5cm} B_{1,1}  &  \quad \ldots & B_{1,m} \\ 
	\ket{\tau_1}  & -A_0^0\, r &  0  &    - A_2^0\, r & - A_3^0\, r & \quad \ldots \quad & - A_\NT^0\, r & \gamma_2 \quad   &  \hspace{-2cm}\hspace{0.9cm}\large{\brokenvert}\hspace{0.1cm}\,\,\,\hspace{0.5cm} B_{2,1}  &  \quad \ldots & B_{2,m} \\ 
	\ket{\tau_2}  & - A_0^0\, r  &  - A_1^0\, r  &    0 & - A_3^0\, r & \quad \ldots \quad & - A_\NT^0\, r & \gamma_3 \quad   &  \hspace{-2cm}\hspace{0.9cm}\large{\brokenvert}\hspace{0.1cm}\,\,\,\hspace{0.5cm} B_{3,1}  &  \quad \ldots & B_{3,m} \\ 		
	\vdots  & \vdots & \vdots & \vdots & \vdots &  & \vdots & \hspace{-0.3cm}\vdots & \hspace{-0.3cm} \vdots  &  &  \vdots\\
	\ket{\tau_{\NT-1}} & - A_0^0\, r  &  - A_1^0\, r  &   - A_2^0\, r  & - A_3^0\, r & \quad \ldots\quad & - A_\NT^0\, r &  \gamma_\NT   & \hspace{0.27cm}\hspace{-2cm} \hspace{0.9cm}\large{\brokenvert}\hspace{0.1cm}\,\,\,\hspace{0.5cm}  B_{\NT,1}  &  \quad \ldots & B_{\NT,m} \\ 
	\ket{\tau_{\NT}} &- A_0^0\, r  &  - A_1^0\, r  &    - A_2^0\, r & - A_3^0\, r & \quad \ldots\quad & 0 &   \gamma_{\NT+1}   &  \,\,\,\,\hspace{0.38cm}\hspace{-2cm}\hspace{0.9cm}\large{\brokenvert} \hspace{0.1cm}\hspace{0.17cm}\hspace{0.5cm} B_{\NT+1,1}  &  \quad \ldots & B_{\NT+1,m}\vspace{0.1cm}	\\ 
	\cline{2-8} 
	\ket{A_1} & B_{\NT+2,0}  &  B_{\NT+2,1}  &   B_{\NT+2,2} & B_{\NT+2,3} & \quad\ldots\quad & B_{\NT+2,\NT} &   B_{\NT+2,\NT+1}   & \hspace{0.49cm}\hspace{-1cm} \hspace{0.5cm} B_{\NT+2,\NT+2}\hspace{-0.5cm}  &  \quad \ldots & B_{\NT+2,\NT+m+1}	\\
	\vdots  & \vdots & \vdots & \vdots & \vdots &  & \vdots & \vdots & \hspace{-0.8cm} \vdots  &  &  \vdots\\
	\ket{A_{m}}  & B_{\NT+m+1,0}  &  B_{\NT+m+1,1}  &   B_{\NT+m+1,2} & B_{\NT+m+1,3} & \quad \ldots\quad & B_{\NT+m+1,\NT} &   B_{\NT+m+1,\NT+1}   &  \hspace{0.49cm}\hspace{-0.6cm} \hspace{0.5cm}B_{\NT+m+1,\NT+2}\hspace{-0.5cm}   &  \quad \ldots & B_{\NT+m+1,\NT +m+1}	
	\end{matrix}
	$
}\nonumber\\
&\vspace{15cm} {}^{}
\end{align}
where we have denoted  $r:=c/s=\cos(\theta)/\sin(\theta)$, and all matrix entries are arbitrary complex numbers such that $U_\measure' \,{U_\measure'}^{\!\dag}={U_\measure'}^{\!\dag}\, U_\measure' =\id$. In writing matrix \cref{table:generix form}, we have used the convention that kets are row vectors and bras are column vectors. We will used this matrix representation convention throughout. The horizontal solid inner line and the vertical dotted inner line are visual aids only. \Mspace

The probability of finding the \rg{} off and the time to be $\tau_l$ upon measurement is $P_\cf^{(l)}=|c A^0_l|^2$, $l\in\nno$.

The example in the main text where there are just two distinguishable times (i.e. $\NT=1$) and requires one ancilla state corresponds to the unitary matrix $U_\measure=U_\textup{ex} U_\measure'$, with $U_\textup{ex}$ determined by $U_\textup{ex} \ket{E}=\ket{E}$, $U_\textup{ex} \ket{\tau_0}=\ket{A}$, $U_\textup{ex} \ket{A}=\ket{\tau_0}$,  $U_\textup{ex} \ket{\tau_1}=\ket{\tau_1}$. The unitary $U_\measure'$ is given by
\begin{subequations}
\begin{equation}\label{table:special form}
\begin{matrix}
\hspace{0.66cm}\vline & \ket{E} & \quad\ket{\tau_0} & \quad\ket{\tau_1}  & \quad \ket{A} \\
\hline\ket{E}\hspace{0.15cm} \vline & \sqrt{1/3} &  \sqrt{1/3}  &  \sqrt{1/3} &  \hspace{-0.73cm}\hspace{0.0cm} \large{\brokenvert}\hspace{0.45cm} 0 \\ 
\ket{\tau_0}\,\, \vline & 0  &  -\sqrt{1/3}  &  \sqrt{1/3}   &\hspace{0.24cm}  \hspace{-0.3cm}\large{\brokenvert}\hspace{0.29cm}\hspace{-0.3cm} -\sqrt{1/3} \\ 
\ket{\tau_1} \,\,\vline & -\sqrt{1/3}  &  0  &    \sqrt{1/3}  & \hspace{-0.45cm} \hspace{0.55cm}\hspace{-0.3cm}\large{\brokenvert}\hspace{0.39cm}\hspace{-0.1cm} \sqrt{1/3} 	\\
\cline{2-4} 
\ket{A\hspace{0.01cm}}\hspace{0.155cm} \vline & \sqrt{1/3}  &  -\sqrt{1/3}  &  0 &  \hspace{0.3cm}\hspace{-0.1cm} \sqrt{1/3}\\
\end{matrix}
\end{equation}
and coefficients $c$ and $s$ by
\begin{align}
c=s=1/\sqrt{2}.\label{eq:eq:c equals s equal 1}
\end{align}
\end{subequations}

\subsection{Backwards in time analysis in the $\NT\geq 2$ scenario}\label{sec:Backwards in time analysis in the multiple tick scenario}
Here we extend the analysis presented in \cref{sec:backwards in time analysis} to the general case of multiple distinguishable times ($\NT\in\nnp$). Since the results are quantitatively the same as in~\cref{sec:backwards in time analysis}, we will not discuss at length the interpretation to avoid repetition. The amplitude corresponding to the $k$th \ifreeoutcome{} is 
\begin{align}\label{eq:two time state general}
\braket{A_k| U_\measure U(\tau_k) U_0 |A_0},
\end{align}
for $k=0,1,2,\ldots, \NT$ and where we are again using the convention $A_0 :=E$.  Using $U_\measure=U_\textup{ex} U_\measure'$ and~\cref{table:generix form}, we have 
\begin{align}\label{eq:A k on U m general}
\begin{split}
\bra{A_k} U_\measure = & A_k^0 \bra{A_0}-\!A_k^0 r\bra{\tau_0} - \!A_k^0 r\bra{\tau_1}-\ldots-\!A_k^0 r\bra{\tau_{k-1}} - \!A_k^0 r\bra{\tau_{k+1}}-\! A_k^0 r\bra{\tau_{k+2}}- \!A_k^0 r\bra{\tau_{k+3}}-\ldots -\! A_k^0 r\bra{\tau_\NT} \\
& + B_{\NT+2,k} \bra{A_1} + B_{\NT+2,k} \bra{A_2}+ B_{\NT+3,k} \bra{A_3}+ \ldots+ B_{\NT+m,k} \bra{A_m}  
\end{split} 
\end{align}
where $\bra{\tau_{-1}}:=0$. Recall that $U(\tau_k)\ket{A_0}=\ket{A_0}$, $U(\tau_k)\ket{A_1}=\ket{A_1}$, $U(\tau_k)\ket{A_2}=\ket{A_2}$, \ldots, $U(\tau_k)\ket{A_m}=\ket{A_m}$, and $U(\tau_k)\ket{\psi}=\ket{\tau_k}$. Therefore, analogously to \cref{sec:backwards in time analysis}, it follows that $\bra{\tau_l}U(\tau_k)=\bra{\psi_\perp^{(l,k)}}$ for all $k,l=0,1,\ldots,\NT$ with $k\neq l$ , where $\{ \bra{\psi_\perp^{(l,k)}} \}_{l,k}$ are orthogonal to $\ket{A_0}$ and $\ket{\psi}$. This follows from writing the ansatz $\bra{\tau_l}U(\tau_k)= c_0^{(l,k)} \bra{\psi_\perp^{(l,k)}}+c_2^{(l,k)}\bra{A_0}+c_3^{(l,k)}\bra{\psi}$, followed by observing the two equalities $\delta_{k,l}=\braket{\tau_l|\tau_k}= \braket{\tau_l|U(\tau_k)U(\tau_k)^\dag|\tau_k}=\big( c_0^{(l,k)} \bra{\psi_\perp^{(l,k)}}+c_2^{(l,k)}\bra{A_0}+c_3^{(l,k)}\bra{\psi} \big)\ket{\psi}= c_3^{(l,k)}$, where $\delta_{k,l}$ is the 
Kronecker delta, and $0=\braket{\tau_l|A_0}= \braket{\tau_l|U(\tau_k)U(\tau_k)^\dag|A_0}=\big( c_0^{(l,k)} \bra{\psi_\perp^{(l,k)}}+c_2^{(l,k)}\bra{A_0}+c_3^{(l,k)}\bra{\psi} \big)\ket{A_0}= c_2^{(l,k)}$. Thus from~\cref{eq:A k on U m general} we arrive at
\begin{align}\label{eq:A k on U m general 2}
\begin{split}
\bra{A_k} U_\measure U(\tau_k) = & A_k^0 \bra{A_0}-\!A_k^0 r\bra{\psi_\perp^{(0,k)}} - \!A_k^0 r\bra{\psi_\perp^{(1,k)}}-\ldots-\!A_k^0 r\bra{\psi_\perp^{(k\!-\! 1,k)}} - \!A_k^0 r \bra{\psi_\perp^{(k+1,k)}}-\! A_k^0 r\bra{\psi_\perp^{(k+2,k)}}\\
 &-\ldots -\! A_k^0 r\bra{\psi_\perp^{(\NT,k)}}  + B_{\NT+2,k} \bra{A_1} + B_{\NT+2,k} \bra{A_2}+ B_{\NT+3,k} \bra{A_3}+ \ldots+ B_{\NT+m,k} \bra{A_m}\\
 &=   A_k^0 \bra{A_0} + \bra{\psi_\perp''^{(k)}},
\end{split} 
\end{align}
where $\bra{\psi_\perp''^{(k)}}$ is another (unnormalised) state orthogonal to $\ket{A_0}$ and $\ket{\psi}$.
Thus recalling that $U_0$ only acts nontrivially on the subspace spanned by $\ket{A_0}$ and $\ket{\psi}$,  we find that 
\begin{align}\label{eq:A k on U m general 3}
\begin{split}
\bra{A_k} U_\measure U(\tau_k) = & A_k^0 \big( c\bra{A_0}-s\bra{\psi} \big)+ \bra{\psi_\perp''^{(k)}}.
\end{split} 
\end{align}
Thus upon pre-selecting onto the initial energy eigenstate $\ket{A_0}$, we observe that the only terms contributing to $\braket{A_k| U_\measure U(\tau_k)|A_0}$ come from non dynamical branches of the wave function.\Mspace

As with the case in which $\NT=1$ presented in \meth~\ref{sec:backwards in time analysis}, we can analyse the probability that a measurement of whether the clock was dynamically evolving during the dynamical stage of our pre and post selected setup. For the case of post-selecting on \ifreeoutcome{} $\ket{A_k}$, following~\cite{TwoStateFormalisumAharonov2008}, we have the probability of the clock being on at time $t\in(0,\tau_k)$ is given by 
\begin{align}
\textup{Prob}\big(\on, A_k\big)&:= \frac{\big{|}\braket{A_k|U_\measure U(\tau_k-t) P_\on U(t) U_0 |A_0}\big{|}^2 }{\big{|}\braket{A_k|U_\measure U(\tau_k-t) P_\on U(t) U_0 |A_0}\big{|}^2+ \big{|}\braket{A_k|U_\measure U(\tau_k-t) P_{{\scriptscriptstyle \backslash}\on} U(t) U_0 |A_0}\big{|}^2},
\end{align}
where $P_{\on}:=\proj{\psi}+\proj{\tau_0}+\proj{\tau_1}+\ldots+\proj{\tau_\NT}$ is the projection onto the on subspace and $P_{{\scriptscriptstyle \backslash}\on}:=\proj{A_0}+\proj{A_1}+\ldots+\proj{A_m}$ projects onto the complementary space (the subspace where no dynamics occurs). Note that since the ancillas and energy eigenstate are stationary, we have that $U(t)=U_{\perp A_0,\ldots, A_m}(t) + \id_{A_0,\ldots, A_m}$, where $\id_{A_0,\ldots, A_m}$ is the identity operator on the subspace spanned by $\ket{A_0}$, $\ket{A_1}$, \ldots, $\ket{A_m}$, and that $U_{\perp A_0,\ldots, A_m}(t)$ is orthogonal to said subspace. Therefore, if we assume that the dynamics is Markovian, namely $U(t_1)U(t_2)=U(t_1+t_2)$ for all $t_1, t_2\in\rr$, it follows that $U(\tau_k-t) P_\on U(t)=U_{\perp A_0,\ldots,A_m}(\tau_k)$ and hence that 
\begin{align}
\textup{Prob}\big(\on, A_k\big)=0
\end{align}
for all $t\in(0,\tau_k)$ and for all $k=0,1,2,\ldots,\NT$.

\subsection{Structure and achievability proofs}\label{sec:Optimality and achievability proofs}

In this subsection we present two propositions. The first one shows that ancillary states are necessary in order for the counterfactual clock to function, while the second shows that counterfactual clocks exist which can distinguish between an arbitrary number of times (i.e. for all $\NT\in\nnp$).

To start with, we need to motivate why it could have been that the ancillary states were not strictly necessary. The most straightforward way to see this is by recalling~\cref{eq:two time state general} from the previous section, and noting that the unitary $U_\measure$ has the form $U_\measure=U_\textup{ex} U_\measure'$. It follows that
\begin{align}\label{eq:two time state general 2}
\braket{A_k| U_\measure U(\tau_k) U_0 |E}= 
\begin{cases}
\braket{E| U_\measure' U(\tau_0) U_0 |E} &\text{if } k=0 \vspace{0.1cm}\\ 
\braket{\tau_{k-1}| U_\measure' U(\tau_k) U_0 |E} &\text{if } k=1,2,3,\ldots,\NT.
\end{cases}
\end{align} 
The physical interpretation of this is that rather than post-selecting on $\ket{E}=\ket{A_0}$ and ancillas $\ket{A_1}$, \ldots, $\ket{A_\NT}$, we could have implemented the unitary $U_\measure'$, in stead of $U_\measure$, and post-selected on $\ket{E}$ and $\ket{\tau_0}$, \ldots, $\ket{\tau_{\NT\!-\!1}}$.  In this latter method, since the states $\ket{\tau_k}$ evolve in time, upon implementing the \ifree{} measurement, the clock would start evolving unless it was immediately turned off (note that this does not happen when post-selecting onto the ancillas, since these are by definition stationary states). However, in the latter method, we see from the r.h.s. of~\cref{eq:two time state general 2} that one could have used unitaries in which the number of ancillas is zero, i.e. $m=0$, since the ancillas only appear implicitly (no pre or post selection onto them is required). The following proposition proves that in the case of $\NT=1$, the unitary $U_\measure'$ needs at least one ancilla. This result implies that the ancillas are necessary.

\begin{proposition}[Ancillas are necessary]\label{Ancillas are necessary}
Consider the counterfactual clock protocol described by the two state formalism on the r.h.s. of~\cref{eq:two time state general 2} for the case where the clock can distinguish between two times ($\NT=1$) and the number of ancilla states is zero ($m=0$), namely 
\begin{align}\label{eq:two time state general 3}
\braket{E| U_\measure' U(\tau_0) U_0 |E} \qquad\text{and }\qquad \braket{\tau_0| U_\measure' U(\tau_1) U_0 |E}.
\end{align} 
 There is no solution for which the probability of the elapsed time being $\tau_0$, namely $P_\cf^{(0)}:=|\braket{E| U_\measure' U(\tau_0) U_0 |E}|^2$, and the probability of the elapsed time being $\tau_1$, namely $P_\cf^{(1)}=|\braket{\tau_{0}| U_\measure' U(\tau_1) U_0 |E}|^2$, are both non-zero.
\end{proposition}
\begin{proof}
	For the case $\NT=1$, $m=0$, the matrix representation of $U_\measure'$ in \cref{table:generix form} reduces to
	\begin{equation}
	\begin{matrix}
	\hspace{0.67cm}\vline & \ket{E} & \hspace{-0.16cm}\quad\ket{\tau_0} & \quad\ket{\tau_1} \\
	\hline
	\ket{E}\hspace{0.16cm} \vline & A_0^0 &  A_1^0  &  \gamma_0  \\ 
	\ket{\tau_0}\hspace{0.13cm} \vline & 0  &  -A_1^0 r  &  \gamma_1    \\ 
	\ket{\tau_1} \hspace{0.13cm}\vline & -A_0^0\,r  &  0  &    \gamma_2  	\\
	\end{matrix}
	\end{equation}
	Unitary matrices require their row and column vectors to be orthogonal. Applying this constraint to the 1st two columns, one finds $A_0^0\, (A_1^0)^\dag=0$, which implies $|A_0^0| |A_1^0|=0$ and hence $A_0^0=0$ and/or $A_1^0=0$. Therefore, since $P_\cf^{(0)}= |c A_0^0|^2$, and $P_\cf^{(1)}= |c A_1^0|^2$, there is no solution for which $P_\cf^{(0)}>0$ and $P_\cf^{(1)}>0$.
\end{proof}

We now prove that for all $\NT\in\nn^+$, there exits a finite dimensional ancilla system and unitary $U_\measure$ such that there is a non zero probability of finding the clock off at all measurement times.

\begin{proposition}[Counterfactual clocks with arbitrarily many distinguishable times exist]\label{prop:clock existence}
Let $\NT\in\nnp$. Then, for any $\{ \tilde A^0_0, \tilde A^0_1,\tilde A^0_2, \ldots, \tilde A^0_\NT, \tilde \gamma_0, \tilde \gamma_1, \tilde \gamma_2, \ldots,\tilde\gamma_{\NT+1}    \}\in \ccr^{2\NT+3}$, 
and $r \neq 0$, there exists a $\gamma>0$ for which a unitary matrix of the form \cref{table:generix form} exists, with the amplitudes corresponding to \ifreeoutcomes{} given by $ A^0_0= \gamma \tilde A^0_0$,  $A^0_1= \gamma \tilde A^0_1$, $A^0_2= \gamma \tilde A^0_2$, \ldots,  $A^0_\NT= \gamma \tilde A^0_\NT$, and gamma coefficients given by $\gamma_0=\gamma\,\tilde \gamma_0$, $\gamma_1=\gamma\,\tilde \gamma_1$, $\gamma_2=\gamma\,\tilde \gamma_2$, \ldots, $\gamma_{\NT+1}=\gamma\,\tilde \gamma_{\NT+1}$, and $m=2(\NT+2)$ ancillary states.
\end{proposition}
\begin{proof} It is by construction, and hence can be used to work out particular values of $\gamma$ for any instance of the problem. Denote by $\vec F_j=\gamma \vec f_j\,$ ($j=1,2,3,\ldots,m$) the $j$th column vector of \cref{table:generix form}. As is well known, since \cref{table:generix form} is a square matrix, it follows that it is a unitary matrix if the vectors $\{\vec F_j\}_{j=1}^m$ form an orthonormal family. We therefore need to prove that the vectors $\vec F_j$ can be made orthonormal.
	
To start with, we will decompose the vector  $\{\vec f_j\}_{j=1}^{\NT+2}$ into a direct sum of three other vectors, namely for $j=1,2,3,\ldots, \NT+2$, let $\vec f_j = \vec e_j \oplus \vec X_j \oplus \vec Y_j$, where $\vec e_j, \vec X_j, \vec Y_j \in \ccr^{\NT+2}$. Since the matrix in \cref{table:generix form} is a square matrix, this choice fixes the number of ancillas to be $m=2(\NT+2)$. Note that the vectors $\{\vec e_j\}_{j=1}^{\NT+2}$ are completely fixed, up to the constant $\gamma$, by the given parameters in the proposition statement. Meanwhile the vectors $\vec X_j, \vec Y_j$ are complete undetermined at this stage. We first fix the $\vec Y_j$ vectors: for $j,k=1,2,3,\ldots, \NT+2$, let $[\vec Y_j]_k= \delta_{j,k} c_j$, where $\delta_{j,k}$ is the Kronecker delta, and the coefficients  $\{c_j\}_{j=1}^{\NT+2}\in\ccr^{\NT+2}$ are to be determined. We will now fix the $\vec X_j$ vectors. We will use them for so-called \emph{dimensional lifting} of the vectors $\vec e_j$ to an orthogonal set. The algorithm in~\cite{DimLift} shows how to choose vectors $\{ \vec x_j\}_{j=1}^{\NT+2}$ so that the vectors $\{\vec e_j\oplus \vec x_j\}_{j=1}^{\NT+2}$ form an orthogonal family for any given set $\{\vec e_j\}_{j=1}^{\NT+2}$. It thus follows due to the form of the vectors $\vec Y_j$, that the vectors $\{\vec F_j = \gamma\,\vec e_j \oplus \vec X_j \oplus \vec Y_j\}_{j=1}^{\NT+2}$ form an orthogonal family for all $\{c_j\}_{j=1}^{\NT+2}\in\ccr^{\NT+2}$ and for all $\gamma>0$. Imposing normalisation on the vectors $\{\vec F_j \}_{j=1}^{\NT+2}$, we find for $j=1,2,3\ldots, \NT+2$:
\begin{align}\label{eq:normalisation of F}
\frac{1}{\gamma^2}-|c_j|^2= (\vec X_j)^\dag\, \vec X_j+(\vec Y_j)^\dag \,\vec Y_j.
\end{align}
Now denote by $D:=\max_{j\in\{1,2,\ldots, \NT+2\}} \left\{ (\vec X_j)^\dag\, \vec X_j+(\vec Y_j)^\dag \,\vec Y_j\right\}$, and the value of $j$ which solves the maximization by $j^*$.  Furthermore, set $c_{j^*}=0$ so that it follows from \cref{eq:normalisation of F}, that $\gamma=1/\sqrt{D}$. We can solve \cref{eq:normalisation of F} for all $j\neq j^*$, with a solution for the coefficients $c_j$ satisfying $0\leq |c_j|^2\leq 1/\gamma^2=D$.

We now have an orthonormal set of vectors  $\{\vec F_j\}_{j=1}^{\NT+2}$, and all that remains is to find the remaining vectors $\{\vec F_j\}_{j=\NT+3}^{\NT+2+m}$, with $m=2(\NT+2)$. Since this latter set of vectors contains elements without any constraint on them, other than those which allow \cref{eq:normalisation of F} to be a unitary matrix, we can simply apply the Gram-Schmidt orthonormalization procedure on the input sequence $\big( \vec F_j \big)_{j=1}^{\NT+2}\,{}^\frown \big(\vec z_j\big)_{j=\NT+3}^{\NT+2+m}$, where $\{\vec z_j\}_{j=\NT+3}^{\NT+2+m}$ is an arbitrary set of vectors linearly independent of $\{\vec F_j\}_{j=1}^{\NT+2}$, and ${}^\frown$ denotes sequence concatenation. The output of the Gram-Schmidt orthogonalisation procedure is then the orthonormal family $\big( \vec F_j \big)_{j=1}^{3(\NT+2)}$ where the 1st $\NT+2$ elements are identical to the first $\NT+2$ orthonormal vectors of the input sequence to the Gram-Schmidt orthogonalisation procedure. This completes the proof.
\end{proof}

\section{Implementation of the proof of no classical model}\label{sec:non contextual proof numerics}
Here we describe the numerical implementation of the proof that there does not exist a non contextual ontic model for the counterfactual clock described in~\meth{} \ref{sec:non contextual proof}.
The entire problem is fully determined by the set $\sets_\cf$ of states (\cref{eq:states}) and the set $\sete_\cf$ of effects (\cref{eq:effects}), with $U_\measure$ and $c,s$ given by~\cref{table:special form,eq:eq:c equals s equal 1}.

We will follow the outline of the algorithm~\cite{Gitton2020} and perform the following steps. The inner product is taken to be the Hilbert-Schmidt inner product:
\begin{enumerate}
	\item Project the set of extreme points of the set $\sets_\cf$ onto $\textup{span}(\sete_\cf)$, were $\textup{span}$ denotes the linear span. Let $P_\sete(\sets)$ denote the resulting set vectors. We call $\textup{span}\big(P_\sete(\sets)\big)$ the \emph{Reduced space} and denote it by $\mathcal{R}$. It is the effective vector space of our clock. 
	\item Using the Gram-Schmidt orthogonalization procedure, construct an orthonormal basis for $\mathcal{R}$.
	\item Project the extreme points of the sets $\sets_\cf$ and $\sete_\cf$ onto $\cal R$. Denote the corresponding new sets as $P_{\cal R}(\sete)$ and $P_{\cal R}(\sets)$. The elements of the latter sets represent rays (also known as half-lines) emanating from the origin (which is the point $\mathbf{0}:=(0,\ldots,0)$ with dimension $\text{dim}(\mathcal{R})$ entries in the orthonormal basis for $\mathcal{R}$.). 
	\item Run Vertex Enumeration algorithm~\cite{avis2000revised} twice. Once for rays $P_{\cal R}(\sete)$ and again for the rays $P_{\cal R}(\sets)$. The input to the algorithm is given using the V-representation which has the format (list of vertices, list of rays) and in our case it simplifies to ($\mathbf{0}$, list of rays).  Note that it differs from the H-representation which is set of linear inequalities corresponding to the intersection of halfspaces. We used Matlab wrapper GeoCalcLib (http://worc4021.github.io/) and performed all the computations using Matlab R2020a. The two runs of the algorithm on inputs ($\mathbf{0}$, $P_{\cal R}(\sets)$) and  ($\mathbf{0}$, $P_{\cal R}(\sete)$) produce the following sets of extreme rays:  $\text{Ray}_S$, $\text{Ray}_E$ respectively.
	\item Form a new set $\text{Ray}_{final} = \{a\otimes b\,\, |\,\, \forall\, a\in \text{Ray}_S \text{ and } \forall\,  b\in \text{Ray}_E \}$.
	\item Run Vertex Enumeration algorithm on ($\mathbf{0}\otimes \mathbf{0}$, $\text{Ray}_{final}$) using the V-representation to obtain a set of extreme rays $W$. The set $W$ contains the list of potential non-classicality witnesses.
	\item To verify that the pair of sets $\sets_\cf$, $\sete_\cf$ does not admit a non contextual ontic model, it suffices to check whether there exists $w\in W$ such that
	\begin{equation}
	\Bigg\langle \sum_{j=1}^{\textup{dim}(\mathcal{R})} \mathcal{R}_j\otimes \mathcal{R}_j\,\, ,\,\, w \Bigg\rangle < 0,
	\end{equation}
	where $\big(\mathcal{R}_j\big)_{j=1}^{\textup{dim}(\mathcal{R})}$ is the sequence of orthonormal basis elements for $\mathcal{R}$ calculated in step 2, and $\langle\cdot,\cdot\rangle$ represent the Hilbert-Schmidt inner product.
\end{enumerate} 
Note that computation times for the last instance of vertex enumeration (with input ($\mathbf{0}\otimes\mathbf{0}$,$\text{Ray}_{final}$) runs for a long time, so we employed a slight optimization by preprocessing the sets $\text{Ray}_S$, $\text{Ray}_E$ by applying the GeoCalcLib vertex Reduction routine which removed non-extremal (redundant) rays of the polyhedrons in V-representation. This allowed us to run the Vertex Enumeration algorithm with input ($\mathbf{0}$,T), where  $T\subset \text{Ray}_{final}$.

We were able to successfully identify an element of $w\in W$ which provides a violation which is two orders of magnitude larger than any error due to rounding and approximations that stem
from using floating-point arithmetic.


\section{Engineered counterfactual clock}
\label{sec:telling time when all is off}

Here we will explain in detail how the counterfactual engineered clock works. The material is divided into four subsections to aid comprehension.

\subsection{Preliminaries}
Here we will describe the required dynamics of the on and off states of the clock, and show that such dynamics is indeed achievable. In particular, we will need to consider two orthonormal states $\ket{\psi_\off}$ and $\ket{\psi_\on}$, on an infinite dimensional Hilbert space, whose dynamics is generated via a Hamiltonian $\hat H$ of the form 
\begin{align}\label{eq:Hamiltonina}
\hat H= \hat H_\on - \hat H_\on \proj{\psi_\off} \hat H_\on/r_0,
\end{align}
where $\hat H_\on$ and $\ket{\psi_\off}$ are arbitrary so long as  $r_0:=\braket{\psi_\off | \hat H_\on | \psi_\off }\neq 0$. We require that the dynamics of $\ket{\psi_\on}$ under $\hat H_\on$ is orthogonal to $\ket{\psi_\off}$ at all times relevant to the experiment: 
\begin{align}\label{eq:orthogonality condition}
\braket{\psi_\off | \psi_\on (t)}=0, \quad \forall\,t\in[0,x_0 T_0]
\end{align}
where 
\begin{align}\label{eq:H on dynamics def}
\ket{\psi_\on(t)}:=\me^{-\mi t\hat H_\on}\ket{\psi_\on},
\end{align}
and $[0,x_0 T_0]$ is the time during which the clock will function and will be detailed later at the beginning of \cref{Protocol and model derivation}). We now demonstrate a simple proposition which shows how $\ket{\psi_\off}$ and $\ket{\psi_\on(t)}$ evolve under the total Hamiltonian $\hat H$: 
\begin{proposition}
For all $t\in[0,x_0 T_0]$,
\begin{align}
\me^{-\mi t\hat H}\ket{\psi_\off}&=\ket{\psi_\off} \label{eq:dynam psi on}\\
\me^{-\mi t\hat H}\ket{\psi_\on}&=\ket{\psi_\on(t)}, \label{eq:dynam psi off}
\end{align}
where $\ket{\psi_\on(t)}$ is given by \cref{eq:orthogonality condition}.
\end{proposition}
\begin{proof}
From \cref{eq:H on dynamics def} it follows
	\begin{align}
	0= \frac{d}{dt} \braket{\psi_\off  | \psi_\on(t) }=-\mi \braket{\psi_\off | \hat H_\on | \psi_\on (t)}.
	\end{align}
We thus conclude $\braket{\psi_\off | \hat H_\on | \psi_\on (t)}=0$ for all $t\in[0,x_0 T_0]$. Therefore, using \cref{eq:Hamiltonina} we find
	\begin{align}
	\hat H \ket{\psi_\on(t)}&= \hat H_\on\ket{\psi_\on(t)},\quad t\in[0,x_0 T_0]\label{eq:H on psi on}\\
	\hat H \ket{\psi_\off}&= \bf{0},\label{eq:H on psi off}
	\end{align}
	where $\bf{0}$ is the zero vector. We can use \cref{eq:H on psi on} to obtain \cref{eq:dynam psi on}:
	\begin{align}
	\me^{-\mi t \hat H}\ket{\psi_\on}&=\lim_{N\to\infty} \left(\id -\mi t \hat H/N\right)^N \ket{\psi_\on}= \lim_{N\to\infty} \left(\id -\mi t\hat H/N\right)^{N-1}\Big( \ket{\psi_\on(t/N)}+\bo(t/N)^2\Big)\\&= \lim_{N\to\infty}\left(\id -\mi t\hat H/N\right)^{N-2}\Big( \ket{\psi_\on(2t/N)}+ \bo(t/N)^2\Big)= \lim_{N\to\infty}\Big( \ket{\psi_\on(t)}+N \bo(t/N)^2\Big)\\
	&=\ket{\psi_\on(t)}.
	\end{align}
	Likewise, \cref{eq:dynam psi off} follows from \cref{eq:H on psi off}.
\end{proof}

\subsection{Protocol and model derivation}\label{Protocol and model derivation}
We 1st state the general dynamical properties of the clock at a qualitative level when turned on (i.e. operated in standard fashion). This will allow for a mental picture which will aid comprehension of the quantitative study which is to follow. When the 1st event occurs, the clock is turned to the on state, $\ket{\psi_\on(0)}$, and it starts ticking at elapsed times $t_1$, $2 t_1$, $3 t_1$, \ldots, $\NT t_1$. Then, at time $T_0=(\NT+1) t_1$, the dynamics of the clock repeats itself. The periodic behaviour is repeated $x_0\in\nnp$ times. Evidently, this clock can only tell the time modulo the period $T_0$. Importantly, the time at which the 2nd event occurs can be \emph{any} time in the interval $[0,x_0 T_0)$, with the answer being an estimate on the number of ticks which have occurred between the 1st and 2nd events.\Mspace

One starts with the clock in the off state $\ket{\psi_\off}$ and applies a unitary to turn it to $\ket{\psi_\on(t)}$ when the 1st event occurs. The clock will then evolve unitarily until the 2nd event occurs at some time $t$, at which point we measure the state $\ket{\psi_\on(t)}$ using an appropriately chosen measurement. The engineered clock's on state, $\ket{\psi_\on(t)}$, will not be orthogonal to itself after ticking: $\braket{\psi_{\on}(t)|\psi_\on(t')}\neq 0$ for $t\in\big[l t_1, (l\!+\! 1)t_1\big)$, $t'\in\big[l' t_1, (l'\!+\! 1)t_1\big)$, with $l\neq l'$, and $l,l'\in\{0,1,2,\ldots, \NT\}$. Since quantum measurements cannot perfectly distinguish non-orthogonal states, this implies that the clock will not be able to tell the time perfectly when used. However, the overlaps will decrease with increasing $|l-l'|$, and hence a well chosen measurement can still provide a good estimate of the number of ticks occurred. This point will not be of much importance when using the clock counterfactually, since in this modus operandi, the clock can only achieve the \ifreeoutcomes{} with a probability less than one anyway (analogously to the elementary clocks from \cref{sec app: clocks in nature}) and we will effectively be performing a form of unambitious quantum state discrimination.\Mspace

While the quantum system used as a counterfactual clock is much more complex in the engineered quantum clock case, the actual protocol is very similar to the one we have seen already in the elementary clock case, namely, starting with the clock in the off state, $\ket{\psi_\off}$, a unitary $U_0$ is applied when the 1st event occurs transforming the clock to $c\ket{\psi_\off}+s \ket{\psi_\on(0)}$, followed by applying a unitary $U_\measure$ when the 2nd event occurs and measuring in the measurement basis. The main physically important new feature is that the elapsed time between the two events can be \emph{any} time in the interval $[0, x_0 T_0]$. It is useful to introduce a set of orthonormal ancillary states, $\{\ket{\tilde A_l}\}_{l=1}^\NT$, which are orthogonal to the Hamiltonian, namely $\hat H \ket{ \tilde{A}_l}=0$. These form part of the basis in which we measure. In this section, the measurement basis, is the set of orthogonal projectors $\{$ $\proj{\Psi_\off}$, $\proj{\tilde{A}_1}$, $\proj{\tilde{A}_2}$, \ldots, $\proj{\tilde{A}_\NT}$, $\id - \proj{\Psi_\off} -\proj{\tilde{A}_1} -\proj{\tilde{A}_2}- \ldots -\proj{\tilde{A}_\NT}$ $\}$, where $\id$ denotes the identity operator.\footnote{\text{Strictly speaking, we are not projecting onto a basis for the entire Hilbert space, but only onto the relevant space for us.}} (In the case considered in main text, for which $\NT=1$, we denoted $\ket{\tilde{A}_1}$ by $\ket{A}$ for simplicity). Diagrammatically, just before the measurement, the protocol is as follows:
\begin{align}\label{eq:app:genreal proto type-2 clock}
\begin{tikzpicture}
\tikzstyle{level 1}=[level distance=1cm, sibling distance=1.5cm]
\tikzstyle{end} = [circle, minimum width=3pt,fill, inner sep=0pt]
\node (root) {$\!\!\!\ket{\psi_\off}$}[grow'=right]
child {
	node [end, label=right:{$c\ket{\psi_\off}$ $\;\quad$\myarrow{t} $c\ket{\psi_\off}\quad\,\overset{U_\measure\,}{\xrightarrow{\hspace*{0.7cm}}}\,\,\,\qquad\, \frac{A_1}{N} \ket{\psi_\off}\,\,\, +\,\,\, \frac{A_1}{N} \ket{\tilde{A}_1}\,\,\,+\,\,\,\ldots\,\,\,+\,\,\, \frac{A_1}{N} \ket{\tilde{A}_{\NT}} + A_2 \ket{\An_\off}$} ] {}}
child {	node [end, label=right:{$s\ket{\psi_\on(0)}$ \myarrow{t} $s\ket{\psi_\on(t)}\overset{U_\measure\,}{\xrightarrow{\hspace*{0.7cm}}}\,-\left( \frac{c}{s} \right)\! A_1\! \ket{\bar \psi_\off(t)} -\left( \frac{c}{s} \right)\! A_1\! \ket{\bar A_1(t)}-\ldots-\left( \frac{c}{s} \right)\! A_1 \!\ket{\bar A_{\NT}(t)} + A_3 \!\ket{\An_\on(t)},$} ] {}};
\end{tikzpicture}
\end{align}
where $N>0$, $A_1>0$, are to be determined and we have defined the amplitudes
\begin{align}
A_2=\sqrt{1-(\NT\! +\! 1)\left(\frac{A_1}{N}\right)^2},\quad 
A_3= \sqrt{1-(\NT\! +\! 1)\left(\frac{c}{s}\right)^2\! A_1^2},
\end{align}
and where all kets are normalised, and we will show that $\braket{\tilde{A}_r |\psi_\off}=\braket{A_\off |\psi_\off}=\braket{\bar A_r(t) |\psi_\off}=\braket{A_\on(t) |\psi_\off}=\braket{\tilde{A}_r |\tilde{A}_l}=\braket{A_\off |\tilde{A}_l}=\braket{\bar\psi_\off |\tilde{A}_l}=\braket{\bar A_r(t) |\tilde{A}_l}=\braket{A_\on(t) |\tilde{A}_l}=\braket{\bar\psi_\off(t) |A_\off}=\braket{\bar A_r(t) |A_\off}=\braket{\bar A_r(t) |\bar\psi_\off(t)}=\braket{A_\on(t) |\bar\psi_\off(t)} =\braket{\bar A_r(t) |\bar A_l(t)}=\braket{A_\on(t) |\bar A_l(t)}=0$  for all $t\geq 0$, and $l,r\in 0,1,2,\ldots, \NT-1$ such that $l\neq r$.

The remaining overlaps, namely $\braket{\bar \psi_\off(t) |\psi_\off}$, $\braket{\bar A_l(t) |A_l}$, $\braket{A_\on(t) |A_\off}$, are to be determined.

As we will soon see, we will associate the \ifreeoutcomes{} with measurement outcomes $\psi_\off$, $\tilde{A}_1$, $\tilde{A}_2$, \ldots, $\tilde{A}_\NT$, resultant from a measurement in the measurement basis performed immediately after $U_\measure$ is applied. To make the analysis a bit easier, we can consider the mathematically equivalent scenario in which, rather than applying the unitary $U_\measure$ to the state before measuring in the measurement basis, we can apply $U_\measure$ to the projectors onto the measurement basis instead, rotating the basis in which me measure to: $\proj{\tilde x_0}:=U_\measure^\dag \proj{\psi_\off}U_\measure$, $\proj{\tilde x_1}:=U_\measure^\dag \proj{\tilde{A}_1}U_\measure$, $\proj{\tilde x_2}:=U_\measure^\dag \proj{\tilde{A}_2}U_\measure$, \ldots, $\proj{\tilde x_\NT}:=U_\measure^\dag \proj{\tilde{A}_{\NT}}U_\measure$,  $\id-\sum_{l=0}^{\NT}\proj{\tilde x_l}$.\Mspace

We additionally chose $U_\measure$ to act trivially on $\ket{A_\off}$ and $\ket{A_\on(t)}$, for all $t\in[0, x_0 T_0)$. Writing the states $\ket{\psi_\off}$ and $\ket{\psi_\on(t)}$ in the new basis, we find
\begin{align}
\ket{\psi_\off}&=\frac{A_1}{N} \ket{\tilde x_0} + \frac{A_1}{N} \ket{\tilde x_1}+\ldots+ \frac{A_1}{N} \ket{\tilde x_\NT} + \sqrt{1-(\NT\! +\! 1)\left(\frac{A_1}{N}\right)^2} \ket{\An_\off}\label{eq:psi off def}\\
\ket{\psi_\on(t)}&=\me^{-\mi t \hat H_\on}\!\left( -\left( \frac{c}{s} \right)\! A_1 \ket{\bar x_0(0)} -\left( \frac{c}{s} \right)\! A_1 \ket{\bar x_1(0)}-\ldots-\left( \frac{c}{s} \right)\! A_1 \ket{\bar x_\NT(0)} + \sqrt{1-(\NT\! +\! 1)\left(\frac{c}{s}\right)^2\! A_1^2} \ket{\An_\on(0)}\!\right)\\
&=-\left( \frac{c}{s} \right)\! A_1 \ket{\bar x_0(t)} -\left( \frac{c}{s} \right)\! A_1 \ket{\bar x_1(t)}-\ldots-\left( \frac{c}{s} \right)\! A_1 \ket{\bar x_\NT(t)} + \sqrt{1-(\NT\! +\! 1)\left(\frac{c}{s}\right)^2\! A_1^2} \ket{\An_\on(t)}\label{eq:psi on def}
\end{align}
where $\ket{\bar x_0(t)}:= U_\measure^\dag \ket{\psi_\off(t)}$, $\ket{\bar x_1(t)}:= U_\measure^\dag \ket{\bar A_1(t)}$, $\ket{\bar x_2(t)}:= U_\measure^\dag \ket{\bar A_2(t)}$, $\ldots$,  $\ket{\bar x_{\NT}(t)}:= U_\measure^\dag \ket{\bar A_\NT(t)}$.

Hence $\braket{\bar \psi_\off(t) |\psi_\off}=\braket{\tilde x_{0}|\bar x_{0}(t)}$ and $\braket{\bar  A_l(t) |\tilde{A}_l}=\braket{\tilde x_{l}|\bar x_{l}(t)}$. We impose the constraint that the overlaps must satisfy
\begin{align}\label{eq:inner produc in term of G def}
\braket{\tilde x_{l}|\bar x_{l}(t)}= \frac{1}{N}\, G_\sigma\!\left( \frac{t-(\theta l +1) t_1}{\NT}-\frac{t_1}{2} \right),\quad l=0,1,\ldots, \NT
\end{align}
where $\theta=-1$ if $\NT=1$, and $\theta=1$ otherwise; and where
\begin{align}\label{eq: G sigma def}
G_\sigma(x):= \sum_{q=-(\theta+1)/2}^{x_0-1} G_{0,\sigma}\left( x- q \frac{T_0}{\NT} \right)
\end{align}
with $T_0=(\NT+1)t_1$ the clock period, $x_0\in\nnp$ the number of cycles through the period the clock performs 
and $G_{0,\sigma}$ is an approximation to the top hat function; namely
\begin{align}
G_{0,\sigma}(t):= \frac{1}{2}\left[ \erf\left(\frac{t/t_1+1/2}{\sqrt{2}\sigma}\right)- \erf\left(\frac{t/t_1-1/2}{\sqrt{2}\sigma}\right)  \right],
\end{align}
where $\erf$ is the \emph{error function} and $\sigma>0$ controls the quality of the approximation. Approximations to the top hat function are know as ``bells''; see \cite{Tophat} and references in the introduction for other possibilities.  The overlap $\braket{\tilde x_{l}|\bar x_{l}(t)}$ is plotted in \cref{fig:overlap} in the limit of small positive $\sigma$. By considering the conditions for a \ifree{} measurement, we see that this is precisely the necessary form of the function $G_\sigma$. Namely, it guarantees that for $t\notin \big[lt_1, (l+1)t_1 \big)$, the probability of obtaining outcome $\tilde x_l$ is zero, while for $t\in \big[lt_1, (l+1)t_1 \big)$, the probability of obtaining outcome $\tilde x_l$ is non zero but the overlap with the on branch is zero, namely $\braket{\tilde x_{l}|\psi_\on (t)}=\braket{\tilde x_{l}|\bar x_{l}(t)}=0$.
\begin{figure}
	\includegraphics[scale=0.8]{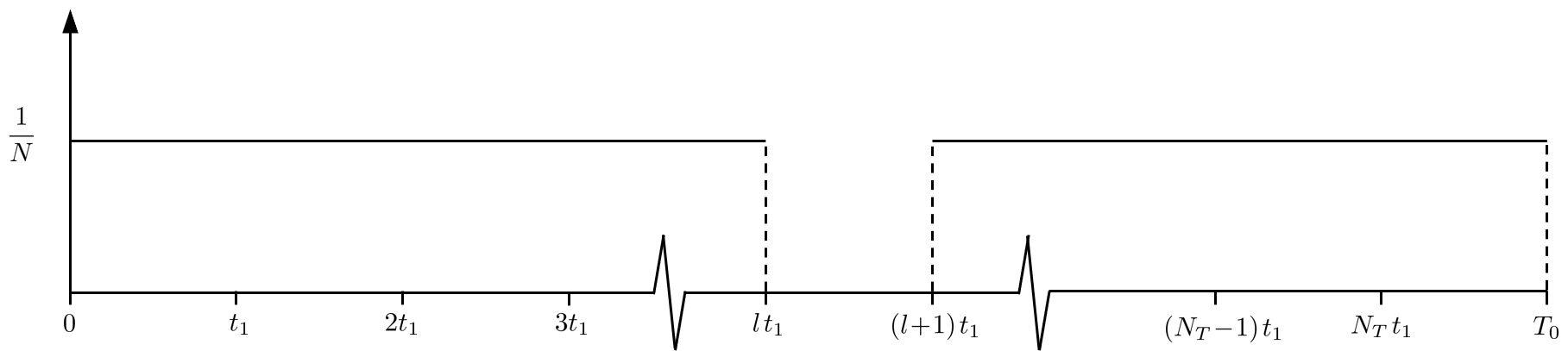}
	\centering
	\caption{Plot of $\braket{\tilde x_l|\bar x_l(t)}$ over one period: the function is zero in the time interval $[lt_1,(l+1)t_1)$ and $1/N$ at all other times.}\label{fig:overlap}
\end{figure}

Expanding the kets $\ket{\tilde x_{l}},\ket{\bar x_{l}}$ in the energy basis, $\ket{\tilde x_{l}}= \int \textup{d}E\,  \tilde x_{l}(E) \ket{E}$, $\ket{\bar x_{l}}=\int \textup{d}E\, \bar x_{l}(E) \ket{E}$, the overlap $\braket{\tilde x_{l}|\bar x_{l}(t)}$ becomes
\begin{align}
\braket{\tilde x_{l}|\bar x_{l}(t)}= \int \textup{d}E\, \tilde{x}^*_{l}(E) \bar x_{l}(E) \,\me^{-\mi 2 \pi E t},
\end{align}
which we identify as the Fourier Transform of $\tilde{x}^*_{l}(E) \bar x_{l}(E)$. Denoting the Fourier Transform and its inverse by $\mathcal{F}[f(t)](E):=\int \textup{d}E f(E) \me^{-\mi 2\pi E t}$ and $\mathcal{F}^{-1}[f(E)](t):=\int \textup{dt} f(t) \me^{\mi 2\pi E t}$ respectively, and using the  shift and rescaling properties of the Fourier Transform, we find 
\begin{align}
\tilde{x}^*_{l}(E) \bar x_{l}(E) &= \mathcal{F}^{-1} \big[\braket{\tilde x_{l}|\bar x_{l}(t)} \big](E)= \frac{1}{N}\,  \mathcal{F}^{-1} \left[ G_\sigma\!\left( \frac{t-(\theta l+1) t_1}{\NT}-\frac{t_1}{2} \right) \right] (E)\\ 
&= \frac{\NT}{N}\, \me^{\mi \pi E \big(2(\theta l+1) +\NT\big)t_1} \mathcal{F}^{-1} \big[ G_\sigma(t)\big] \big(E \NT\big),
\end{align}
Using \Cref{lem:fourier of quasi period fuction}, we find the inverse Fourier transform of $G_\sigma(t)$: 
\begin{align}
\mathcal{F}^{-1} \big[ G_\sigma(t)\big] \big(E\NT\big) &= \mathcal{F}^{-1} \big[ G_{0,\sigma}(t)\big] \big(E\NT\big)\, \frac{\sin\Big(\pi \big(x_0+(\theta+1)/2\big) E T_0 \Big)}{\sin\big(\pi E T_0\big)}\, \me^{-\mi \pi ET_0 (\theta+1)(x_0+2)/2}\\
&=\mathcal{F}^{-1} \big[ G_{0,\sigma}(t)\big] \big(E\NT\big)\, \frac{\sin\Big(\pi (\NT+1) \big(x_0+(\theta+1)/2\big) E t_1 \Big)}{\sin\big(\pi (\NT+1) E t_1\big)}\, \me^{-\mi \pi E t_1 (\NT+1) (\theta+1)(x_0+2)/2} ,
\end{align} 
where by direct calculation one finds
\begin{align}\label{eq:FT Gauss bell}
\mathcal{F}^{-1} \big[G_{0,\sigma}(t)\big] (x)= t_1 \me^{-2 \pi^2 \sigma^2 (x t_1)^2}\, \sinc(\pi x t_1),
\end{align}
where $\sinc(x)=\sin(x)/x$ is the \emph{sinc} function. We now make the trial solutions 
\begin{align}\label{x bar and tilde eqs}
\begin{split}
\tilde x^*_{l}(E)&= \sqrt{\frac{\NT}{N}\,  \mathcal{F}^{-1} \big[ G_\sigma(t)\big] \big(E \NT\big) \,\frac{\sin\Big(\pi \big( x_0+(\theta+1)/2 \big) E T_0 \Big)}{\sin\big(\pi E T_0\big)}\, } \, \me^{\mi l \lambda E t_1 }\\
\bar x_{l}(E)&= \sqrt{\frac{\NT}{N}\,  \mathcal{F}^{-1} \big[ G_\sigma(t)\big] \big(E \NT\big) \,\frac{\sin\Big(\pi \big( x_0+(\theta+1)/2 \big) E T_0 \Big)}{\sin\big(\pi E T_0\big)}\, } \, \me^{-\mi l \lambda E t_1 }  \me^{\mi \pi E \big(2(\theta l+1) +\NT\big)t_1} \, \me^{-\mi \pi E T_0 (\theta+1)(x_0+2)/2}
\end{split}
\end{align}
 for $l=0,1,2,\ldots,\NT$; and where $\sqrt{\cdot}$ denotes the principle square root and $\lambda\in\rr$ is to be determined. We can now calculate the normalisation constant $N$: 
\begin{align}
1= \braket{\tilde x_{l} | \tilde x_{l}}= \int \textup{d}E \,\tilde x^*_{l}(E) \tilde x_{l}(E)= \frac{\NT}{N} \int \textup{d}E \,\left| \sqrt{\mathcal{F}^{-1} \big[ G_\sigma(t)\big] \big(E \NT\big) \,\frac{\sin\Big(\pi \big( x_0+(\theta+1)/2 \big) E T_0 \Big)}{\sin\big(\pi E T_0\big)}\,} \,\right|^2,
\end{align}
hence using the identity $|\sqrt{z}|^2=|z|$ for all $z\in\cc$, 
\begin{align}
N &= \NT \int \textup{d}E\, \left|\mathcal{F}^{-1} \big[ G_\sigma(t)\big] \big(E \NT\big) \,\frac{\sin\Big(\pi \big( x_0+(\theta+1)/2 \big) E T_0 \Big)}{\sin\big(\pi E T_0\big)}  \right|\\
& = \NT \int \textup{d}E\, t_1 \me^{-2 \pi^2 \sigma^2 (E\NT t_1)^2}\,  \left|\sinc(\pi E\NT t_1) \,\frac{\sin\Big(\pi \big( x_0+(\theta+1)/2 \big) E T_0 \Big)}{\sin\big(\pi E T_0\big)}\, \right|\\
& = \NT \int \textup{d}y\, \me^{-2 (\pi \NT \sigma)^2 y^2}\,  \left|\sinc(\pi\NT y) \,\frac{\sin\Big(\pi (\NT+1)\big( x_0+(\theta+1)/2 \big) y \Big)}{\sin\big(\pi (\NT+1)y \big)}\, \right|.\label{eq:N as a function of sigma} 
\end{align}
Observe that this choice for $N$ also implies $\braket{\bar x_{l} (t)| \bar x_{l}(t)}=1$ for all $t\in\rr$. 
We now examine the implications of the requirement that the states $\ket{\psi_\off}$ and $\ket{\psi_\on(t)}$ have to be orthogonal at all times:
\begin{align}
0= \braket{\psi_\off|\psi_\on(t)}= - \left(\frac{c}{s}\right) \frac{A_1^2}{N} \left(\sum_{l=0}^\NT \braket{\tilde x_{l}|\bar x_{l}(t)} \right) + \sqrt{\left(1-(\NT\! +\! 1) \left(\frac{A_1}{N}\right)^2\right) \left(1-(\NT\! +\! 1) \left(\frac{c}{s}\right)^2 A_1^2 \right)} \braket{A_\off|A_\on(t)}.
\end{align} 
Therefore, $A_1$ must satisfy the equation
\begin{align}\label{eq:a1 must satisfy}
c_0= \frac{\left(\frac{c}{s}\right) \frac{A_1^2}{N}}{\sqrt{\left(1-(\NT\! +\! 1) \left(\frac{A_1}{N}\right)^2\right) \left(1-(\NT\! +\! 1) \left(\frac{c}{s}\right)^2 A_1^2 \right)}}\in\rr,
\end{align}
where $c_0$ satisfies
\begin{align}\label{eq:c0 def}
\braket{A_\off|A_\on(t)}= c_0  \sum_{l=0}^\NT \braket{\tilde x_{l}|\bar x_{l}(t)} .
\end{align}
We will now find a value of $c_0$ which satisfies \cref{eq:c0 def} and come back to \cref{eq:a1 must satisfy} later. Expanding $\ket{A_\off}$ and $\ket{A_\on(t)}$ in the energy basis, $\ket{A_\off}=\int\textup{d}E A_\off(E) \ket{E}$, $\ket{A_\on(0)}=\int\textup{d}E A_\on(E) \ket{E}$, we have that $\braket{A_\off|A_\on(t)}=\int\textup{d}E A_\off^*(E)A_\on(E) \me^{-\mi 2\pi E t} $ and hence
\begin{align}\label{eq:A off on constraint}
A_\off^*(E)A_\on(E)= c_0\, \mathcal{F}^{-1} \left[ \sum_{l=0}^\NT \braket{\tilde x_{l}|\bar x_{l}(t)} \right](E).
\end{align}
We calculate the summation before proceeding with taking the inverse Fourier transform. Taking into account \cref{eq:inner produc in term of G def,eq: G sigma def} one finds
\begin{align}
 \sum_{l=0}^\NT \braket{\tilde x_{l}|\bar x_{l}(t)} &= \sum_{l=0}^\NT \frac{1}{N}\, G_\sigma\!\left( \frac{t-(\theta l+1) t_1}{\NT}-\frac{t_1}{2} \right)= \frac{1}{N}\sum_{l=0}^\NT \sum_{q=-(\theta+1)/2}^{x_0-1} G_{0,\sigma}\!\left( \frac{t-(\theta l+1) t_1}{\NT}-\frac{t_1}{2}-q \frac{T_0}{\NT} \right)\\
 &  = \frac{1}{N}\sum_{l=0}^\NT \sum_{q=-(\theta+1)/2}^{x_0-1} G_{0,\sigma}\!\left( \frac{t}{\NT}-\frac{t_1}{2}- \big((\theta l+1)+ q(\NT+1) \big)\frac{ t_1}{\NT} \right)\label{eq:line 2nd} \\
 &= \frac{1}{N} \sum_{q=-\NT -(\theta-1)/2}^{\NT +(x_0-1)(\NT+1)-(\theta+1)/2} G_{0,\sigma}\!\left( \frac{t}{\NT}-\frac{t_1}{2}- q\frac{ t_1}{\NT} \right)=\frac{1}{N}  G_{\sigma}^{(2)}\!\left( \frac{t}{\NT}-\frac{t_1}{2} \right),\label{eq:line 3nd}
\end{align}
where in going from line \ref{eq:line 2nd} to line \ref{eq:line 3nd}, we have use the identity $\sum_{n=0}^{d-1} \sum_{m=a}^b f(n+md)=\sum_{n=a d}^{d-1+b d} f(n)$, which holds for arbitrary function $f$ and $d\in\nnp$, $a\leq b,$\, $a,b\in\zz$ and we have fined
\begin{align}
G_{\sigma}^{(2)}(x) :=\sum_{q=-\NT -(\theta-1)/2}^{\NT +(x_0-1)(\NT+1)-(\theta+1)/2} G_{0,\sigma}\!\left(x- q\frac{ t_1}{\NT} \right).
\end{align}
We now return to the task of taking the inverse Fourier transform in \cref{eq:A off on constraint}:
\begin{align}
A_\off^*(E)A_\on(E)&= c_0\, \mathcal{F}^{-1} \left[ \sum_{l=0}^{\NT} \braket{\tilde x_l|\bar x_l(t)} \right](E) = \frac{c_0}{N}  \mathcal{F}^{-1} \left[  G_{\sigma}^{(2)}\!\left( \frac{t}{\NT}-\frac{t_1}{2} \right)\right](E)\\
&= c_0 \frac{\NT}{N} \me^{\mi \pi E t_1 \NT}  \mathcal{F}^{-1} \left[  G_{\sigma}^{(2)}(t) \right](E \NT),
\end{align}
where by direct calculation using \cref{lem:fourier of quasi period fuction}, one has
\begin{align}
\mathcal{F}^{-1} \left[  G_{\sigma}^{(2)}(t) \right](E\NT)&= \mathcal{F}^{-1} \big[ G_{0,\sigma}(t)\big] (E\NT)\, \frac{\sin\big(\pi (x_0+1) E T_0\big)}{\sin(\pi E t_1)}\, \me^{-\mi \pi Et_1 \big( (\NT+1)(x_0+1)-1 \big)\big(\NT+(\theta-1)/2\big) },\\
&= \mathcal{F}^{-1} \big[ G_{0,\sigma}(t)\big] (E\NT)\, \frac{\sin\big(\pi (\NT+1)(x_0+1) E t_1\big)}{\sin(\pi E t_1)}\, \me^{-\mi \pi E t_1\big( (\NT+1)(x_0+1)-1 \big)\big(\NT+(\theta-1)/2\big) },
\end{align}
where recall that $\mathcal{F}^{-1} \big[ G_{0,\sigma}(t)\big] (E\NT) $ is given by \cref{eq:FT Gauss bell}. We thus make the following trial solutions for $A_\off^*(E)$ and $A_\on(E)$
\begin{align}\label{eq:A on off def}
\begin{split}
A_\off^*(E)&= \sqrt{c_0\frac{\NT}{N}\, \mathcal{F}^{-1} \big[ G_{0,\sigma}(t)\big] (E\NT)\, \frac{\sin\big(\pi (\NT+1)(x_0+1) E t_1\big)}{\sin(\pi E t_1)} \,} \, \me^{\mi \lambda E t_1/2}\\
A_\on(E)&= \sqrt{c_0 \frac{\NT}{N}\, \mathcal{F}^{-1} \big[ G_{0,\sigma}(t)\big] (E\NT)\, \frac{\sin\big(\pi (\NT+1)(x_0+1) E t_1\big)}{\sin(\pi E t_1)}\, } \, \me^{-\mi \lambda E t_1/2}\, \me^{-\mi \pi E t_1\big( (\NT+1)(x_0+1)-1 \big) \big( \NT+(\theta-1)/2 \big) },
\end{split}
\end{align}
where, as before, we use the principle square root. Normalisation of $\ket{A_\off}$ and $\ket{A_\on(t)}$ imply
\begin{align}
1=  \frac{\NT}{N}\int \textup{d}E\, \left|\sqrt{c_0\, \mathcal{F}^{-1} \big[ G_{0,\sigma}(t)\big] (E\NT)\, \frac{\sin\big(\pi (\NT+1)(x_0+1) E t_1\big)}{\sin(\pi E t_1)}\,}\, \right|^2,
\end{align}
from which, using the identity $|\sqrt{z}\,|^2=|z|$ for all $z\in\cc$, we determine the value of  $1/|c_0|$ to be
\begin{align}
1/|c_0| &= \frac{\NT}{N}\int \textup{d}E \left|\mathcal{F}^{-1} \big[ G_{0,\sigma}(t)\big] (E\NT)\, \frac{\sin\big(\pi (\NT+1)(x_0+1) E t_1\big)}{\sin(\pi E t_1)}\, \right|\\
&=  \frac{\NT}{N} \int \textup{d}y\, \me^{-2 (\pi\NT \sigma)^2 y^2}\, \left|\, \sinc(\pi \NT y) \, \frac{\sin\big(\pi (\NT+1)(x_0+1) y\big)}{\sin(\pi y)}\, \right|,\label{eq:c_0 xpression}
\end{align}
and where the sign of $c_0\in\rr$ is given by the sign of $c/s$; this follows from \cref{eq:a1 must satisfy}.
Now that we have a value for $c_0$, we can determine $A_1^2$ using \cref{eq:a1 must satisfy}, which is a quadratic equation in $A_1^2$. We thus find
\begin{align}\label{}
\frac{A_1^2}{N^2}= \frac{(\NT+1)\Big( \left(\frac{c}{s}\right)^2 N^2+1\Big) +\Gamma_\pm \sqrt{(\NT+1)^2 \Big( \left(\frac{c}{s}\right)^2 N^2+1\Big)^2 -4 \left(\frac{c}{s}\right)^2 \Big( (\NT+1)^2 - 1/c_0^2\Big)} }{2 \left(\frac{c}{s}\right)^2  N^2\Big( (\NT+1)^2 - 1/c_0^2\Big)},
\end{align}
where $\Gamma_\pm$ is -1 or 1 depending on which solution to the equation we choose. Since we have assumed the amplitudes associated with the kets $\ket{A_\off}$ and $\ket{A_\on}$ to be both real in \cref{eq:psi off def,eq:psi on def} respectively, we need to take the $\Gamma_\pm=-1$ solution. 
Since for small $\sigma$, the probability of measuring a tick counterfactually is 
$P_{cf}= c^2 A_1^2/N^2$ for any one of the $\NT$ ticks, we have that
\begin{align}
P_{cf}= s^2\frac{(\NT+1)\Big( \left(\frac{c}{s}\right)^2 N^2+1\Big) - \sqrt{(\NT+1)^2 \Big( \left(\frac{c}{s}\right)^2 N^2+1\Big)^2 -4 \left(\frac{c}{s}\right)^2 \Big( (\NT+1)^2 - 1/c_0^2\Big)} }{2 N^2\Big( (\NT+1)^2 - 1/c_0^2\Big)},
\end{align}
where recall that $c=\cos(\theta), s=\sin(\theta)$, and $N=N(\sigma)$, $c_0=c_0(\sigma)$ are given by \cref{eq:N as a function of sigma,eq:c_0 xpression} respectively. The numerical values of $P_{cf}$ for a given $\sigma$ presented in the main text and are calculated by setting $\NT=x_0=1$ and evaluating $N$ and $c_0$ numerically for this $\sigma$, followed by maximising $P_{cf}$ numerically over $\theta\in[0,2 \pi]$.\Mspace

Finally, in order for the trial solutions for the wave functions ${\tilde x_l}(E)$, ${\bar x_l}(E)$, ${A_\off}(E)$ and ${A_\on}(E)$ to be valid, we must verify their orthogonality relations; recall that we have so far assumed 
$\braket{\tilde x_l|\tilde x_r}=\braket{\tilde x_l|\bar x_r(t)}= \braket{\tilde x_l|A_\on(t)}= \braket{\bar x_l(t)|A_\on(t)}= \braket{\tilde x_l|A_\off}= \braket{\bar x_l(t)|A_\off}=0$ for all $t\geq 0$, and $l,r\in 0,1,2,\ldots, \NT$ such that $l\neq r$.  This is where the phase factors $\lambda$ come in to play.
From \cref{x bar and tilde eqs,eq:A on off def}, we see that these overlaps are proportional in absolute value, to integrals of the two following different forms:
\begin{itemize}
	\item [Case 1:]
\begin{align}\label{eq:F 1 lambda}
F_1(\lambda):= \left|\, \int \textup{d}x\, \me^{-2 (\pi A \sigma)^2 x^2}\,  \sinc(\pi B_1 x) \, \frac{\sin\big(\pi C_1 x\big)}{\sin(\pi D_1 x)}\,\me^{\mi H_1 x}  \me^{\mi H_2  \lambda x}\,  \right|,
\end{align}
where $A, B_1, C_1, D_1 , C_1/D_1 \in\nnp$, $H_1\in\zz$, while $H_2$ is either a non zero integer or half integer. \\
\item [Case 2:]
\begin{align}\label{eq:F 2 lambda def}
F_2(\lambda):= \left|\, \int \textup{d}x\, \me^{-2 (\pi A \sigma)^2 x^2}\, \sqrt{ \sinc(\pi B_1 x) \, \frac{\sin\big(\pi C_1 x\big)}{\sin(\pi D_1 x)}} {\left(\!\sqrt{ \sinc(\pi B_2 x) \, \frac{\sin\big(\pi C_2 x\big)}{\sin(\pi D_2 x)}}\,\right)\!\!}^* \,\me^{\mi H_1 x}  \me^{\mi H_2  \lambda x} \,\right|,
\end{align}
where $C_1,D_1,H_1,H_2$ satisfy the same constraints as in Case 1, while $B_2, C_2, D_2 , C_2/D_2 \in\nnp$. 
\end{itemize}

We will now verify that we can make $F_1(\lambda), F_2(\lambda)$ arbitrarily small by choosing $\lambda>0$ large enough.

For case 1, since the integrand of $F_1(\lambda)$ is smooth and converges absolutely, it follows from the principle of stationary phase, that $F_1(\lambda)$ tends to zero as $\lambda$ tends to infinity. 

For case 2, the integrand is not absolutely continuous, due to the changes in sign of the functions under the square roots. As such, the principle of stationary phase does not directly apply. However, a variant of the usual stationary phase type argument still applies. For this we will make use of \cref{lem:upper bound on integral}. We start by choosing a consistent expression for $c,f,G,$ and $g$:

\begin{align}
c(x)&= \sqrt{\sign\left( \sinc(\pi B_1 x) \frac{\sin\big(\pi C_1 x\big)}{\sin(\pi D_1 x)} \right)} \left(\sqrt{\sign\left(  \sinc(\pi B_2 x) \, \frac{\sin\big(\pi C_2 x\big)}{\sin(\pi D_2 x)} \right)}\,\,\right)^{\!\!\!*},\\
f(x)&=  \sqrt{ \left|\, \sinc(\pi B_1 x) \frac{\sin\big(\pi C_1 x\big)}{\sin(\pi D_1 x)} \right|} {\left(\!\sqrt{\left|  \sinc(\pi B_2 x) \, \frac{\sin\big(\pi C_2 x\big)}{\sin(\pi D_2 x)} \right|}\,\right)\!\!}\\
&= \sqrt{\left| \,\sin(\pi B_1 x)\sin(\pi B_2 x) \frac{\sin\big(\pi C_1 x\big)  \sin\big(\pi C_2 x\big)}{\sin(\pi D_1 x)  \sin(\pi D_2 x)} \right|}\\
G(x)&=\me^{-2 (\pi A \sigma)^2 x^2},\\
g(x)&= \me^{\mi H_1 x}  \me^{\mi H_2  \lambda x},
\end{align}
where $\sign(\cdot)$ is the \emph{sign} function. \Cref{lem:upper bound on integral} allows us to upper bound $F_2(\lambda)$. Since $\sum_{n\in\zz} G(nT)= \sum_{n\in\zz} \me^{-2 (\pi A \sigma)^2 T^2 n^2} < \infty$, the right hand side of \cref{eq:I up bound} is finite in this case. Moreover, the only $\lambda$ dependency enters in the $v_0$ term. Since $v(x)=\int \textup{d}x\, \me^{\mi H_1 x}  \me^{\mi H_2  \lambda x}= \me^{\mi H_1 x+\mi H_2  \lambda x}/ (\mi H_1 +\mi H_2  \lambda)$, we can choose $v_0= 1 / | H_1 + H_2  \lambda |$. Therefore, $F_2(\lambda)$ tends to zero as $\lambda$ tends to infinity.

\subsection{Precision quantification}\label{Precision quantification}
We now show how to calculate the fidelity between the prescribed protocol in \cref{sec:telling time when all is off} and that of a hypothetical ``idealised'' version of the clock which satisfies both of the following:
\begin{itemize}
\item [1)] The probability that the clock had been on when one of the outcomes associated with the projectors $\Big\{\proj{\tilde x_l}\Big\}_{l=0}^{N_T}$ is obtained, is exactly zero.
\item [2)] The probability that the clock would have ticked $n$ times, had it been on, when the outcome associated with projector $\proj{\tilde x_n}$ is obtained, is exactly one.  
\end{itemize} 

The only difference between the idealised version of the clock and that of \cref{sec:telling time when all is off}, will be the dynamics of the $\ket{\psi_{\on}(t)}$ state. Specifically, we replace \cref{eq:psi on def}, namely, 
\begin{align}
\ket{\psi_\on(t)}&=-\left( \frac{c}{s} \right)\! A_1 \ket{\bar x_0(t)} -\left( \frac{c}{s} \right)\! A_1 \ket{\bar x_1(t)}-\ldots-\left( \frac{c}{s} \right)\! A_1 \ket{\bar x_\NT(t)} + \sqrt{1-(\NT\! +\! 1)\left(\frac{c}{s}\right)^2\! A_1^2} \ket{\An_\on(t)},
\end{align}
with 
\begin{align}
\ket{\psi_\on I (t)}&:=-\left( \frac{c}{s} \right)\! A_1 \ket{\bar x_0 I(t)} -\left( \frac{c}{s} \right)\! A_1 \ket{\bar x_1 I(t)}-\ldots-\left( \frac{c}{s} \right)\! A_1 \ket{\bar x_\NT I(t)} + \sqrt{1-(\NT\! +\! 1)\left(\frac{c}{s}\right)^2\! A_1^2} \ket{\An_\on I(t)},
\end{align}
where \emph{I} stands for idealised, and the kets $\ket{\bar x_m I(t)}$ satisfy $\braket{\tilde x_l | \bar x_m I (t)}=\frac{1}{N} \delta_{l,m} \delta_l(t)$, where $\delta_{l,m}$ is the Kronecker delta, and where $\delta_l(t)=0$ if $t\in\cup_{q=0}^{x_0-1}\big[l\, t_1+qT_0,\, (l+1) t_1+qT_0\big)$ and $\delta_l(t)=1$ otherwise. The ket $\ket{\An_\on I(t)}$ is orthogonal to the kets $\big\{\ket{\tilde x_l}\big\}_{l=0}^{N_T}$ and obeys $\braket{\An_\off| \An_\on I (t)}= c_0 \sum_{l=0}^{\NT} \braket{\tilde x_l | \bar x_l I (t)}$ with $c_0$ given by \cref{eq:a1 must satisfy}, in analogy with \cref{eq:c0 def} satisfied by $\ket{\An_\on (t)}$.

Note that while if one performs the protocol associated with the counterfactual clock laid out in \cref{Protocol and model derivation}, one finds that the clock works ``perfectly'', that is to say, satisfies 1) and 2) above, it is unclear whether such dynamics are achievable with a time independent Hamiltonian. This is the justification for calling it idealised. Its purpose is to show that our protocol which is realised via a time independent Hamiltonian, approximates the idealised clock in fidelity up to an arbitrary precision, by choosing $\delta>0$ sufficiently small. In the special case in which the clock ticks just once ($\NT=1$), this description of an idealised clock is the same as \cref{eq:x 0 overlap,eq:x 1 overlap} in the main text. In particular, we are only concerned with the differences in fidelity when obtaining one of the outcomes $\big\{\ket{\tilde x_l}\big\}_{l=0}^{N_T}$, that is to say, the outcomes which are relevant for the counterfactual operation of the clock. We start by evaluating the difference in fidelities at time $t$ between what is actually obtained with our protocol, and what would have been obtained in the idealised case:
\begin{align}
	\textup{Dif}_p(\sigma,t)&:= \big|  \bra{\tilde x_p} \big(c \ket{\psi_\off}+s \ket{\psi_\on(t)}\big) \big|^2 -\big|  \bra{\tilde x_p} \big(c \ket{\psi_\off}+s \ket{\psi_\on I(t)}\big) \big|^2\\
	&=\frac{A_1^2 c^2}{N^2} \left( \big| 1- N\braket{\tilde x_p | \bar x_p(t)} \big|^2 - \big| 1- \delta_p(t) \big|^2 \right),
\end{align}
for $p=0,1,\ldots, \NT$. Furthermore, since the clock could be measured at any time, it is instructive to consider the time averaged error rate over the total operation time of the clock, namely the interval $[0,x_0 T_0)$. We can furthermore subdivide this interval into two disjoint parts. The first is the interval in which $p$ ticks should have occurred (and thus the overlaps should be close to one), namely $\cup_{q=0}^{x_0-1} [ p t_1+q T_0, (p+1) t_1 + q T_0)$. The second consists in the remaining intervals, namely $\cup_{q=0}^{x_0-1} \big( [q T_0, p t_1 + q T_0) \cup [(p+1) t_1+ q T_0, (q+1) T_0]  \big) $, for which the probability of the clock ticking $p$ times, had it been on, should be very small.  
For the 1st type of error, we find
\begin{align}
\textup{Dif}^{(1)}_p(\sigma) &:= \frac{1}{x_0 t_1}\sum_{q=0}^{x_0-1} \int_{p t_1 + q T_0}^{(p+1) t_1 + q T_0} \textup{Dif}^{(1)}_p(\sigma,t) \\
&= \frac{A_1^2 c^2}{N^2 x_0 t_1}  \sum_{q=0}^{x_0-1} \int_{p t_1 + q T_0}^{(p+1) t_1 + q T_0} dt \left( \bigg| 1- G_\sigma\left( \frac{t-(\theta p+1)t_1}{\NT}-\frac{t_1}{2} \right) \bigg|^2-1\right)\\
&= \frac{A_1^2 c^2}{N^2} \left(-1+  \sum_{q=0}^{x_0-1} \int_{0}^{1} dx\,  \bigg| 1- G_\sigma\left( \left[\frac{x-1+(1-\theta)p+q(\NT+1)}{\NT}-\frac{1}{2}\right]t_1\! \right) \bigg|^2\right),
\end{align}
which, recalling the definition of $G_\sigma$  (\cref{eq:inner produc in term of G def,eq: G sigma def}), is observed to be $t_1$ independent.
For the case studied in the main text, $x_0=1$, $\theta=-1$, $\NT=1$, $p=0,1$, and we find
\begin{align}
\textup{Dif}^{(1)}_p(\sigma) &= \frac{A_1^2 c^2}{N^2} \left(-1+   \int_{0}^{1} dx\,  \bigg| 1- G_\sigma\left( \left[x-1+2p-\frac{1}{2}\right]t_1\! \right) \bigg|^2\right)\\
&= \frac{A_1^2 c^2}{N^2} \left(-1+   \int_{0}^{1} dx\,  \bigg| 1- G_\sigma\left( \left[x-1-\frac{1}{2}\right]t_1\! \right) \bigg|^2\right),
\end{align}
which holds for both $p=0,1$ since $G_\sigma$ is a symmetric function.
Similarly, we have 
\begin{align}
\textup{Dif}^{(2)}_p(\sigma) :=& \frac{1}{x_0(T_0- t_1)}\sum_{q=0}^{x_0-1} \int_{q T_0}^{pt_1 + q T_0} \textup{Dif}^{(2)}_p(\sigma,t)+ \int_{(p+1)t_1+q T_0}^{(q+1) T_0} \textup{Dif}^{(2)}_p(\sigma,t)\\
=&  \frac{1}{x_0 \NT}\sum_{q=0}^{x_0-1} \Bigg[\int_{0}^{p} dx \left| 1 -G_\sigma \left( \left[ \frac{x-(\theta p+1)+q(\NT+1)}{\NT}-\frac{1}{2} \right]t_1 \!\right)  \right|^2\\
&  + \int_{0}^{\NT-p} dx \left| 1 -G_\sigma \left( \left[ \frac{x+(1-\theta) p+q(\NT+1)}{\NT}-\frac{1}{2} \right] t_1 \!\right) \right|^2  \Bigg],
\end{align}
which, similarly to $\textup{Dif}^{(1)}_p(\sigma)$, is observed to be $t_1$ independent. For the special case of one tick described in the main text ( $x_0=1$, $\theta=-1$, $\NT=1$, $p=0,1$), the expression reduces to
\begin{align}
\textup{Dif}^{(2)}_p(\sigma)&=  \Bigg[\int_{0}^{p} dx \left| 1 -G_\sigma \left( \left[ x+ p-1-\frac{1}{2} \right]t_1 \!\right)  \right|^2  + \int_{0}^{1-p} dx \left| 1 -G_\sigma \left( \left[ x+2 p-\frac{1}{2} \right] t_1 \!\right) \right|^2  \Bigg]\\
&= \int_{0}^{1} dx \left| 1 -G_\sigma \left( \left[ x-\frac{1}{2} \right]t_1 \!\right)  \right|^2,
\end{align}
for both $p=0$ and $p=1$.

\subsection{Technical lemmas}

For the following lemma, recall the definition of the inverse Fourier transform: $\mathcal{F}^{-1}[f(x)](y):=\int \textup{d}x f(x) \me^{\mi 2\pi x y}$.   
\begin{lemma}\label{lem:fourier of quasi period fuction}
Consider a function $f(x)=\sum_{m=a}^{b} f_0\left(x-mT\right)$, generated by some Schwartz space function $f_0:\rr\to\rr$, with parameters $a\leq b$,\, $a,b\in\zz$, $T\in\rr$. Its inverse Fourier transform takes on the form
\begin{align}
\mathcal{F}^{-1}\big[f(x)\big](y)= \mathcal{F}^{-1}\big[f_0(x)\big] (y) \,\frac{\sin\big((b-a+1) \pi y T \big)}{\sin\big(\pi y T\big)}\, \me^{\mi \pi y a T (2+b-a)} 
\end{align}
for all $y\in\rr$ and where the singular points are assigned by continuity. 
\end{lemma}	

\begin{proof}
First observe that $f$ can be written in terms of convolution of $f_0$ with a Dirac comb:
	\begin{align}
		f(x)= f_0(x)\, \hat \otimes \sum_{q=a}^{b} \delta(x - q T),
	\end{align}
	where $\hat \otimes$ denotes convolution. Now apply the inverse Fourier transform to both sides of the equation and invoke the convolution theorem. This gives
	\begin{align}\label{eq:check sign}
		\mathcal{F}^{-1}\big[f(x)\big](y)= 	\mathcal{F}^{-1}\big[f_0(x)\big](y)\, \mathcal{F}^{-1}\left[\sum_{q=a}^{b} \delta(x - q T)\right](y)=	\mathcal{F}^{-1}\big[f_0(x)\big](y) \sum_{q=a}^{b} \me^{ \mi 2\pi q yT}.
	\end{align}	
Finally, take real and imaginary parts of $\sum_{q=a}^{b} \me^{ \mi 2\pi q yT}$ followed by applying the sums of cosines and sines arithmetic progressions formulas from~\cite{GeometriSeries}.	
\end{proof}

\begin{lemma}\label{lem:upper bound on integral}
	Consider a convergent integral of the form $I:=\int_{-\infty}^\infty\textup{d}x\, c(x)f(x)G(x)g(x),$ where $c:\rr\to\cc$, $f:\rr\to\rr$, $G:\rr\to\rro,$ $g:\rr\to\cc$ and $G,g\in {\bf C}^{1}$ with $\frac{d}{dx}G(x)\geq 0$ for $x\in(-\infty,0]$ and $\frac{d}{dx}G(x)\leq 0$ for $x\in[0,\infty)$. 
	Furthermore, let $f$ be periodic with a countable number of zeros greater or equal to one in $[0,T]$, where $T$ is its period; and let $f\in {\bf C}^1$ on the intervals $\rr\backslash \{\ldots, a_{-2},a_{-1},a_1,a_2,\ldots\}$ where $\ldots, a_{-2},a_{-1},a_1,a_2,\ldots$ is the sequence of zeros of $f$ in ascending order.	
	In addition, let there exist a sequence of constants $(c_l)_l$ such that $c(x)=c_l$ for all $x\in \big(a_l,a_{l+1}\big)$, $l\in\zz$, where $a_0=0$. Let both $c(x)$ and  $v(x):=\int \textup{d}x\, g(x)$ be uniformly bounded:  $|c(x)|=1$ and $|v(x)|\leq v_0$ for some $v_0\geq 0$, for all $x\in\rr$. The following bound holds:
	\begin{align}\label{eq:I up bound}
	|I| \leq 2 N_\textup{Z} \big(1+2 N_\textup{TP}\big) \left(\max_{x\in[0,T]} |f(x)| \right)  \left( 2 G(0)+ \sum_{n\in\zz}\, G(nT)\right) v_0,
	\end{align}
	where $N_\textup{TP}\in\nnp$ is the number of tuning points of $f$ in $[0,T]\,\backslash \,(\ldots, a_{-2},a_{-1},a_1,a_2,\ldots)$, while $N_\textup{Z}\in\nnp$ is number of zeros of $f$ in interval $[0,T]$.
\end{lemma}
\begin{proof}
	Since $f$ is periodic and has a countable number of zeros, without loss of generality, let $a_l<0$ for $l<0$ and $a_l\geq 0$ for $l> 0$. Define 
	$c_0:=c_1$ unless $0$ is a zero of $f$, in which case $c_0$ is already defined. We start by observing
	\begin{align}
	I=\sum_{l\in\zz} c_l \int_{a_l}^{a_{l+1}} \textup{d}x\, f(x) G(x)g(x).
	\end{align}
	Since $f,G$ and $g$ are smooth on the intervals $(a_l,a_{l+1})$, we can now integrate by parts:
	\begin{align}
	I=\sum_{l\in\zz}c_l \Big[ f(x) G(x) v(x)\Big]_{a_l}^{a_{l+1}}+  c_l \int_{a_l}^{a_{l+1}} \textup{d}x\, \left(\frac{d}{dx} f(x)G(x)\right) v(x).
	\end{align}
	Taking absolute values, employing the triangle inequality and rearranging the summation, we arrive at
	\begin{align}\label{eq:first part of I bound 0}
	|I| &\leq \sum_{l\in\zz}  \int_{a_l}^{a_{l+1}}  \textup{d}x \left|\frac{d}{dx} f(x) G(x)\right| \Big|v(x)\Big| +2 \sum_{l\in\zz}  \Big|f(a_l)G(a_l)v(a_l)\Big|\\
	& \leq v_0  \sum_{l\in\zz}  \int_{a_l}^{a_{l+1}}  \textup{d}x \left|\left(\frac{d}{dx} f(x)\right) G(x)+ f(x) \left(\frac{d}{dx} G(x)\right) \right|,\\
	& \leq v_0  \sum_{l\in\zz}  \int_{a_l}^{a_{l+1}}  \textup{d}x \left|\frac{d}{dx} f(x)\right| \big| G(x) \big|+ \int_{a_l}^{a_{l+1}}  \textup{d}x \,\big|f(x)\big|   \left|\frac{d}{dx} G(x)\right| \\
	& \leq  v_0\!  \left(\sum_{ \substack{l\in\zz \\ l< 0}} G(a_{l+1})  \! \int_{a_l}^{a_{l+1}}  \textup{d}x \bigg|\frac{d}{dx} f(x)\bigg|\, +\sum_{ \substack{l\in\zz \\ l\geq 0}}G(a_{l})  \! \int_{a_l}^{a_{l+1}}  \textup{d}x \bigg|\frac{d}{dx} f(x)\bigg| + J \sum_{l\in\zz} \int_{a_l}^{a_{l+1}}\textup{d}x  \left|\frac{d}{dx} G(x)\right| \right),
\end{align}
	where we have denoted $J:=\left(\max_{x\in[0,T]} |f(x)| \right)$. Observing that the last term simplifies to 
\begin{align}
	& \sum_{l\in\zz} \int_{a_l}^{a_{l+1}}\textup{d}x  \left|\frac{d}{dx} G(x)\right|=  \sum_{l\in\zz} \left| \int_{a_l}^{a_{l+1}}\textup{d}x  \frac{d}{dx} G(x)\right| =  \sum_{l\in\zz} \big| G(a_{l+1})-G(a_l)\big| = 2 \sum_{l\in\zz} G(a_{l}),
\end{align}
we conclude 
\begin{align}
	|I| \leq  v_0\!  \left(\sum_{ \substack{l\in\zz \\ l< 0}} G(a_{l+1})  \! \int_{a_l}^{a_{l+1}}  \textup{d}x \bigg|\frac{d}{dx} f(x)\bigg|\, +\sum_{ \substack{l\in\zz \\ l\geq 0}}G(a_{l})  \! \int_{a_l}^{a_{l+1}}  \textup{d}x \bigg|\frac{d}{dx} f(x)\bigg| + 2 J \sum_{l\in\zz} G(a_{l}) \right).\label{eq:first part of I bound}
\end{align}
	We now further decompose the integral $\int_{a_l}^{a_{l+1}}  \textup{d}x \big|\frac{d}{dx} f(x)\big|$ into a summation over the interval where the derivative of $f$ does not change sign, namely let $b_{r(l)}, b_{r(l)+1}, b_{r(l)+2},\ldots, b_{r(l)+m(l)}$ be the sequence of the locations of the turning points of $f$ in interval $(a_l,a_{l+1})$ ordered in ascending order. We thus have: 
	\begin{align}\label{eq:second I bound part}
	&\int_{a_l}^{a_{l+1}}  \textup{d}x \bigg|\frac{d}{dx} f(x)\bigg|\\
	& = \int_{a_l}^{b_{r(l)}}  \textup{d}x \left|\frac{d}{dx} f(x)\right| + \int_{b_{r(l)}}^{b_{r(l)+1}}  \textup{d}x \left|\frac{d}{dx} f(x)\right| + \int_{b_{r(l)+1}}^{b_{r(l)+2}}  \textup{d}x \left|\frac{d}{dx} f(x)\right| +\ldots + \int_{b_{r(l)+m(l)}}^{a_{l+1}}  \textup{d}x \left|\frac{d}{dx} f(x)\right| \\
	& =  \left|\int_{a_l}^{b_{r(l)}}  \textup{d}x\, \frac{d}{dx} f(x)\right| + \left|\int_{b_{r(l)}}^{b_{r(l)+1}}  \textup{d}x\, \frac{d}{dx} f(x)\right| + \left|\int_{b_{r(l)+1}}^{b_{r(l)+2}}  \textup{d}x\, \frac{d}{dx} f(x)\right| +\ldots + \left|\int_{b_{r(l)+m(l)}}^{a_{l+1}}  \textup{d}x\, \frac{d}{dx} f(x)\right|\\
	& \leq 2 \Big( |f(b_{r(l)})| + |f(b_{r(l)+1})| + |f(b_{r(l)+2})|+\ldots + |f(b_{r(l)+m(l)})|\Big)\\
	& \leq 2 N_\textup{TP} J
	\end{align}
	for all $l\in\zz$. We can now plug \cref{eq:second I bound part} into \cref{eq:first part of I bound}, yielding:
	\begin{align}
	|I| &\leq  v_0\!  \left(\left(\sum_{ \substack{l\in\zz \\ l< 0}} G(a_{l+1})   +\sum_{ \substack{l\in\zz \\ l\geq 0}}G(a_{l})  \right) 2 N_\textup{TP} J + 2 J \sum_{l\in\zz} G(a_{l}) \right)\\
	&= 2J v_0  \left(2N_\textup{TP} G(0) + \big(2 N_\textup{TP}+1 \big) \sum_{l\in\zz}G(a_l) \right),
	\end{align}
	thus observing that $\sum_{l\in\zz}G(a_l) \leq N_\textup{Z} \left(\sum_{n\in\nno} G(n T) + G(-n T) \right)=N_\textup{Z} \left(G(0)+\sum_{n\in\zz} G\big(n T\big)\right)$ we conclude the proof.
\end{proof}


\end{document}